\documentclass[aps,prx,print,normalem]{revtex4-2}

\usepackage{dcolumn}   
\usepackage{bm}        

\usepackage{hyperref}
\usepackage{physics}
\usepackage{capt-of}
\usepackage{mathtools}

\usepackage{amsmath,amsthm,amssymb,amsfonts,enumerate,color}
\usepackage{graphicx,psfrag}
\usepackage{listings}
\usepackage{epstopdf}
\usepackage{float}
\usepackage{amsmath}

\usepackage{geometry}
\newgeometry{vmargin={20mm}, hmargin={16mm,16mm}} 

\usepackage{fancyhdr} 
\usepackage{extramarks} 
\usepackage{color} 
\usepackage{tabularx}

\usepackage{array}
\usepackage{accents}
\usepackage{geometry}
\usepackage{courier} 
\usepackage{listings}
\usepackage{float}
\usepackage{varwidth}
\usepackage{verbatim} 
\usepackage{comment} 
\usepackage[utf8]{inputenc}
\usepackage[english]{babel}
\usepackage{tikz}
\usetikzlibrary{arrows,arrows}
\usetikzlibrary{positioning}
\usepackage{centernot}
\usepackage{enumitem}
\usepackage{lipsum}

\tikzstyle{block} = [draw, fill=blue!20, rectangle, 
    minimum height=3em, minimum width=6em]
\tikzstyle{sum} = [draw, circle, node distance=1cm]
\tikzstyle{input} = [coordinate]
\tikzstyle{output} = [coordinate]
\tikzstyle{pinstyle} = [pin edge={to-,thin,black}]

\def\be{\begin{equation}}
\def\ee{\end{equation}}
\def\bea{\begin{eqnarray}}
\def\eea{\end{eqnarray}}

\def\v{{\bf v}}

\def\({\left(}
\def\){\right)}

\def\sfu{\mbox{\tiny{fus}}}
\def\sfi{\mbox{\tiny{fis}}}

\def\red{\textcolor{red}}
\def\blue{\textcolor{blue}}

\usepackage{marginnote}

\usepackage{dsfont}
\usepackage{bbold}
\usepackage{mathbbol}
\usepackage{newtxmath}

\usepackage{fancyhdr} 
\usepackage{extramarks} 
\usepackage{color} 
\usepackage{array}
\usepackage{accents}
\usepackage{ulem}
\usepackage{geometry}
\usepackage{courier} 
\usepackage{listings}
\usepackage{float}
\usepackage{varwidth}
\usepackage{verbatim} 
\usepackage{comment} 
\usepackage[utf8]{inputenc}
\usepackage[english]{babel}
\usepackage{tikz}
\usetikzlibrary{arrows,arrows}
\usetikzlibrary{positioning}
\usepackage[version=4]{mhchem}
\usepackage{tcolorbox}
\usepackage{enumitem}

\newcounter{glossarycount}

\newtcolorbox{glossarybox}[2][]{
  float,
  floatplacement=tbp, 
  #1,  
  code={\refstepcounter{glossarycount}},
  title={Glossary \theglossarycount: #2},
  width=\textwidth,
  boxrule=0.5pt,
  colback=white,
  colframe=gray!60!white,
  colbacktitle=gray!15!white,
  coltitle=black,
  fonttitle=\bfseries\large,
  halign title=center,
  sharp corners
}

\usepackage{xr}
\def\be{\begin{equation}}
\def\ee{\end{equation}}
\def\bea{\begin{eqnarray}}
\def\eea{\end{eqnarray}}

\def\v{{\bf v}}

\def\({\left(}
\def\){\right)}

\def\v{\textsubscript{v}}
\def\red{\textcolor{red}}
\def\blue{\textcolor{blue}}

\usepackage{amsthm}
\theoremstyle{plain}
\newtheorem{prop}{Proposition}
\newtheorem{corollary}{Corollary}

\theoremstyle{definition} 
\newtheorem{remark}{Remark}

\makeatletter\AtBeginDocument{\let\@elt\relax}\makeatother
\usepackage{graphicx}

\begin{document}

 \title{
 Building and maintaining a System of Intracellular Compartments
}
\author{Amit Kumar$^{1,2}$}
\author{Madan Rao$^{2}$}
\email{rao.madan@gmail.com}
\affiliation{
$^{1}$Raman Research Institute, C.V. Raman Avenue, Bengaluru 560080, India\\
$^{2}$Centre for Living Machines, National Centre for Biological Sciences (TIFR), GKVK Campus, Bellary Road, Bengaluru 560065, India}






\begin{abstract}
 Organelle patterning and its heritability remain central mysteries in cell biology, highlighting the fundamental tension between genetic inheritance and self-assembly. Here, we explore the nonequilibrium assembly and emdedded size control of the Golgi cisternae and endosomes, amid a continuous flux of membrane traffic, within a stochastic framework of mechanochemical fusion-fission cycles that violate detailed balance. Using a dynamical systems approach, we identify distinct, robust  regimes, ranging from fixed points to limit cycles with definite phase relations between cisternae. We identify these dynamical regimes with diverse phenotypes, from stable cisternae to periodic, cell-cycle-dependent dissolution/reassembly of cisternae to  cisternal progression.  We analyse its dynamic response to systematic perturbations or driving protocols and make definite predictions that may be tested experimentally. Our analysis reveals that the two competing models of Golgi organization—vesicular transport and cisternal progression—are, in fact, two phases of the same underlying nonequilibrium process. We see that cisternal size homeostasis is brought about by
a size-dependent embedded control
system driven by fusion-fission kernels. Finally, our framework offers a strategy for controlling cisternal number and chemical identity by modulating the interplay between glycosylation enzymes and membrane fission-fusion dynamics.

\end{abstract}

\maketitle

\section{Introduction} 

Understanding how subcellular structures of definite size, shape and chemical identity are assembled and maintained under nonequilibrium conditions remains a fundamental challenge in cell biology~\cite{Alberts2008,Misteli2001,Marshall2020a}. The assembly and homeostatic control of organelles, such as the Golgi complex and Endosomes, are of particular interest~\cite{chan}; these membrane-bound compartments must reliably maintain their identity despite a continuous vectorial flux of vesicles driven by active transport, fusion and fission.
Despite detailed  knowledge of the molecular players involved in their membrane remodeling, 
the key {\it physical principles} underlying organelle 
assembly in such {\it open} systems
remain unclear~\cite{Glick2009}.
The objective of this paper is to arrive at a quantitative description that clearly lays out the organising principles governing the assembly of the system of Golgi cisternae 
and make testable predictions. 
In doing so, we arrive at some surprising  insights into intracellular organisation.



Over the years, many  descriptive models 
 of Golgi organisation have been proposed, such as vesicular transport~\cite{rothemanves} and cisternal progression~\cite{losev,bonfanti,Matsuura-Tokita},
 often portrayed as being contesting~\cite{Sens2013,rotheman}. 
Interpretations of these models using computational and physics-based approaches to qualitatively and quantitatively confront with experiments, by and large do not explore the
space of steady states as a function of fusion-fission parameters or their
 dynamical stability,  nor do they study feedback control mechanisms that maintain the nonequilibrium steady state; see however~\cite{Sachdeva2016,Vagne2018,Sens2013}. Yet it is precisely these features — the structure of the steady state and its stability — that we believe are central to understanding Golgi organisation. In our view, the hallmark of Golgi nonequilibrium assembly is its dynamical steady state — a system of cisternae sustained by a finite vesicular flux — that remains robust to cellular noise.

 


Our analysis draws inspiration from recent discussions
on the nonequilibrium assembly and size control of the diversity of filamentous organelles in eukaryotic cells, such as 
flagella, cilia, filopodia, which are built from the  assembly-disassembly of microtubule or actin based structures~\cite{Mohapatra2016, banerjee, mukherjee, debashish, Marshall2020b}. 
The takeaway from these studies is that homeostasis of structures at a fixed size can be achieved only with size-dependent assembly/disassembly rates~\cite{Mohapatra2016}. 

 In this paper, we study  the stochastic time evolution of the  sizes  of the Golgi cisternae subject to nonequilibrium fusion and fission modelled as  
discrete Markov cycles.
By integrating over the cycle time of the ``fast'' fusion-fission processes, we derive dynamical mean field equations for the cisternal sizes. 
The dynamical system for multiple cisternae shows additional nonequilibrium steady states, including 
 limit cycles with definite phase relations between cisternae.
We identify such nonequilibrium steady states with distinct phenotypic outcomes such as vesicle transport, cisternal progression, or the periodic dissolution and reformation of cisternae across the cell cycle. 
We analyse its dynamic response to systematic perturbations or driving protocols and make definite predictions that may be tested experimentally. Our work applies to cisternal dynamics and their organisation in a variety of cellular contexts, both in  health and in disease.
Finally, we study the robustness of this system of cisternae 
to extrinsic (influx) and intrinsic (chemical) noise, using analytic approaches and Gillespie simulations. 
Taken as a whole,  our study argues that the Vesicular Transport and Cisternal Progression models and their variants are different phenotypic outcomes of the same underlying physical description that 
must occupy different locations in a nonequilibrium phase diagram. 
We see that the active mechanochemistry of the fusion-fission cycles naturally 
gives rise to a size-dependent {\it embedded feedforward control}  via fusion-fission kernels, that
successfully
maintains stable cisternae in the presence of a net vectorial flux of vesicles.
For mathematical derivations and analyses of the geometry of flows of dynamical systems, the reader may consult \ref{sec:si00}-\ref{sec:robust_si} of the 
{\it Supplementary Information} (SI). For convenience, we display the numerical values of parameters used in the study in Table\,\ref{tab:tabl1} and a compact Glossary\,\ref{box:definitions} of dynamical systems nomenclature.







\section{Stochastic description of the nonequilibrium assembly of cisternae}
\label{sec:physical_basis}


The buildup and maintenance of Golgi cisternae is a consequence of three major nonequilibrium processes, namely, the continual anterograde-retrograde flux of spherical ($50$\,nm diameter) and/or tubular transport vesicles injected from the endoplasmic reticulum (ER) towards the plasma membrane (PM) and the active processes of fusion and fission of these vesicles to and from cisternae (Fig.\,\ref{fig:golgi_flux}(a)) mediated by specific GTP binding enzymes and specialised proteins, such as Rab proteins, tethering proteins, v-SNARE and t-SNARE for fusion~\cite{Grosshans,kliesh} and Arf and COP proteins for fission~\cite{glick1,mal1,weiland}. 
We study the kinetics of nonequilibrium assembly at time scales longer than enzymatic cycle time scales \cite{arf1,allin} and ignore details of the various molecules involved in fusion and fission -- we will collectively refer to these enzymes and their associated proteins as {\it fusogens} and {\it fisogens}, respectively. Each fusion (fission)  event 
adds (removes) membrane area and lumenal volume to (from) the cisternae,  referred to as addition (removal) of cisternal ``mass'' or ``size''~\cite{Sachdeva2016,Sens2013}.




We first describe the kinetics of nonequilibrium assembly of a {\it single} Golgi cisterna of size $M$ (in units of typical transport vesicle size) in terms of a stochastic master equation that incorporates the physics of fusion and fission via 
 discrete Markov cycles (for the master equation for $n$-cisternae, see \ref{sec:si0} and \ref{sec:Mastereqtwocisternae}). These cycles,
defined by  internal states 
and inter-state transitions, Fig.\,\ref{fig:golgi_flux}(b,c), are as follows: 
Upon activation from an {\it external} state (labelled $\vert 0\rangle$), the local cisternal membrane configuration associated with fusogens ($+$ species) or fisogens ($-$ species) could transition into an ordered  sequence of internal states that are indexed by $i\in\{1,2,3,\dots, \aleph_\pm\}$, where $\aleph_\pm$ is the number of internal states,
before winding back to the external state $\vert 0\rangle$, terminating this cycle and priming for the next.

We make the following assumptions, supported empirically -- (i) 
Vesicles of unit size are {\it injected} from the ER at a specified rate, constituting a net vectorial current
from the ER (left) to the PM (right)~\cite{dmitrieff, wang2011golgi, altanbonnet2004molecular, Glick2009, lowe2011recent} (Fig.\,\ref{fig:golgi_flux}(a)). (ii) Vesicles injected from the ER fuse at a preassigned {\it nucleation} site with a given rate to initiate a growing cisterna~\cite{loweg}. 
(iii) Initiation of the 
    Fusion and Fission cycles are {\it independent} stochastic events, catalysed by distinct sets of enzymes. (iv)
    There is a finite cisternal pool of fusogens and fisogens, set by expression levels of the cell.
    (v)
    Internal state transitions in the fusion-fission cycle are {\it fast} (milliseconds-seconds~\cite{arf1,allin}) compared to the time scale over which cisternal size is updated (seconds-minutes~\cite{Matsuura-Tokita}). (vi)
    Fusion and Fission are {\it fast} and so we assume that the fusion or fission events are discrete and non-overlapping.
    In \ref{sec:si0}, we consider the case of multiple fusion-fission events in a time interval $\Delta t$ over which the cisternal size is updated.

With these assumptions, we write down a master equation describing the time evolution of the
probability distribution
$P(M, i, k, t \,\vert \,N\v,N_{+},N_{-})$
of a cisterna of size $M$ at specific internal states
 $i \in \{1, \ldots, \aleph_{+}\}$ and $k \in \{1, \ldots, \aleph_{-}\}$, 
 of the fusion and fission cycle, respectively, at time $t$,
conditioned on the availability of $N\v$ vesicles,  $N_{+}$ fusogens, and $N_{-}$ fisogens. 
For concreteness, we will take the fusion and fission Markov cycles to have  4-states, $\aleph_{+}=\aleph_{-}=4$
  (Fig.\,\ref{fig:golgi_flux}(b,c)). The change in cisternal size $M$ is actuated by transitions between internal states (Fig.\,\ref{fig:golgi_fluxSI}(e)), resulting in the dynamics of the conditional distribution above, abbreviated $P(M, i, k, t)$
  
  
\begin{figure*}[t!]
\centering
\includegraphics[width=\textwidth]{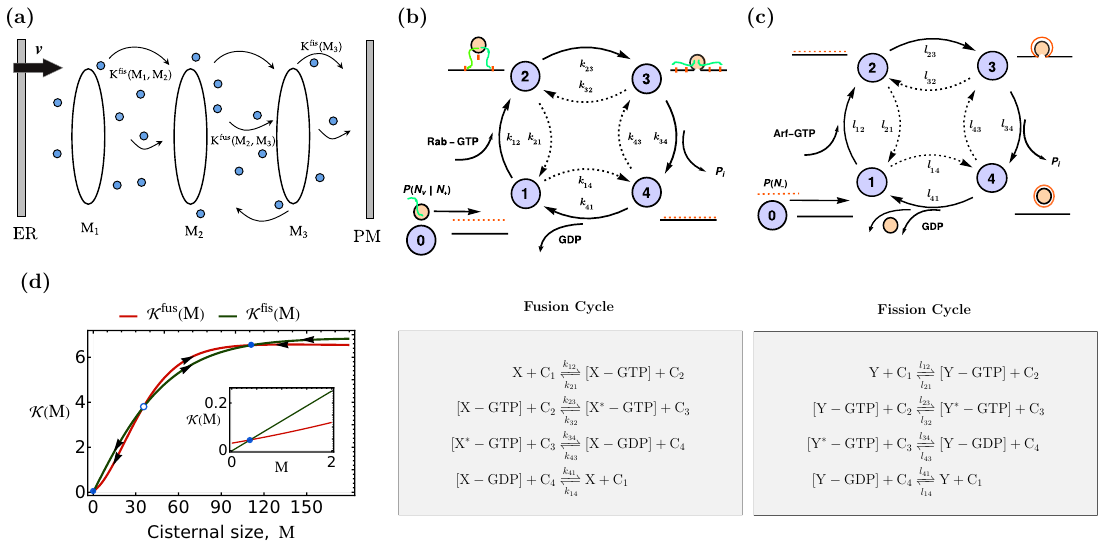}
\caption{{\bf Nonequilibrum assembly of Golgi cisternae via fusion-fission.} (a) Schematic showing Golgi cisternae of sizes $M_1, M_2, M_3$ maintained by the active fusion and fission of transport vesicles (blue filled circles) that form a continual net vectorial flux  from the endoplasmic reticulum (ER) to the plasma membrane (PM), with influx rate $v$ and  fusion-fission flux kernels $K^{\sfu}, K^{\sfi}$. 
 (b) Active fusion described by a nonequilibrium 4-state Markov cycle, is activated by Rab-GTPase, tethering proteins and other fusogens such as v-SNARE at the transport vesicle and its cognate t-SNARE proteins at the cisterna that mark the docking site~\cite{Alberts2008}. 
 The external
state $|0\rangle$ represents the fusion-competent transport vesicle+fusogen.  The internal states denoted by $|i\rangle$, $i=1,\ldots,4$, represent the {\it chemical} coordinate of the enzyme-fusogen-vesicle complex and the {\it configurational} coordinate of the cisternal membrane.
Thus, $|1\rangle$ $\eqqcolon$ (X, C$_1$), represents Rab-GDP-fusogen-vesicle
  and the undeformed cisternal membrane. 
Upon binding to GTP, $|1\rangle \to |2\rangle$ $\eqqcolon$ (X-GTP, C$_2$), which represents  Rab-GTP-fusogen-vesicle bound to the cisternal membrane
at the docking site.
This drives the formation of a membrane fusion intermediate, while undergoing a conformation change -  $|3\rangle$ $\eqqcolon$ (X$^*$-GTP, C$_3$) -  enroute to complete fusion  $|4\rangle$ $\eqqcolon$ (X$^*$-GDP, C$_4$) and the conversion to the inactivated Rab-GDP which then unbinds from the cisternal membrane, thus completing the fusion cycle and making the fusogens available for another round of fusion. The backward transition rates (dashed line-arrows) are slower, the energy consuming rates and the consequent breaking of detailed balance ensure that the enzymatic cycle predominantly proceeds in the forward direction.
On completion of this fusion cycle, the cisternal size increases by one unit. 
(c) Active fission described by a nonequilibrium 4-state Markov cycle, is activated by Arf-GTPase and other fisogens such as coat proteins COPII~\cite{Alberts2008}. The external state $|0\rangle$ 
 denotes the recruitment of fisogens (dashed red line) proximal to the cisternal membrane. 
  State    $|1\rangle$ $\eqqcolon$ (Y, C$_1$) represents Arf-GDP-fisogen and the undeformed cisternal membrane. Upon binding to GTP, $|1\rangle \to |2\rangle$
$\eqqcolon$ (Y-GTP, C$_2$), which represents Arf-GTP-fisogens bound to the budding cisternal membrane and marks the exit site. This drives the formation of a complete bud, while undergoing a conformation change -  $|3\rangle$ $\eqqcolon$ (Y$^*$-GTP, C$_3$) - enroute to complete fission
   $|4\rangle$ $\eqqcolon$ (Y-GDP, C$_4$), and the conversion to the inactive GDP-bound form. The Y-GDP then unbinds from the vesiculated bud, and returns to $|0\rangle$. At the end of this fission cycle, the cisternal size decreases by one unit. 
  (d)  An example of the profile of  fusion (red line) and fission (green line) flux kernels as a function of $M$, obtained from the microscopic model. By varying the microscopic rates, one gets three qualitatively distinct scenarios, corresponding to the topology of intersections of the fusion and fission kernels, which dictate the number of stable (filled circle) and unstable (open circle) fixed points (arrows denote the flows), see Fig.\,\ref{fig:golgi_flux_123}. (Inset) For de novo biogenesis the fusion flux rate should be larger than the fission rate at small $M$.}
\label{fig:golgi_flux}
\end{figure*}

\begin{equation}
{\dot P}(M, i, k, t) =   \sum_{j \in  \mathfrak{I}_{\hspace{-0.5pt}i}^{\hspace{-0.1pt}+} }  \; \left( - K_{ij}^{+} \; P(M, i, k, t) + K_{ji}^{+} \; P(M + \mathbb{M}^+_{ji}, j, k, t) \right)  +   \sum_{l \in \mathfrak{I}_{\hspace{-0.5pt}k}^{\hspace{-0.1pt}-} } \; \left(-K_{kl}^{-} \; P(M, i, k, t)  + K_{l k}^{-} \; P(M - \mathbb{M}^-_{lk}, i, l, t) \right) \,, 
\label{eq:master_all1}
\end{equation}
where the overdot represents the time derivative,  and  the index sets $\mathfrak{I}_{\hspace{-0.5pt}i}^{\hspace{-0.1pt}+}$ and $\mathfrak{I}_{\hspace{-0.5pt}k}^{\hspace{-0.1pt}-}$ contain all the internal states of the fusion-fission cycle which are directly connected to states $i$ and $k$ respectively; for example, the states $\vert 1\rangle$ and $\vert 0\rangle$ in   Fig.\,\ref{fig:golgi_flux}(b) have $\mathfrak{I}_{\hspace{-0.5pt}1}^{\hspace{-0.15pt}+} = \left\{0,2,4\right\}$ and $\mathfrak{I}_{\hspace{-0.5pt}0}^{\hspace{-0.15pt}+} = \left\{1\right\}$, respectively. $\mathbb{M}^+_{ji}$ (and $\mathbb{M}^-_{lk}$) denotes the amount of ``virtual''  mass exchanged in a transition from the state $j$ to $i$ (and $l$ to $k$) in the $\pm$-cycle (``virtual''  mass added being positive (negative) for forward (backward) transitions in the $+$ -cycle, and “virtual” mass added being negative (positive) for forward (backward) transitions in $-$ -cycle, see Fig.\,\ref{fig:golgi_fluxSI}(e)).
At the end of the cycle, exactly one unit of ``real'' mass of a single vesicle is transferred, thus, $\sum^{4}_{j=1}\sum_{i\hspace{0.5pt}\in\hspace{1pt}\mathfrak{I}_{\hspace{-1pt}j}^{\hspace{-0.15pt}+}}^{i>j} \mathbb{M}^+_{ji} = - \sum^{4}_{l=1}\sum_{k\hspace{0.5pt}\in\hspace{1pt}\mathfrak{I}_{\hspace{-1pt}k}^{\hspace{-0.15pt}-}}^{k>l} \mathbb{M}^-_{lk} = 1$. The fusion-fission rates     $K_{ij}^{\pm}$ depend on the availability of  vesicles $N\v$, fusogens $N_+$ and fisogens $N_-$ at the cisterna, and the mechanochemistry underlying enzyme-membrane interactions,
which we will see later, depend on $M$.


We now invoke the assumptions:
(i) the dynamics through the internal states is ``fast'', allowing us to integrate over  the chemical cycle time; (ii) the dynamics governing the availability of vesicles and the $\pm$-species equilibrates fast, allowing us to use their steady state distributions $P_{ss}(N\v,N_{+}), \, P_{ss}(N_{-})$, and (iii) 
the fusion and fission events are independent. With these, the dynamics of the marginal distribution $P(M,t)$ takes the form (see \ref{sec:si0} for derivation),  
\begin{equation}
\label{eqn:markov_eff}
\dot{P}(M,t) =  K^{\sfu}(M-1,\boldsymbol{\lambda})P(M-1,t) + K^{\sfi}(M+1, \boldsymbol{\mu})P(M+1,t) -(K^{\sfu}(M, \boldsymbol{ \lambda}) + K^{\sfi}(M, \boldsymbol{\mu})) P(M,t) \, ,
\end{equation}
where the {\it nonequilibrium fusion and fission flux kernels}  $K^{\sfu}(M,\boldsymbol{\lambda})$ and $K^{\sfi}(M, \boldsymbol{\mu})$ are functions of cisternal size $M$ and  parameter vectors $\boldsymbol{\lambda}$ and $\boldsymbol{\mu}$ describe the microscopic transition rates of the fusion and fission processes, respectively. The flux kernels and the steady-state distribution of vesicles and the $\pm$-species, are formally obtained from 
the steady state currents evaluated over the fusion and fission cycles (\ref{sec:si0}).
Our task now is to provide tractable analytical forms for these fusion-fission kernels, based on a physico-chemical model of the fusion-fission  cycles, Fig.\,\ref{fig:golgi_flux}(b).
\\

\noindent
{\it Physico-chemical basis for nonequilibrium flux kernels}:
We obtain
$P_{ss}(N\v,N_{+}), \, P_{ss}(N_{-})$ by assuming that the vesicle production, fusogen and fisogen availability are independent Poisson processes  (see \ref{subsec:markovcycle} for details). This helps us determine the probability of being in state $|1\rangle$, when either the vesicles are fusion competent (Fig.\,\ref{fig:golgi_flux}(b)), or when the cisterna is primed for fission (Fig.\,\ref{fig:golgi_flux}(c)).
In \ref{subsec:markovcycle} we model the internal state dynamics in the fusion-fission cycles as a set of enzymatic chemical reactions with cooperativity, that describe the microscopic dynamics of fusogens and fisogens on the membrane, resulting in effective transition rates $k_{ij}, \ell_{ij}$ between the states, as shown in Fig.\,\ref{fig:golgi_flux}(b,c). 
We find that these transition rates depend on cisternal size $M$ on account of -- (i)  limited
availability of fusogens/fisogens, tethering proteins etc. that mark the  docking/exit sites at the cisterna; and 
(ii)
dependence on composition (e.g., lipid specificity in a phase segregated domain) and mechanical properties (e.g., tension or curvature \cite{goud, sens1, sens,  dai, stone, bigay, has,Rautu2024}) of the cisternal membrane, the latter
 either directly via mechanochemistry or actuated via a chemical pathway~\cite{joseph}. This model for the internal state dynamics, allows us to compute the dependence of $K^{\sfu}(M, \boldsymbol{\lambda})$ and $K^{\sfi}(M, \boldsymbol{\mu})$ on $M$ (see Fig.\,\ref{fig:golgi_flux}(d) as example and \ref{subsec:markovcycle} for details). 
 \\
\\
\noindent
{\it Reduction to a low dimensional dynamical system}: 
Taking the first moment of Eq.\,\eqref{eqn:markov_eff}, and using a {\it mean field decoupling} (see \ref{sec:si2}), we find that the mean cisternal size, for which we continue to use the symbol $M$, satisfies 
\begin{equation}
\label{eq:markov_mean}
\dot{M}  =  \; \; {\cal K}^{\sfu}(M, \boldsymbol{\lambda}) - {\cal K}^{\sfi}(M, \boldsymbol{\mu}) \, ,
\end{equation}
a nonlinear dynamical system, where ${\cal K}^{\mbox{\tiny{fus/fis}}}$ related to $K^{\mbox{\tiny{fus/fis}}}$,  now has dimensions of mass current. 
While it might appear that there are many microscopic parameters that go to define the steady state currents (\ref{sec:si0}), we see that they can be combined into a probability of initiation of the rate limiting enzymatic reaction which depends on the enzyme availability and a net current associated with the enzyme catalysed reactions, leading to a dimensional reduction (see \ref{subsec:dimredux}), a consequence of the timescales separation.



 This dimensional reduction can also be understood 
 from the generic form of microscopic ${\cal K}^{\sfu}$ and ${\cal K}^{\sfi}$ -- their asymptotic behaviour and the nature of intersections when plotted against 
 $M$.  Owing to the limiting levels of the fusogens and fisogens, the availability of fusion-fission sites,
 and the dependence of the fusion-fission rates on membrane tension~\cite{dai, goud}, we find that the microscopic kernels increase with $M$ at first before saturating. 
Furthermore, varying the microscopic rates gives three qualitatively distinct scenarios associated with the topology of intersections of the fusion and fission kernels, which determines the number of stable and unstable fixed points, see Fig.\,\ref{fig:golgi_flux}(d) as example. This can be represented to arbitrary accuracy by the low dimensional forms (see \ref{sec:si2}, in particular Eqs.\,\eqref{eq:kernels_analyt_form2_Ckz},\eqref{eq:kernels_analyt_form2_C1kz}), 
\begin{eqnarray}
{\cal K}^{\sfu}  =  a +  \frac{b \,M^\alpha }{C_1 +  \, M^\alpha }  \,\,\,\,\,\,\,\,\,\,\,\,
{\cal K}^{\sfi}  = \frac{d\,M^\beta }{C_2 +  \, M^\beta }
\label{eq:Rfisfus}
\end{eqnarray}
Here, $b,d$ are the effective 
fusion (which includes the influx rate from the ER) and fission rates; $a$ (where $a/b \ll 1$) is the effective nucleation rate;
and $C_1, C_2$ are the Hill saturation constants.  
The parameters  $b, d, C_1, C_2$, are increasing functions of 
the availability of fusogens/fisogens~\cite{keener},
while the Hill-exponents $\alpha,\beta$, express the cooperativity of the fusogen and fisogens, respectively.
This Michaelis-Menten form of the fusion-fission kernels has indeed been used in previous studies~\cite{Sachdeva2016,Vagne2018}.
The dimensional reduction strategy outlined here 
when extended to the case of multiple cisternae, 
leads to a tractable mean field dynamical system for the sizes of individual cisternae (see \ref{sec:si2}).


 
\section{Nonequilibrium assembly of a single cisterna}
\label{sec:single_cist}
Scaling time in units of fission rate $d$, Eqs.\,\eqref{eq:markov_mean},\,\eqref{eq:Rfisfus} together can be rewritten as, 
\begin{eqnarray}
\dot{M} =   v \, \left(a_0 +   \frac{M^\alpha}{C_1 + M^\alpha} \right)- \left(\frac{ M^\beta}{C_2 + M^\beta}\right)\, .
\label{eq:veq_Hd}
\end{eqnarray}
 with dimensionless  parameters $v= b/d$ and $a_0 = a/b$. For this dynamical equation to describe stable de novo biogenesis, we should have $\alpha > \beta$ (see \ref{sec:si4}).
Our analysis of the nullcline (\ref{sec:si4})
reveals that regardless of the values of $\alpha$ and $\beta$ (as long as they are positive definite), the topology of intersections  (root structure)  is maintained and always
admits at most three fixed points.
Thus without loss of generality, we may fix 
$\alpha=2$ and $\beta=1$ (an argument for such cooperativity was also made in \cite{Misteli2001}).
Using parameter values extracted from experiments (Table\,\ref{tab:tabl1}), sets the unit of time $t$ to be $1$s.

\begin{glossarybox}[label=box:definitions,float=!h]{Dynamical systems and Bifurcations \cite{strogatz,GH,Izhikevich:2007}}
\hspace{0.25cm} A {\it continuous  dynamical system} can be represented by
\bea
\frac{d \bf{X}}{dt} = \mathbf{f}(\mathbf{X},\boldsymbol{\lambda},t)
\eea
where $\mathbf{X} \in \mathbb{R}^n$ is the state vector whose time ($t$) evolution is given by a smooth vector field  $\mathbf{f}: \mathbb{R}^n \rightarrow \mathbb{R}^n$ along with  $\boldsymbol{\lambda}$, the system parameters. 

\hspace{0.25cm} The \textit{state space} comprises all possible states $\mathbf{X}$. Within this space, we identify \textit{invariant sets}—such as \textit{fixed points} (FP) given by $\mathbf{f(\mathbf{X},\boldsymbol{\lambda},t)}=0$) and \textit{periodic orbits} (limit cycles, LC) — whose stability is defined by their \textit{basin of attraction}, the set of initial conditions asymptotically approaching them. An invariant set, e.g. an isolated fixed point,  is \textit{locally stable} if its basin covers a local neighborhood, \textit{globally stable} if it spans the entire state space, and \textit{unstable} if generic nearby trajectories diverge. Consequently, we categorize the system's behaviour in the parameter space $\boldsymbol{\lambda}$ into distinct phases: a \textit{stable phase} (globally stable attractor), a \textit{multistable phase} (multiple coexisting attractors), and a \textit{semistable phase} (where the state space is partitioned into regions that either converge to a bounded solution or diverge to infinity). We call these different qualitative behaviours that the system admits as different \textit{solution classes.}

\hspace{0.25cm} \textit{Bifurcations} describe qualitative changes in a system's behaviour, e.g., number or stability of equilibria/periodic orbits, as a parameter varies. Bifurcations can be local, driven by stability changes at a specific fixed point (detectable via linearization), or global, which reshape larger trajectories in phase space and cannot be detected by local analysis. We will be concerned with the following bifurcations in this work,

\begin{enumerate}[label=\arabic*.]
    \item \textit{Saddle-Node} (SN) \textit{bifurcation} is  a local bifurcation, that happens when two fixed points, that differ in stability along at least one direction, coalesce and vanish (or equivalently are created) as one varies the bifurcation parameter. It is accompanied by critical slowing down of system trajectory near the vanished fixed points, known as \textit{saddle-node ghost delay} (Fig.\,\ref{fig:bif_ex1}(a)) \cite{strogatz}. 
    
\item \textit{Hopf bifurcation} manifests when a fixed point changes stability as a pair of complex-conjugate eigenvalues cross the imaginary axis, and it gives rise to a periodic orbit (limit cycle) locally  (Fig.\,\ref{fig:bif_ex1}(b))  \cite{strogatz}.

\item \textit{Saddle-Node on an Invariant Circle} (SNIC) \textit{bifurcation} is a global bifurcation that involves a saddle-node bifurcation occurring directly on an invariant circle (Fig.\,\ref{fig:bif_ex1}(c)). At onset, the SNIC bifurcation produces a limit cycle of infinite period due to the saddle-node ghost; the period $\tau$ decreases continuously as the bifurcation parameter, $\mu$ moves away from criticality. Near the bifurcation point $\mu_0$, the period scales as: $\tau \sim 1/\sqrt{\mu - \mu_0}$, a dependence inherited from the saddle-node ghost \cite{strogatz}.


\item \textit{Saddle-Node of limit cycles} (SNLC) \textit{bifurcation} is also a  global bifurcation in which a stable limit cycle and an unstable limit cycle coalesce and annihilate as the bifurcation parameter varies, leaving the system to relax to a coexisting stable attractor ((Fig.\,\ref{fig:bif_ex1}(d)).

\end{enumerate}

  In \ref{sec:nullc2} and \ref{sec:denovo_tau}, we perform a detailed analysis of these bifurcations that demarcate the onset of the phases. The nature of these bifurcations determine the time scales of assembly of structures (Figs.\,\ref{fig:phases_1d}(d, f)\,and\,\ref{fig:fp_lc_sol}(d)).
 
.

  \vspace{0.1cm} 
  \centering
  \includegraphics[width=0.95\textwidth]{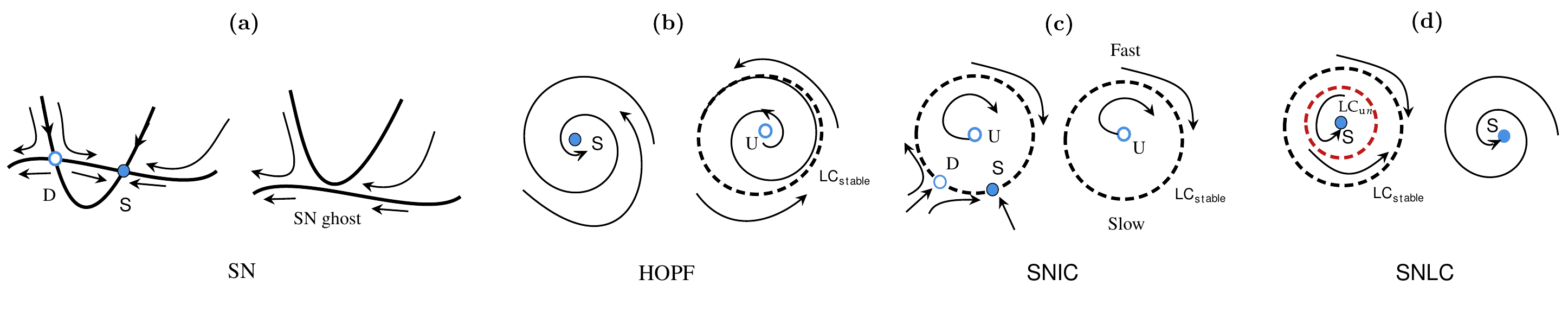}
  
      \captionof{figure}{{\bf Typology of bifurcations and flows.} (a-d) Change in flows as well as in fixed points (FPs) before (left) and after (right) the different bifurcations : $\mathbf{S}$ (stable FP), $\mathbf{D}$ (saddle FP), $\mathbf{U}$ (unstable FP), thick lines (nullclines), lines with arrowheads (flows), dashed lines (limit cycles), slow (SN ghost region), fast (rapid region).}
  \label{fig:bif_ex1}

\end{glossarybox}

We study the nature of intersections of the kernels by varying the effective parameters $\nu$,  $C_1, C_2$ and $a_0$. The fixed points of Eq.\,\eqref{eq:veq_Hd} with $\alpha=2$, $\beta=1$ are obtained from the cubic equation,
\begin{equation}
a_0 C_1 C_2 v + C_1 M (a_0 v -1) +C_2 M^2 (a_0+1)+ M^3 (a_0  + v - 1 ) = 0\, .
\label{eq:cubicpoly}
\end{equation}
By analysing the positive real roots of $M$ in Eq.\,\eqref{eq:cubicpoly} and their stability, we obtain the phase diagrams displayed in Fig.\,\ref{fig:phases_1d}(a,b) and  \ref{sec:si4}. These show a stable cisterna phase at intermediate values of the influx rate (keeping all other parameters fixed); this is consistent with the flux analysis done in~\cite{HimaniPRE}. The cisternal size increases with influx rate till it blows up when the fission flux rate goes beyond the fusion flux rate. 
The cubic polynomial loses (or gains) a pair of roots through a \textit{saddle-node} (SN) \textit{bifurcation},
leading to critical slowing down near the bifurcation point (see Glossary\,\ref{box:definitions}).
This bifurcation changes the number of stable fixed points (1 $\leftrightarrow$ 2), resulting in the creation or destruction of a stable equilibrium. In the case of two stable fixed points, the steady state cisternal size depends on initial conditions.
The robustness of the phase diagram Fig.\,\ref{fig:phases_1d}(a,b) 
 to changes in the rest of the parameters is ensured by the property that Eq.\,\eqref{eq:cubic} represents a {\it stable unfolding} (see~\ref{sec:si4} and~\cite{thom,golubitzky}).
Bounds on the robustness of the phase diagram to extrinsic and intrinsic noise are discussed in \ref{sec:robust_si}. 
 
\begin{figure}[t!]
\centering
\includegraphics[width=.95\textwidth]{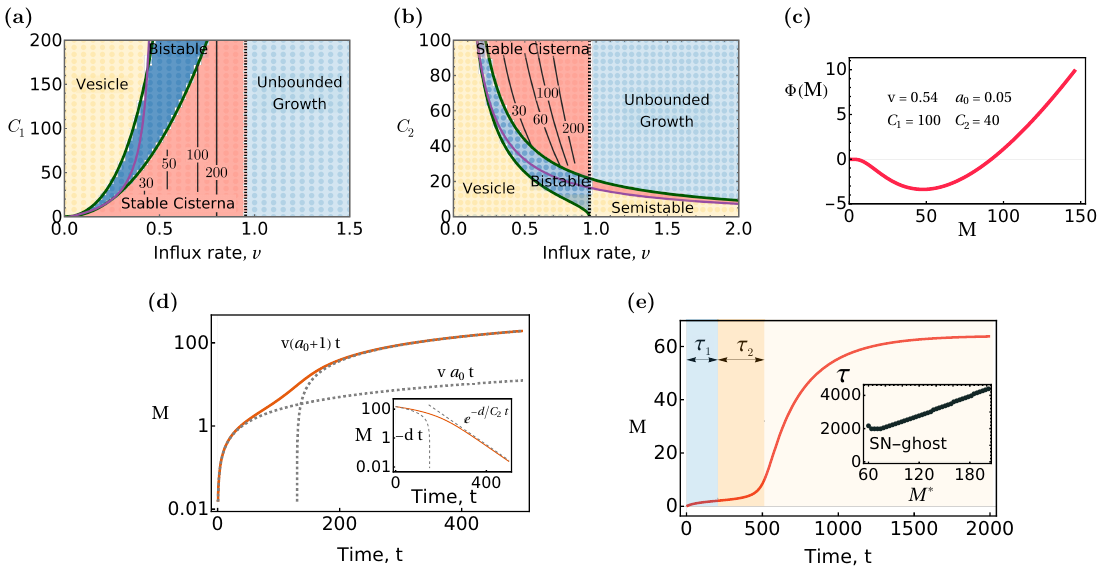}
\caption{{\bf Phase diagram for single cisterna.} 
(a) Phase diagram in the $C_1$ (fusogen levels) - $v$ (influx rate) plane shows two stable phases -- vesicle ($M \leq 1$) and cisterna ($M > 1$), together with their coexistence ({\it bistable}).
Cisternal size increases with influx rate $v$ (steady state sizes  are shown as black contour lines) till there are no fixed points and the cisterna grows unboundedly.  
Here $C_2=40$ and $a_0=0.05$. The phase boundaries are derived analytically from Eq.\,\eqref{eq:cubicpoly} -- the 
bistability region is obtained from the cubic discriminant $\Delta$ (thick green line), and 
the onset of cisterna formation (i.e. for a given $C_1$, there is a minimum value of influx rate $v$ required to form a stable cisterna, purple line), can be extracted from the 
coefficients of cubic polynomial Eq.\,\eqref{eq:cubicpoly} resulting from the transformation $M\rightarrow M+1$, Eq.\,\eqref{eq:transform_cf}. 
(b) Phase diagram upon varying $C_2$ (fisogen levels) and $v$ with $C_1 = 100$ and $a_0=0.05$.  The {\it semistable} phase corresponds to the situation where the system converges to a stable fixed point or grows unboundedly depending on the initial condition (Fig.\,\ref{fig:golgi_flux_123}(c)) and in most physical situations show up as a stable vesicle phase, ($M < 1$) for parameter values used here.
 (c) Stable fixed points of the single cisterna dynamics can also be interpreted as minima of a potential $\Phi(M)$. (d) If fission is stopped, the cisterna swells linearly at a rate set by the peak influx rate, $v(a_0+1) \approx v$. (inset) On the other hand, if  influx from the ER is stopped, the cisterna initially decays as $M(t)\sim M(0) - d \, t$, 
 and  later as $M(t)\sim M(0)\exp(- d/C_2 \, t)$ (parameter values $v=0.5$, $C_1=100$, $C_2=40$, $a_0=0.05$, $d=1$).
 (e) Dynamics of de novo cisternal formation shows a slow initial growth before it  rapidly increases to the steady state value $M^*$, plotted here for $v=0.6,\,C_1=100,\,C_2=40,\,a_0 = 0.05$. The slowing down is due to the presence of {\it saddle-node ghost}. Inset shows the non-monotonic dependence of the de novo Golgi formation time, $\tau=\tau_1+\tau_2$, at lower values of steady state cisternal size $M^*$~\cite{langhans,pypaert}. Glossary\,\ref{box:definitions} explains the nomenclature from dynamical systems theory.}
\label{fig:phases_1d}
\end{figure}

Equation\,\eqref{eq:veq_Hd} can be interpreted as the overdamped dynamics of a particle in a potential $\Phi(M)$ (Fig.\,\ref{fig:phases_1d}(c)). The minimum of $\Phi(M)$ locates the fixed point value of the cisternal size $M^*$, allowing us to address four relevant questions regarding the fate of a single cisterna -- {\it (i)  What sets the cisternal size $M^*$ at steady state?} This is primarily set by $C_1, C_2$, the Hill saturation constants for fusion and fission flux kernels, which in turn is set by the levels of fusogens and fisogens (Fig.\,\ref{fig:1cist_Master_sim}(e,g)) through the mean enzymatic rates  $k_z$ (see~\ref{sec:si2}). Further, the steady state size will 
increase with the fusion rate, $v$ (Fig.\,\ref{fig:phases_1d}(a)), and decrease with the fission rate $d$.
{\it (ii) What happens when 
cisternal fission is stopped?} This is commonly achieved by treatment with the drug Brefeldin-A~\cite{lippincott2000,helms1992,niu2005}. We see that this leads to cisternal swelling that grows linearly at a rate set by the influx rate $v$ (Fig.\,\ref{fig:phases_1d}(d)).
{\it (iii) What happens when the influx from the ER is stopped?} This clearly leads to
cisternal degradation, with a size that decays linearly at first followed by exponentially fast over a time scale  set by the fission rate $d$ and the levels of the fisogen $C_2$ (Fig.\,\ref{fig:phases_1d}(d,inset)). 
{\it (iv) What sets the time scale for de novo cisternal formation?} 
The time taken for de novo cisternal formation $\tau$ (Fig.\,\ref{fig:phases_1d}(e))is determined by the sum of a slow lag time ($\tau_1$), a consequence of the initial flatness of the potential $\Phi(M)$ (Fig.\,\ref{fig:phases_1d}(c)) and the faster saturation time ($\tau_2$), the characteristic time scale at the final steady state $M^*$ (given by the eigenvalue of the Jacobian of the dynamical system at $M^*$). Although the time scale can be computed from a straightforward numerical integration of Eq.\,\eqref{eq:veq_Hd}, one may obtain an approximate analytical estimate from dynamical systems theory~\cite{strogatz}. The slowing down of the trajectory is attributed to the presence of the saddle-node bifurcation that influences dynamical trajectories in the vicinity of the bifurcation point (known as a \textit{saddle-node ghost}~\cite{strogatz}, details in~\ref{sec:1cist_tau}). Using simplified forms  for the fusion and unsaturated fission kernels (see \ref{sec:1cist_tau}), we obtain estimates for steady state cisternal size $M^* = (1+a_0) C_2 \, v/(d-a_0-a_0\, v)$ and $\tau = \tau_1(\delta) + \tau_2(M^*) \ \sim (C_1 + C_2)/(d-a_0-a_0\, v) - C_2 \, d/(d-a_0-a_0\, v)^2 \log (1-\delta)$, 
in the vicinity $\delta= (d-a_0\,v)C_1/(a_0\, C_2\, v)$ of  the bifurcation point.
This provides a numerical estimate for the de novo cisterna formation time to be $5-10$\,mins (Fig.\,\ref{fig:phases_1d}(e)) consistent with experiments on the recovery time of the cisterna upon Brefeldin-A washout\,\cite{langhans,itocis}.

\section{Nonequilibrium assembly of multiple cisternae}
\label{sec:two_cisternae}

In extending our analysis to multiple cisternae (Fig.\,\ref{fig:golgi_flux}(a)), we immediately realise that the dynamical equations describing cisternal size updates have to be {\it non-local in time}.
This is because the fusion flux into cisterna $i$ from cisternae $i-1$ and $i+1$, which goes to update the size of cisterna $i$ at time $t$, depends on the fission dynamics of the donor
cisternae at an earlier time. To circumvent this difficulty, Refs.\,\cite{Sachdeva2016,Vagne2018,mani} consider the local dynamics of all membrane bound entities of sizes $M \geq 1$ and numerically compute the size distribution at steady state, identifying the distribution at large $M$ with cisternae. The price one pays is that the dynamical system is high (infinite) dimensional and has to be analysed numerically by statistically sampling over many 
vesicle
trafficking networks (or initial conditions, as in~\cite{Sachdeva2016}) with specified fusion and fission rules. 
Here, we will confront the nonlocal dynamics of the $N$-cisternae system head-on, and attempt a simplification which will reduce to an $N$-dimensional autonomous dynamical system. We will see that the nonlocality in time, together with the constraint that there is a fixed pool of cisternal specific fusogens and fisogens,
result
in inter-cisternal correlations in the fission and fusion flux kernels, leading to complex flows in the $N$-dimensional phase space. Details are presented in~\ref{sec:2cist_T}.

We first recall that the intercisternal flux depends on the availability of cognate pairs of v-SNAREs and t-SNAREs at the donor and the target cisterna, respectively~\cite{Munro2004, Bonifacino2003, Robinson2004, Traub2009, Yu2010, Jahn2006, Wickner2008, DSouzaSchorey2006, Stenmark2009, itocis, mani}.
Consider the fission flux $J^{fis}_{i,i+1}$  from the donor cisterna $i$ to target cisternae $i+1$ at time $t$. The appropriate v-SNAREs need to be transported on vesicles  destined to the target cisterna, leading to a depletion of the SNARE-pool and a reduced probability for subsequent fission events at $i$, unless replenished by fusion events at $i$ from $i+1$~\cite{willett}. This implies that the fission flux $J^{fis}_{i,i+1}$ at cisterna $i$ at time $t$ depends on the fusion events at cisterna $i+1$ {\it at earlier times within a transit time window}, which in turn depend on the size $M_{i+1}$ (following the discussion in Sect.\,\ref{sec:physical_basis}). This contributes as a nonlocal kernel in time and cisternal index in the integro-differential equation for $M_i$ (as in, for  example, in Eq.\,\eqref{eq:yx} below),
\bea
\mathcal{J}^{fis}_{i,i+1} (M_i,M_{i+1},t) \equiv \int d\,t_w \, P(t_w) \, J^{fis}_{i,i+1} (M_i,t; M_{i+1},t-t_w)
\label{eq:integral}
\eea
where the intercisternal time window~\cite{storrie,jackson,dmitrieff} is drawn from a probability distribution $P(t_w)$ which we take to have compact support (or an exponential decay).
Similarly, the fusion flux $J^{fus}_{i+1,i}$ at the target cisterna $i+1$ at time $t$ depends on the fission events at the donor cisterna $i$ at earlier times within an intercisternal time window, which enters as a nonlocal kernel in time and cisternal index in the integro-differential equation for $M_{i+1}$.
We may now use a mean field decoupling approximation~(\ref{sec:2cist_T}), 
\bea
\mathcal{J}^{fis}_{i,i+1} (M_i,M_{i+1},t) \propto T \, {\cal K}^{\sfi} ( M_i(t)) \, {\cal K}^{\sfu} (M_{i+1}(t))
\label{eq:algebraic}
\eea
that reduces the integro-differential equation to an ODE, where $T$ is the support of $P(t_w)$.
This implicitly assumes that $T$ is small and that the cisternal sizes do not vary much over this time scale.

This allows us to prove the following proposition and corollary (proof in~\ref{sec:prop_reduced}); for simplicity, we restrict ourselves to the 2-cisternae case, but the results hold more generally:





\begin{prop}\label{prop:prop}
Consider the general flux system,
\begin{eqnarray}
\label{eq:yx}
 \dot{M_1} &=& J_{in}(M_1) - \mathcal{J}^{fis}_{12} (M_1,M_2) + \, \mathcal{J}^{fus}_{21} (M_1,M_2) \\
 \dot{M_2} &=&  \,  \mathcal{J}^{fus}_{12} (M_1,M_2) - \mathcal{J}^{fis}_{21} (M_1,M_2) - J_{ex}(M_2) 
  \label{eq:xy}
\end{eqnarray}
where $J_{in}$ and $J_{ex}$ are the influx and the exit flux, respectively, and the rest are intercisternal fluxes, written in general in its integral form Eq.\,\eqref{eq:integral}, all of which are positive for $M_{1,2} \geq 0$. We consider the following two cases, 
\begin{enumerate}[label=(\roman*)]
\item 
If all the intercisternal fluxes are local, i.e.,
$\mathcal{J}^{fis}_{12}(M_1,M_2) = \mathcal{J}^{fis}_{12}(M_1)$, $\mathcal{J}^{fus}_{21}(M_1,M_2) = \mathcal{J}^{fus}_{21}(M_1)$, $\mathcal{J}^{fus}_{12}(M_1,M_2) = \mathcal{J}^{fus}_{12}(M_2)$ and $\mathcal{J}^{fis}_{21}(M_1,M_2) = \mathcal{J}^{fis}_{21}(M_2)$, then the above flux system can only have real eigenvalues. 
\item Let the intercisternal fluxes
depend on the size of the donor cisterna alone, i.e.,
$\mathcal{J}^{fis}_{12}(M_1,M_2) = \mathcal{J}^{fis}_{12}(M_1)$, $\mathcal{J}^{fus}_{21}(M_1,M_2) = \mathcal{J}^{fus}_{21}(M_2)$, $\mathcal{J}^{fus}_{12}(M_1,M_2) = \mathcal{J}^{fus}_{12}(M_1)$ and $\mathcal{J}^{fis}_{21}(M_1,M_2) = \mathcal{J}^{fis}_{21}(M_2)$. If $ \mathcal{J}^{fus}_{12}(M_1),  \,\mathcal{J}^{fus}_{21}(M_2)$ are co-monotonic as functions of their arguments, then the above flux system can only have real eigenvalues.
\end{enumerate}
\end{prop}

\begin{corollary}
It follows that the necessary condition for the above flux system to  have complex eigenvalues, and hence limit cycles, is that either or both of the intercisternal flux kernels $\mathcal{J}^{fus}$ and $\mathcal{J}^{fis}$  have to be functions of both $M_1$ and $M_2$, and there is 
no restriction on the sign of  $\partial_{M_2} \left(\mathcal{J}^{fus}_{21}-\mathcal{J}^{fis}_{12}\right)$ and $\partial_{M_1} \left(\mathcal{J}^{fus}_{12}-\mathcal{J}^{fis}_{21}\right)$, i.e., the corresponding Jacobian should have off-diagonal elements with opposite sign.
\end{corollary}

This unrestricted sign of the off-diagonal elements is supported by the finding that oversaturation of the trans-Golgi network with anterograde cargo reduces the efficiency and kinetics of retrograde transport. This, in turn, depletes early cisternae of the fusion and fission machinery essential for anterograde transport from the cis and medial Golgi~\cite{hirata2015,griffiths} — a bottleneck that arises because intercisternal flux depends critically on the availability of cognate SNARE pairs~\cite{willett,Blackburn2019}.


We now propose an explicit realisation of intercisternal fluxes of the 2-cisternae system Eqs.\,\eqref{eq:yx},\eqref{eq:xy}, that must satisfy the following reasonable criteria: (i) De novo biogenesis, implying  presence of a nucleation seed for each cisterna. 
(ii) Mass action structure, implying
${\cal J}^{fis}_{12}\rightarrow 0$ as $M_1 \rightarrow 0$ and $ {\cal J}^{fis}_{21} \rightarrow 0$ as $M_2 \rightarrow 0$. 
(iii) 
For small influx, the system stays in a stable vesicle phase, implying that near $M_1,M_2 \sim 0$, the intersection of fission and fusion kernels should give rise to a stable fixed point $(M_1^*,M_2^*)$. Therefore, in the immediate neighborhood, for $(M_1,M_2) > (M_1^*,M_2^*)$, the fission rate must dominate the fusion rate and vice versa (see Fig.\,\ref{fig:2cist_roots}(a) and Fig.\,\ref{fig:golgi_flux}(d), inset). (iv) In addition, we impose two physical conditions -- for $M_2 \rightarrow 0$, system should reduce to $1$-cisterna dynamics and for $M_1\rightarrow 0$, $M_2$ should also vanish. The second condition ensures that the mass must accumulate in the first cisterna before the second cisterna starts growing and along with the first condition endows the two-cisternae system with a hierarchical structure, where the dynamics of the second cisterna is conditioned on the fixed point structure of the first cisterna. Any system of equations with these properties will give rise to a qualitative similar phase diagram.

With these considerations, together with the mean field decoupling Eq.\,\eqref{eq:algebraic}, we study a specific realisation of Eqs.\eqref{eq:yx},\eqref{eq:xy}, 
\begin{eqnarray}
\label{eq:retrograde1}
\dot{M}_1 &=& \underbrace{v \left(a_1 + \frac{ M_1^2}{C_{11} + M_1^2}\right)}_{\text{J}_{\text{in}}} -\underbrace{\frac{d_1 \,M_1}{C_{12}  + M_1}}_{\text{J}_{\text{leak}}}  -  \underbrace{\frac{d_{12} \, M_1}{C_{12}  + M_1}\left(a_2 + \frac{ M_2^2}{C_{21}+M_2^2}\right)}_{\mathcal{J}_{12}^{\text{fis}}} + \underbrace{\frac{d_{21}\,M_2}{C_{22} + M_2}\left(a_1 + \frac{ M_1^2}{C_{11} + M_1^2}\right)}_{\mathcal{J}_{21}^{\text{fus}}}   \\ 
\dot{M}_2 &=&   \underbrace{\frac{d_{12} \,M_1}{C_{12}  + M_1}\left(a_2 + \frac{M_2^2}{C_{21}+M_2^2}\right)}_{\mathcal{J}_{12}^{\text{fus}}}- \underbrace{\frac{ d_{21}  \, M_2}{C_{22} + M_2}\left(a_1 + \frac{ M_1^2}{C_{11} + M_1^2}\right)}_{\mathcal{J}_{21}^{\text{fis}}}   - \underbrace{\frac{ d_{2} \, M_2}{C_{22} + M_2}}_{\text{J}_{\text{ex}}}
\label{eq:retrograde2}
\end{eqnarray}
where we have included a leakage flux $\text{J}_{\text{leak}}$ from the first cisterna~\cite{dmitrieff}. As before, the Hill-exponents for fusion and fission have been chosen to be $\alpha=2$ and $\beta=1$, respectively.
The intercisternal fluxes $\mathcal{J}_{12}^{\text{fus/fis}}$ and $\mathcal{J}_{21}^{\text{fus/fis}}$ are parameterised by $d_{12}, d_{21}$ and, we have scaled time w.r.t. the peak leakage rate $d_1$ at cisterna $1$. 
The above dynamical system can be extended to multiple cisternae $i=1, 2, \ldots N$ by defining an \textit{intercisternal flux matrix} ${\cal J}$ whose $i^{th}$ row 
has terms corresponding to fission, fusion, leakage and inter-cisternal transfer $i-1 \rightleftarrows  i  \rightleftarrows i+1$ at cisterna $i$ (except for $i=1$ and $N$) (\ref{sec:2cist_T}).

\begin{figure}[t!]
\centering
\includegraphics[width=\textwidth]{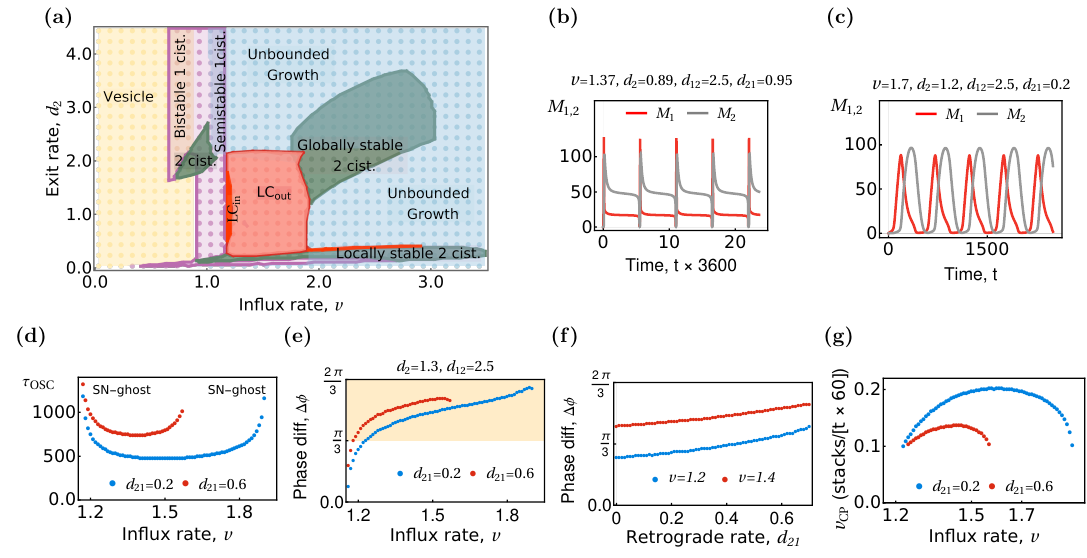}
\caption{{\bf Phase diagram for the two cisternae system.} (a) Numerical phase diagram for $2$-cisternae system with stable fixed point (FP), Limit cycle (LC) and unbounded (Growing) states, in ($v,d_{2}$) plane. Phase diagrams in ($v,d_{12}$) and ($v,d_{21}$) planes are shown in Fig.\,\ref{fig:overlap}.
This shows regions of multistability, similar to Fig.\,\ref{fig:phases_1d}(b),  where the configuration converges to one phase or the other depending on the initial condition (see Fig.\,\ref{fig:overlap}). Note there are three branches for the 2-cisterna phase -- $M_1\gg M_2 > M_{th}$ (threshold size) at low $v$ and high $d_2$ (locally stable, light green), $M_2\gg M_1 > M_{th}$ at low $d_2$ and high $v$ (locally stable, light green)  and $M_1\sim  M_2 > M_{th}$ at $v \sim d_2$ (globally stable, dark green).
Phase boundaries or bifurcation lines, can be obtained analytically from the behaviour of the eigenvalues of the Jacobian, and appear as saddle-node, SN bifurcations and Hopf bifurcations~(see Glossary~\ref{box:definitions}), as shown in Fig.\,\ref{fig:overlap}.
Limit cycle solutions are categorized based on the phase difference between times series of $M_1$ and $M_2$ -- in-phase limit cycle (LC$_{in}: 0-60$ deg., shown in dark red), and out-of-phase limit cycle (LC$_{out}: 60-180$ deg., shown in light red). Rest of the parameter values are $C_{11} = C_{21} = 100$, $C_{22} = 20$, $C_{12} = 20$, $d_{12} = 2.5$, $d_{21}=0.2$.
(b) In-phase oscillatory solutions represent spontaneous periodic dissolution and reformation of the two cisternae. (c) Out-of-phase oscillatory solutions 
can be associated with cisternal progression (CP) models and its variants based on the phase difference. 
(d) Time period of the limit cycle solution, $\tau_{\textrm{OSC}}$,  in units of $d_1=1s^{-1}$, varies with influx $v$ and is determined by the saddle-node delay due 
to the closest SNIC bifurcation. 
There are two saddle-node ghost delays for the oscillatory solutions in the system, one at small $v$, when limit cycle solutions appear via SNIC bifurcation and the other at large $v$, when limit cycle solutions disappear either via SNIC bifurcation or SNLC bifurcation (see Fig.\,\ref{fig:overlap}(a,b,c)). (e) Variation of phase difference $\Delta \phi$ between times series of $M_1$ and $M_2$ with influx rate $v$. At lower influx $v$, the saddle-node ghost appears near $M_1,M_2 \sim 0$, i.e. when the cisternae sizes are small (see bifurcation diagram Fig.\,\ref{fig:2cist_roots}(h) and Fig.\,\ref{fig:SNIC_2cist}(d)). Near this saddle-node ghost, the trajectories are slow and the phase lag, $\Delta \phi$ between $M_1,M_2$ does not change much -- we therefore obtain in-phase oscillations. However, as $v$ increases, the trajectories come out of the ghost region and move fast -- the phase lag between $M_1,M_2$ increases, which is further increased by the saddle-node ghost appearing at larger cisternal sizes, leading to out-of-phase oscillations (shaded region). The {\it susceptibility} $\partial_v(\Delta \phi)$  peaks at each ghost (criticality) and is suppressed in the fast-transit region between them.
(f) The phase difference increases with the retrograde rate $d_{21}$, supporting observations that cisternal progression is accompanied by retrograde transport to recycle resident Golgi enzymes against the direction of cisternal flow~\cite{glick,Glick1998,losev,Matsuura-Tokita}.
(g) Phase velocity, computed from the ratio of $\Delta \phi$ and time period $\tau_{\textrm{OSC}}$. Our estimate of $\textit{v}_{CP}\sim 0.1-0.2$ stacks/min agrees with the cisternal progression speed obtained from analysis of experiments~\cite{dmitrieff}. For (d-g), parameter values are $d_2=1.3,\,d_{12} =2.5$.} 
\label{fig:fp_lc_sol}
\end{figure}

\subsection{Results for the 2-cisternae system}
As before, the steady state phase diagram of the $N$-cisternae assembly can be obtained from an analysis of the root structure of the nullclines. 
For a general high-dimensional system, this  is quite involved \cite{fulton};  we will then rely on graphical and numerical analyses of the intersections of the nullclines and the use of bifurcation theory~\cite{strogatz,GH,kuznetsov} to categorize the solutions (\ref{sec:nullc2}). 
In simple cases, such as for 2-cisternae, the nature and stability of the fixed points can be studied analytically using a two-dimensional version of the {\it Routh-Hurwitz criterion}~\cite{hairer}, i.e. analysing the  characteristic polynomial of the Jacobian matrix $J$ of the dynamical system evaluated at fixed points $(M_1^*,M_2^*)$~(\ref{sec:RH}).
Whether a stable solution (a fixed point or a limit cycle) is accessible from  given initial conditions (in the case of multiple locally stable solutions), requires numerical analysis of the asymptotic solution of the dynamical system (\textit{asymptotic time domain analysis}, \ref{sec:RH}).
In this way, we demarcate the various phases of the 2-cisternae system \eqref{eq:retrograde1},\eqref{eq:retrograde2} in $v,\, d_2,\, d_{12}, d_{21}$ space (Fig.\,\ref{fig:fp_lc_sol}(a) and Fig.\,\ref{fig:overlap}). 
Stable fixed point solutions with two finite size cisternae ($M_{1,2}>M_{th}$, a threshold size), which we call the 2-cisternae phase, is accompanied by an in-flux, out-flux and intercisternal flux of cargo vesicles -- we  identify this phase with {\it vesicular transport} (VT). In addition, the dynamical system \eqref{eq:retrograde1},\eqref{eq:retrograde2}, supports limit-cycle solutions with definite phase relations between $M_1$ and $M_2$. The in-phase oscillatory solutions correspond to the spontaneous periodic dissolution and reformation of the two cisternae. The out-of-phase oscillatory solutions correspond to a kind of travelling wave \footnote{Coupling this to space in the manner done in~\cite{dmitrieff}, this out-of-phase regime corresponds to traveling wave solutions.}, which we identify with {\it cisternal progression} (CP), see supplementary movies FPsolution, OutphaseLC and InphaseLC.  The phase velocity (stack speed) is computed from the ratio of the phase difference $\Delta \phi$ and time period $\tau_{\textrm{OSC}}$, Fig.\,\ref{fig:fp_lc_sol}(f), and its numerical value is found to be comparable to the value of the cisternal drift estimated in~\cite{dmitrieff}. Oscillatory solutions with  intermediate phase difference can be identified with variants of the VT and CP models~\cite{glick, Glick1998, Papanikou2014, Glick2009, Pantazopoulou2019, Rothman2010, rothemanves, rotheman, vvquin, losev, bonfanti, Matsuura-Tokita, pfeffer}.
The phase boundaries can be obtained algebraically through a bifurcation analysis (\ref{sec:RH}) and are shown in Fig.\,\ref{fig:overlap}(a,b,c). We note from Fig.\,\ref{fig:fp_lc_sol}(e,f) that the out-of-phase oscillations (CP phase) are separated from the in-phase oscillation as well from the stable cisternal phase by the presence of saddle-node ghosts, leading to a sharp rise of the {\it susceptibility} $\partial_v(\Delta \phi)$ upon approaching each ghost at the phase boundaries. Hence, unless the system is fine tuned to operate near these boundaries,
the limit cycle solutions generically remain in the CP phase.

Returning to Fig.\,\ref{fig:fp_lc_sol}(a), we note that at large exit rate $d_{2}$, the dynamical system has the same fixed point structure as the single cisterna case (see Fig.\,\ref{fig:phases_1d}(b), and consistent with the analysis in~\cite{HimaniPRE}),
going from vesicle $\rightarrow$ bistable $\rightarrow$  stable single cisterna $\rightarrow$  semistable $\rightarrow$  unbounded growth, as the influx rate $v$ increases. As $v$ is increased further, cisterna-1 keeps on growing unboundedly, until enough mass accumulates in the second cisterna, increasing the intercisternal flux from cisterna 1, leading to a stabilisation of the two-cisternae system (via a Hopf bifurcation, see  Fig.\,\ref{fig:overlap}(a)). This follows from Eq.\,\eqref{eq:retrograde1}, as the outflux $\mathcal{J}_{12}^{fis}$ from cisterna $1$ increases with $M_2$. 
For smaller exit rate $d_{2}$, however, both cisternae are immediately formed -- on increasing $v$, the vesicle phase transforms to a $2$-cisternae  phase (locally stable, shown in light green or locally unstable, shown in light violet). 
This can transform to a stable $2$-cisternae phase by going  through a 
a limit cycle (LC) 
via a SNIC bifurcation (see Fig.\,\ref{fig:2cist_roots}(e,f)). 
Subsequently, limit cycle solutions disappear either via {\it saddle-node of limit cycles} (SNLC) or via another SNIC bifurcation (Fig.\,\ref{fig:overlap}(a,b,c)). See Glossary \ref{box:definitions} for dynamical systems nomenclature. 
The nullclines, bifurcations and the phase diagrams for the 2-cisternae system are presented in detail in \ref{sec:nullc2} and \ref{sec:RH}.

\subsection{Response to systematic external and internal variations}
\label{subsect:response}
\begin{figure}[t!]
\centering
\includegraphics[width=0.95\textwidth]{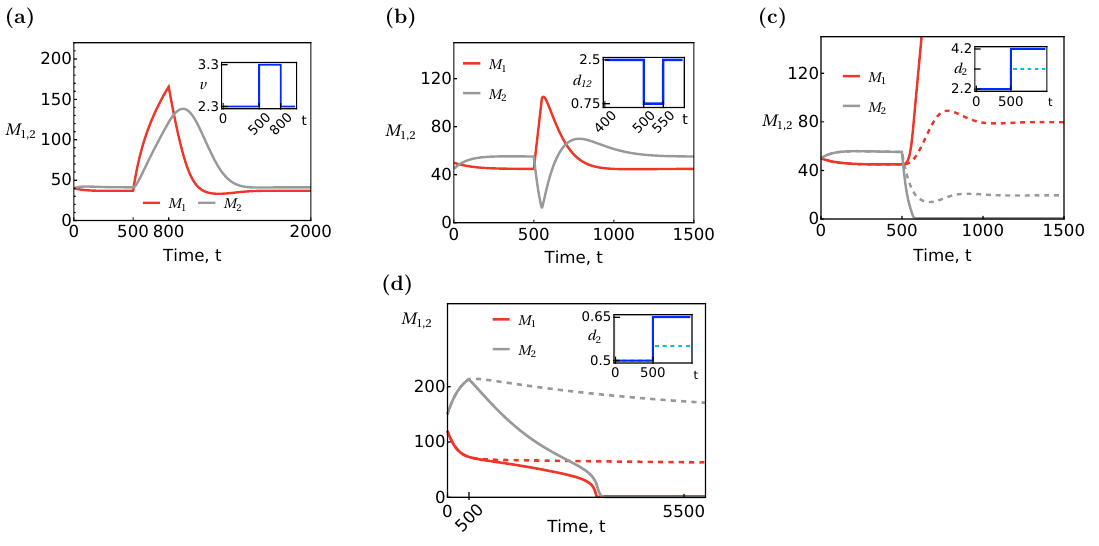}
\caption{{\bf Response to systematic perturbations.}  (a-d) Response of the stable cisternal phase,  
to a sudden step or rectangular pulse, as described in the text.
(a) Dynamical response of the steady state to a sudden positive pulse in influx rate $v$ (inset). The sizes of both the cisternae grow linearly at first and decay after the pulse. Note that initially the rate of increase of $M_1$ is higher than $M_2$, and that after the pulse, $M_1$ reaches its steady state value quicker than $M_2$ (parameter values: $d_{12} = 2.5$, $d_{21} =0.2$, $d_2 = 2.2$).
(b) Dynamical response of the steady state to a sudden negative pulse in intercisternal transfer $d_{12}$ (inset). 
As the intercisternal flux from cisterna $1$ to $2$ goes down, cisterna $1$ grows and cisterna $2$ diminishes. Here again, after the pulse, $M_1$ reaches its steady state value quicker than $M_2$ (parameter values: $v = 2.3$, $d_{21} =0.2$, $d_2 = 2.2$).   (c) Dynamical response of the steady state to a sudden step jump in exit rate $d_{2}$ (inset), at  small retrograde rate, $d_{21}$  (parameter values $v = 2.3$, $d_{12} = 2.5$, $d_{21} =0.2$). At a smaller step jump $d_2$ (inset, cyan dashed line), the size of the second cisterna decreases, leading to reduction in  anterograde flux from the first cisterna,
and a concomitant increase in the size of the first cisterna (dashed red line). At larger step jump $d_2$ (inset blue line), the second cisterna gets rapidly depleted, and the first cisterna grows unboundedly.
(d) Dynamical response of the steady state to a sudden step jump in exit rate $d_{2}$ (inset), at larger retrograde rate, $d_{21}$ (parameter values $v=1.3$, $d_{12}=2$, $d_{21}=1.2$). At a smaller step jump $d_2$ (inset, cyan dashed line), the cisternae resettle into a new stable configuration  with a smaller size (dashed lines). At larger step jump $d_2$ (inset blue line), the second cisterna gets rapidly depleted, reducing the retrograde as well as the anterograde flux. However, in this regime, the net influx ($\text{J}_\text{in} + \mathcal{J}_{21}^{\text{fus}}$) at the first cisterna is not sufficient to balance the net outflux  ($\mathcal{J}_{12}^{\text{fis}} + \text{J}_{\text{leak}}$)  and the system collapses into a vesicle phase (solid lines). This can also be seen from the nullclines and flows in Fig.\,\ref{fig:2cist_retro_fig}(d,e) (see \ref{sec:2cist_retro} for details).}
\label{fig:sypt}
\end{figure}

We now ask how the different nonequilibrium phases, associated with stable fixed points or limit cycles, respond to external perturbations or systematic internal (cellular) variations. We first study the dynamical response of stable cisternal structures to 
sudden perturbations, such as (i) sudden change in influx $v$, (ii) sudden change in intercisternal flux $d_{12}$, and (iii) sudden change in outflux $d_2$, that may be triggered by biochemical agents or being subject to cellular stress. For this analysis, we continue to use Eqs.\,\eqref{eq:retrograde1},\eqref{eq:retrograde2} with initial conditions set at the steady state configurations.

We consider the response of the stable cisternal phases  
to a sudden step or rectangular pulse, say in the influx rate
$v = v_0 + v_p \,\mathcal{W} (t_0,t_w)$, where the perturbation
is switched on at $t_0$ for a time $t_w$ with unit amplitude $\mathcal{W}(t_0,t_w)$,
and similarly for the other parameter rates.
Such perturbations are easily achieved by treatment with specific drugs and subsequent washout. For instance, Brefeldin A (BFA), a fungal metabolite, prevents the formation of COPI vesicles~\cite{helms1992,itocis}, thus affecting fission from the cis-Golgi cisterna and H89 inhibits COPII recruitment to ER exit sites \cite{puri2003}, thus changing the influx rate $v$.  
Ilimaquinone (IQ), a sea sponge metabolite, appears to  specifically target the trans Golgi network (TGN) and induce rapid PKD (Protein Kinase D)-dependent fission reaction at TGN~\cite{vvquin}, thus enhancing the exit rate $d_2$.
Figure\,\ref{fig:sypt}(a-d) displays the 
dynamical responses of the stable 2-cisternae phase to these sudden perturbations.

The dynamical response of the system to these perturbations provides insight into the form of the fusion-fission kernels and feedback structure 
 that determine the solution space, 
and can be monitored using live-cell imaging.
The dynamical response to sudden changes in the influx (Fig.\,\ref{fig:sypt}(a)) is along expected lines.
It is of interest, however, that a sudden negative pulse in the intercisternal flux $d_{12}$ (obtained, for instance, by reducing the fission from cisterna-1) leads to a swelling of cisterna-1, together with a decrease in the size of cisterna-2 (Fig.\,\ref{fig:sypt}(b)). This has potential implications for the fate of trans-Golgi cisterna upon treatment with Brefeldin-A~\cite{lippincott2000,helms1992,niu2005} which directly acts on the cis and medial Golgi cisternae.
On the other hand, the response 
to a sudden step jump in the exit rate $d_2$ (Fig.\,\ref{fig:sypt}(c,d)) is nontrivial and depends on the intercisternal rate 
$d_{21}$ and on the functional form of the intercisternal kernels. Experiments using IQ to disrupt the TGN often find that the preceding cisternae also get disrupted~\cite{vvquin}. Our analysis would suggest that this would be the case when the retrograde rate $d_{21}$ is high. The time for complete fragmentation depends on the magnitude of the step jump in $d_2$. Further, we propose that in situations where the retrograde rate $d_{21}$ is low, the response would be closer to Fig.\,\ref{fig:sypt}(c) where the outer cisterna disrupt while the inner cisterna is maintained at a larger size.


Next we turn our attention to systematic periodic variations, where the nonequilibrium cisternal dynamics couples to internal cellular networks, such as the cell cycle, 
leading to strong cell-state dependent changes. The cell cycle is under the control of circadian clocks~\cite{Matsuo2003},
via a variety of kinases (such as CK1~\cite{komatsu}) that might regulate the synthesis or influx from the ER~\cite{komatsu}.
This results in a periodic influx rate $v$ which, for the 2-cisternae dynamics, we model as 
\begin{eqnarray}
\label{eq:chaos}
\dot{M}_1 = s_D \, sq(\omega_D , t)  + \mathcal{J}_1 (M_1,M_2) \hspace{2cm} \dot{M}_2 =  \mathcal{J}_2 (M_1,M_2)
\end{eqnarray}
where  $s(t)= s_D \, sq(\omega_D,t)$ is a periodic rectangular pulse with driving frequency $\omega_D$ and strength $s_D$, and
$\mathcal{J}_1,\mathcal{J}_2$ are given by R.H.S. of Eqs.\,\eqref{eq:retrograde1},\eqref{eq:retrograde2}, respectively.
To model the cell cycle, we take the time period of the driving signal (cell cycle time) to be $\sim 10$\,hrs, corresponding to $\omega_D \sim 0.00017 \, rad\,s^{-1}$, and the driving strength such that the minimum influx approaches zero during the mitotic phase of the cell. 
Observationally, it is well established that Golgi cisternae in mammalian cells respond by disassembling  at the onset of mitosis and reassembling at its completion~\cite{shima,preuss}.





We first study the effect of the periodic influx on the stable 2-cisternae phase. In this case, provided the amplitude of the perturbation is large enough to cross the phase boundary to the vesicle or limit cycle phase, the cisternae simultaneously fragment and reassemble Fig.\,\ref{fig:arnold_forced}(a) at a period predominantly set by the cell cycle frequency, with nonlinearities exciting higher frequency modes.




We next study the effect of the periodic influx on the in-phase limit cycle phase. 
We find that for small driving amplitude, the observed frequency, $\omega_{obs}$ of the driven system becomes mode-locked to a rational multiple of the driving frequency, $\omega_{D}$, a phenomena known as entrainment (Fig.\,\ref{fig:arnold_forced}(b)). This locking persists over a finite range of detuning (or frequency mismatch) as we vary the system frequency $\omega_0$ by changing the retrograde rate $d_{21}$, creating the regions known as {\it Arnold Tongues} (Fig.\,\ref{fig:arnold_forced}(c))~\cite{strogatz, kuznetsov}. Complete dissolution observed in Fig.\,\ref{fig:arnold_forced}(b) for limit cycle solutions occurs even when the driving amplitudes are small. 
This happens because the system spends significant time near the \textit{SN-ghost}.

It is observed that in mammalian tissue culture cells, but not in plant cells or in S. cerevisiae, protein transport from the ER is blocked during mitosis~\cite{VM1}. Thus  plant cells and S. cerevisiae are not subject to periodic influx rate and so do not undergo  fragmentation~\cite{VM1}. This suggests that the periodic dissolution and reformation of cisterna across the cell cycle should be associated with driving the stable 2-cisternae phase with a
periodic influx, rather than the entrainment of the limit cycle phase.
Nevertheless, it would be interesting to search for evidence of in-phase limit cycle solutions in the Golgi apparatus across the eukaryotic domain, or in synthetic realizations.



\begin{figure}[t!]
\centering
\includegraphics[width=0.95\textwidth]{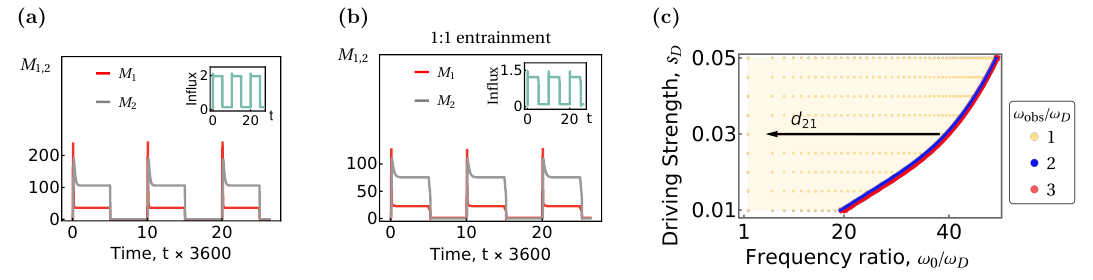}

\caption{{\bf Forced oscillations and entrainment for the $2$-cisternae system.} 
(a) A periodic influx drive (inset) applied to the stable 2-cisternae phase behaves as a periodically forced stable system  (parameter values $v=2$, $d_2=1.8$, $d_{12}=2.5$, $d_{21}=0.2$, $d_1= 1\,s^{-1}$, time period of the pulse $=\, 10 \,hrs$,  $s_D=0.2$).  For large enough amplitude ($s_D=0.2$), the driving force can make the system cross the phase boundary ($2$-cisternae phase $\leftrightarrow$ limit cycle phase, for the parameter values used), leading to periodic dissolution and reassembly of the cisternae.
(b) A periodic influx drive (inset) applied to the in-phase oscillatory phase leads to entrainment
for low driving strength ($s_D \sim 0.02$) and long time period $=\, 10 \,hrs$. 
The parameter values used are, $v=1.37$, $d_2=0.89$, $d_{12}=2.5$, $d_{21}=0.95$, $d_1= 1\,s^{-1}$, which results in a system frequency, $\omega_0 =0.00032\,d_1\, rad$.
(c) Entrainment of the spontaneously oscillating system with the driving frequency showing the characteristic Arnold tongues. On the x-axis is the natural frequency, $\omega_{0}$ (i.e., system frequency without drive) in units of the driving frequency $\omega_D$, while the observed system frequency $\omega_{obs}$ in the presence of drive appears in the adjacent panel.
We see from the Arnold tongue diagram that for non-zero driving strength,  as we vary the system frequency $\omega_{0}$ by varying the retrograde rate $d_{21}$, the observed frequency stays locked in $1:1$ ratio with the driving frequency (yellow region) before moving to other frequency ratios at higher signal frequencies (blue and red strips). Parameter values are $v=1.37$, $d_2=0.89$, $d_{12}=2.5$, $d_{21}=0.7-0.95$, $d_1= 1\,s^{-1}$, with the time period of the driving pulse signal $=\, 10 \,hrs$. At larger retrograde flux rate, the width of the first Arnold tongue shrinks, as shown by the arrow.
 This entrainment is observed even at small driving amplitudes ($s_D \sim 0.01-0.05$). 
}
\label{fig:arnold_forced}
\end{figure}

\subsection{Robustness of the solution space}
So far, our analysis has been mean field, where we have ignored the effect of noise, an inevitable presence in cellular systems. Cellular noise can arise from parametric noise, extrinsic noise (say, influx rate) and intrinsic noise (say, low copy number or chemical). An important requirement of organelle biogenesis in a living cell should be robustness under noise. 
In Sect.\,\ref{sec:single_cist} and \ref{sec:single_cist_si}, we have discussed the \textit{structural stability} of the single cisterna dynamical system Eq.\,\eqref{eq:veq_Hd} under parametric noise, 
i.e. states in the vicinity of a viable solution also converge to the same solution class or phenotype. The structural stability of the $2$-cisternae dynamical system Eqs.\eqref{eq:retrograde1},\,\eqref{eq:retrograde2}, although not as straightforward as the single cisterna, can be established numerically (see Fig.\,\ref{fig:robust_algebraic} and \ref{sec:cist_struct} for details) by checking that the number of distinct roots that the dynamical system admits matches its algebraic capacity (the generic number of roots that the system algebraically supports)
using the Gr\"obner basis~\cite{Cox2015,golubitzky}.

It has been suggested that the structural stability of a dynamical system ensures that the root structure of the system is robust to low amplitude extrinsic noise~\cite{cobb}. 
Changes in the root structure of the $1$-cisterna and $2$-cisternae systems in the presence of extrinsic multiplicative noise can be studied by extending the deterministic system to a Langevin system \cite{kampen}, for instance, by adding a multiplicative noise with noise strength $\sqrt{2\xi_i(M_i)}$, which is a smooth positive function solely dependent on the size of the cisterna $i$. 
For the 1-cisterna case, taking the noise strength $\xi(M) = \epsilon M$, which would be the case if the stochastic addition and removal of vesicles are proportional to the cisternal size,
we see that the nucleation rate is modified $v \, a_0 \rightarrow (v \, a_0-\epsilon)$. Although it can change the root structure, it also suggests that a non-zero nucleation rate $v \, a_0$ provides robustness against noise with strength $\sqrt{\epsilon \, M}$, with the bound for robustness on the noise strength being $\epsilon < a_0 \, v$. For further details, see \ref{sec:robust_si}.

The robustness of the solution space (phases) under intrinsic noise can be studied using  stochastic Gillespie simulations~\cite{gillespieE} with propensities given by the effective fission and fusion rates. The asymptotic mean value from the stochastic trajectories can be used to generate a phase diagram that includes intrinsic noise. For single cisterna system, Fig.\,\ref{fig:robust_struct}(a)  shows that in the presence of intrinsic noise, the cisternal phase apparently occupies a larger region (light orange) in parameter space. However, analysis of the cisternal size shows a large coefficient of variation, the ratio of the standard deviation to the mean (CV $>1$), suggesting that cisternal configurations are dominated by large copy number fluctuations, especially near the phase boundary adjoining the vesicle phase.
Similar results are observed for the two cisternae system,  Fig.\,\ref{fig:robust_struct}(b). Further details in~\ref{sec:robust_si}.

\begin{figure}[t!]
\centering
\includegraphics[width=0.95\textwidth]{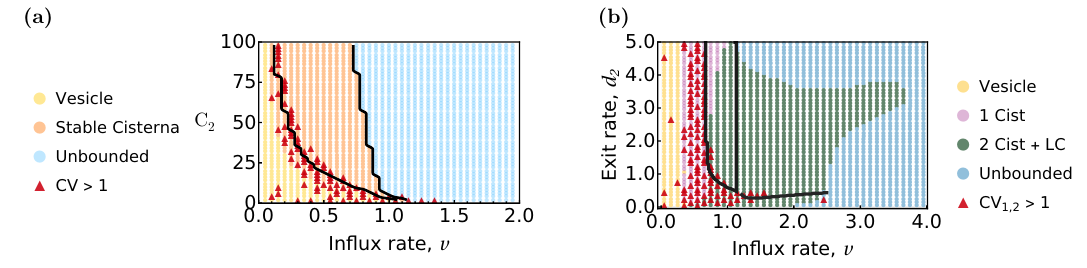}
\caption{{\bf Robustness of the stable cisternal phases to intrinsic noise}. 
We include the effects of intrinsic noise using stochastic Gillespie simulations~\cite{gillespieE} to
generate phase diagrams for (a) single cisterna  (parameter values $C_{1}= 100$, $a_{0}=0.05$, see Eq.\,\eqref{eq:veq_Hd})  and (b) two cisternae (parameter values $C_{11} = C_{21} = 100$, $C_{22} = 20$, $C_{12} = 20$, $d_{12} = 3$, $d_{21}=0$, $a_{1}=a_{2}=0.05$, see Eqs.\,\eqref{eq:retrograde1},\eqref{eq:retrograde2}). In (a), we have plotted the phase boundaries (thick lines) between
stable cisterna-unbounded growth and  vesicle-cisterna, of the deterministic system. See Fig.\,\ref{fig:p_roots1}(b) for a more elaborate phase diagram. 
In (b), we have plotted the phase boundaries (thick lines) between
vesicle and $1$-cisterna, vesicle and $2$-cisternae, $1$-cisterna and $2$-cisternae (here we do not differentiate between stable $2$-cisternae and limit cycle), of the deterministic system. See Fig.\,\ref{fig:overlap}(a) for a more elaborate phase diagram.
We see that the mean field phases are robust under intrinsic noise, except at small values of influx rate, $v$, in the vicinity of the bifurcations.
In these regions, marked by red triangles, the cisternal sizes exhibit higher coefficient of variation (CV) suggesting the system is dominated by copy number variations (also see Fig.\,\ref{fig:robust_crit_in} and \ref{fig:robust_noise_intrinsic}). }
\label{fig:robust_struct}
\end{figure}

\subsection{Size-dependent embedded feedforward control via fusion-fission kernels \label{sec:control_FF}}

The above considerations suggest that the {\it system of Golgi cisternae} is maintained at a homeostatic set point \cite{DorfBishop2017} $(M_1^*, M_2^*, \ldots M_n^*)$ in the presence of a net vectorial flux  of vesicles. Here we see that this homeostasis is brought about by a nested control with the outer loop endowed with fixed levels of cisterna-specific fusogens and fisogens, and the inner loop comprising embedded 
active fusion-fission cycles. Together, this gives rise to a size-dependent embedded control system driven by fusion-fission kernels. 


This nested control architecture, Fig.\,\ref{fig:contol_sys_FF}, has the following elements:\\
\begin{itemize}
\item \textit{Feedforward control}: The continuous flux of COPII-coated vesicles from the ER provides the bulk material to assemble the Golgi cisternae~\cite{Alberts2008}. However, the homeostatic set point of the cisternae is contingent not only on this flux {\it but also on the levels of cisternal-specific fusogens and fisogens}, which in turn determine the other (inter-cisternal) fluxes, as discussed in Sect.\,\ref{sec:two_cisternae}.
This control element is represented as the outer feedforward layer in Fig.\,\ref{fig:contol_sys_FF}. This involves both the anterograde and the retrograde fluxes.

\item \textit{Embedded feedback controls}: However, this on its own does not ensure robust homeostasis of cisternal sizes.
To ensure this, the Golgi cisternae must rely on size-dependent feedback loops, as described in \ref{sec:si00}-\ref{sec:Mastereqtwocisternae}, depicted as the inner regulatory layers in Fig.\,\ref{fig:contol_sys_FF}. This involves a coupling between the outer and the inner loops in the form of a feedback.
\end{itemize}

Taken together, Fig.\,\ref{fig:contol_sys_FF}, depicts that  the Golgi cisternae operate {\it macroscopically} as a  feedforward control system \cite{DorfBishop2017} -- the  vectorial vesicle flux from the ER serves as the input, and the material outflux, alongside the homeostatic set point of cisternae ($M_1^*, M_2^*$) serves as the output. To ensure the robustness of this process, the system relies on a {\it microscopic} embedded feedback elements. These nested elements dynamically regulate cisternal size in response to the instantaneous ER flux and the availability of fusogens and fisogens via cisternal-specific fusion and fission kernels.

\begin{figure*}[t!]
\centering
\includegraphics[width=0.98\textwidth]{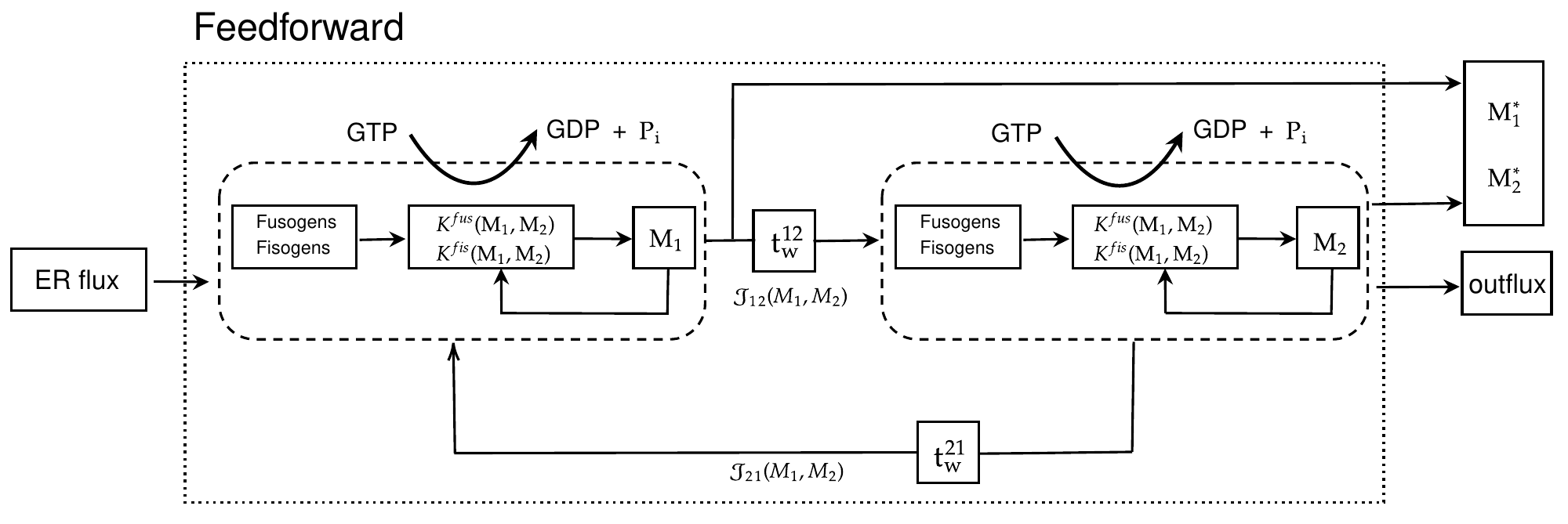}
\caption{{\bf Feedforward control with embedded feedback}. The outer dashed box is the feedforward layer -- input to this layer is the vectorial flux of vesicles from the ER and the output is material outflux along with formation of stable cisternae ($M_1^*,M_2^*$), with no external measurement of $M_1(t),M_2(t)$ feeding back. The inner dashed boxes are embedded feedback controls that regulate the cisternal size for a given ER flux and availability of fusogens and fisogens via fusion and fission kernels, $\mathcal{K}^{fus}(M_1,M_2)$ and $\mathcal{K}^{fis}(M_1,M_2)$ respectively. In the schematic, $\mathcal{K}^{fus}(M_1,M_2)$ and $\mathcal{K}^{fis}(M_1,M_2)$ include the influx from ER, leak fluxes as well as anterograde $\mathcal{J}_{12} (M_1,M_2)$ and retrograde $\mathcal{J}_{21} (M_1,M_2)$ fluxes within intercisternal delay time windows $t_w^{12}$ and $t_w^{21}$, respectively (see Eqs.\,\eqref{eq:integral},\eqref{eq:retrograde1},\eqref{eq:retrograde2}).}
\label{fig:contol_sys_FF}
\end{figure*}

\section{Discussion}
\label{sec:discussion}


In this paper, we have presented a general nonequilibrium framework to study the establishment, control, robustness and dynamical response of a system of intracellular membrane bound organelles, such as the Golgi complex, which are subject to a nonequilibrium flux of material processes through synthesis, transport, fusion and fission. To do this, we have had to derive macroscopic fusion and fission flux kernels from microscopic considerations, that form a size-dependent {\it embedded feedforward control system} that ensure the assembly of multiple Golgi cisternae.
This assembly, depending on the functional form of the flux kernels, admits different solution classes, which we identify with distinct phenotypic outcomes such as vesicular transport, cisternal progression, or the periodic dissolution into tiny vesicles. We have studied how these outcomes may be entrained with the cell cycle, giving rise to cell state dependent cisternal organisation.

We have determined how the different phases of cisternal assembly dynamically respond to systematic perturbations, both external and internal. In future we propose to study how specific dynamical responses to experimentally realisable protocols, may enable the quantitative reconstruction of the fusion and fission kernels and its embedded feedback structure, providing insights into the physical 
principles underlying cisternal assembly. 
We showed  that the deduced effective fusion and fission kernels give rise to stable unfoldings that maintain the solution structure under perturbations with finite strength. We have further elaborated on this issue by subjecting the dynamical system model to intrinsic and extrinsic noise and derive bounds on the noise strengths for which system remains in the same solution class, hence a robust solution structure.

There remain many open issues which we hope to take up in the future -- 

\begin{enumerate}

\item  While the cellular dynamics discussed here, provides a framework to study its biogenesis and the control and maintenance of cisternal size, it says says nothing about what controls cisternal number~\cite{Alkesh}. This could likely arise from additional cisternal attributes not  included in our study of cisternal dynamics. Indeed cisternae are identified not only by their  position relative to the ER and PM, but also by their lumenal or membranal chemical  content, such as cisternal specific enzymes. We therefore proposed  to extend the nonequilibrium framework developed here to include the dynamics of cisternal enzymes and their retention~\cite{Alkesh}. This may be done by coupling the mechanochemical enzymes involved in the fission-fusion cycles to the 
glycosylation enzymes, resulting in a hierarchical control of cisternal size and number,
a preliminary attempt using a spatial model was made in~\cite{Sachdeva2016}. 





\item  An interesting cell biology question is how  organelle size scales with cell size~\cite{chan, marshallS}. In the context of the Golgi apparatus, an increase in cell volume could result in an increase in the size of individual cisternae, an increase in cisternal number or, an increase in Golgi copy number. 
Can we address this within the framework of our 
Golgi assembly model, for instance by looking at how the model parameters, such as rate of influx, and levels of fisogens and fusogens, depend on cell size?

\item Since finite size cellular structures are a consequence of the nonequilibrium dynamics and feedback control, we may ask for the thermodynamic cost 
in creating and maintaining stable cisternae~\cite{vaikuntanathan}. 

\end{enumerate}




Finally, we believe our model for the biogenesis, maintenance and control of the Golgi cisternae can be extended to other intracellular organelles~\cite{marshallS}: {\it open, nonconserved systems} like the Golgi, such as the endosomal and lysosomal systems, and {\it closed, conserved systems} such as the mitochondrial system.

\section{Acknowledgements}
We thank Mukund Thattai and Vivek Malhotra for numerous discussions.
This work was supported by the Department of Atomic Energy, Government of India, Project Identification No.\,RTI 4006.
We acknowledge support from the Simons Foundation (Grant No.\,287975) and DST (India) for a JC Bose Fellowship (JCB/2018/00030) and thank NCBS-TIFR for Computational Facilities.
MR acknowledges the Physics of Life Chair Professorship supported by TTK-Prestige.





\medskip
\normalem
\bibliographystyle{unsrt}

\appendix 

\setcounter{section}{0}
\setcounter{equation}{0} 
\setcounter{figure}{0}
\setcounter{table}{0}
\counterwithout{equation}{section}
\counterwithout{figure}{section}
\counterwithout{table}{section}
    \renewcommand{\thesection}{S\arabic{section}}
    \renewcommand{\thesubsection}{S\arabic{section}.\arabic{subsection}}
    
\renewcommand{\theequation}{S\arabic{equation}}
\renewcommand{\thetable}{S\arabic{table}}
\renewcommand{\thefigure}{S\arabic{figure}}
\makeatletter
\@addtoreset{prop}{section}
\@addtoreset{corollary}{section}
\renewcommand{\p@subsection}{}
\makeatother
\section{Master equation for de novo assembly of Golgi cisternae\label{sec:si00}}
\subsection{Detailed derivation for $n$-enzyme model for  Master equation\label{sec:si0}}

Consider a cisterna of size $M$, with $N_+$ fusion enzymes and $N_-$ fission enzymes. 
The fusion and fission cycles consist of enzyme-cisternal membrane complex going over a four state Markov cycle,  in the end of which size of the cisterna increases or decreases by one unit of vesicle. 
The complete state space for this process can be represented by $(M,\{i_{r}\},\{k_{s}\},t)$, $r \in (1,\ldots,N_+)$ and $s \in (1,\ldots,N_-)$ where $r,s$ represent the $r^{\it th}$ and the $s^{\it th}$, fusion and fission cycles,  respectively.  The change in cisternal size $M$ is actuated by transitions between internal states (see Fig.\,\ref{fig:golgi_fluxSI}(e)), resulting in the 
dynamics of the probability distribution $P(M,\{i_{r}\},\{k_{s}\},t)$, 
\begin{small}
\begin{eqnarray}
{\dot P}(M, \{i_r\}, \{k_s\}, t) &=&   \sum_{\{j_r\} \in  \mathfrak{I}_{\hspace{-0.5pt}\{i_r\}}^{\hspace{-0.1pt}+} }  \; \left( - K_{\{i_r \, j_r\}}^{+} \; P(M,  \{i_r\}, \{k_s\}, t) +  K_{\{j_r \, i_r\}}^{+} \; P(M + \mathbb{M}^+_{\{j_r \, i_r\}}, \{j_r\}, \{k_s\}, t) \right) \nonumber \\ &+&  \sum_{\{l_s\} \in \mathfrak{I}_{\hspace{-0.5pt}\{k_s\}}^{\hspace{-0.1pt}-} } \; \left(-K_{\{k_s l_s\}}^{-} \; P(M,  \{i_r\}, \{k_s\}, t)  + K_{\{l_s k_s\}}^{-} \; P(M - \mathbb{M}^-_{\{l_s k_s\}},  \{i_r\}, \{l_s\}, t) \right)  \,, 
\label{eq:Master_all_si}
\end{eqnarray}
\end{small}
where the overdot represents the time derivative,  and  the index sets $\mathfrak{I}_{\hspace{-0.5pt}\{i_r\}}^{\hspace{-0.1pt}+}$ and $\mathfrak{I}_{\hspace{-0.5pt}\{k_s\}}^{\hspace{-0.1pt}-}$ contain all internal states that are directly connected to the states $i_r$ and $k_s$ in the $r^{\it th}$ fusion and the $s^{\it th}$ fission cycle, respectively; for example, the states one and zero in Fig.\,\ref{fig:golgi_fluxSI}(c,d) have $\mathfrak{I}_{\hspace{-0.5pt}1}^{\hspace{-0.15pt}\pm} = \left\{0,2,4\right\}$ and $\mathfrak{I}_{\hspace{-0.5pt}0}^{\hspace{-0.15pt}\pm} = \left\{1\right\}$, respectively, for all cycles. Hence, $M$ is a random variable  that represents the mass (size) states and $i_r,j_r,k_s,l_s$ are internal states. We denote the amount of ``virtual''  mass exchanged in a transition from the state $j_r (l_s)$ to $i_r (k_s)$ for the $\pm\,$-cycle, by $\mathbb{M}^{+}_{j_{r} \,i_{r}} (\mathbb{M}^{-}_{l_{s}\, k_{s}})$. At the end of each cycle, exactly one unit of ``real'' mass of a single vesicle is transferred, thus, $\sum^{4}_{j_r=1} \sum_{i_r\hspace{0.5pt}\in\hspace{1pt}\mathfrak{I}_{\hspace{-1pt}j_{r}}^{\hspace{-0.15pt}+}}^{i_r>j_r} \mathbb{M}^+_{j_r \, i_r} = - \sum^{4}_{l_s=1}\sum_{k_s\hspace{0.5pt}\in\hspace{1pt}\mathfrak{I}_{\hspace{-1pt}k_{s}}^{\hspace{-0.15pt}-}}^{k_s>l_s}  \mathbb{M}^-_{l_s \, k_s} = 1$, for all $r,s$.
The transition rates in the fusion (fission) cycle, ${\bf K}^{\pm}$ depend on the availability of  vesicles $N_\v$, fusogens (i.e., fusion enzymes and their associated proteins) $N_+$ and fisogens (i.e., fission enzymes and their associated proteins) $N_-$ at the cisternae, and $P(M,  \{i_r\}, \{k_s\}, t) $ above is conditioned over these. These transition rates also depend on the mechano-chemistry of underlying enzyme-membrane interactions that depend on the instantaneous membrane properties, such as membrane tension and rigidity (and composition), which depend on $M$. More details on these dependencies on the cisternal size are discussed in \ref{subsec:markovcycle}.

Since the dynamics through the internal states is ``fast'' \cite{arf1,allin}, and so is the dynamics governing the availability of vesicles and the $\pm$-species,  one may multiply Eq.\,\eqref{eq:Master_all_si} by $P_{ss}(N_\v,N_+,N_-)$,
sum over the internal states and $N_\v, N_+, N_-$. The left side of Eq.\,\eqref{eq:Master_all_si} can then be simplified using the relation,
\begin{small}
\begin{equation}
 \label{eqn:sum_internal}
\sum_{N_\v,N_+,N_-} \, \sum_{\{i_r\}, \{k_s\}}\, P_{ss}(N_\v,N_+,N_-) \, P(M, \{i_r\}, \{k_s\}, t|N_\v,N_+,N_-) \!=\!
P(M, t)\,.
\end{equation} 
\end{small}
To simplify the right hand side of Eq.\,\eqref{eq:Master_all_si}, we will need a further assumption that fusion and fission events are independent, and so 
\begin{small}
\begin{equation}
 \label{eqn:decomposition}
P(M, \{i_r\}, \{k_s\}, t \vert N_\v,N_+,N_-) = P(M,t \vert N_\v,N_+,N_-) P_{ss}(\{i_r\} \vert M, N_\v, N_+)\, P_{ss}(\{k_s\} \vert M, N_-) \,,
\end{equation} 
\end{small}
where the last two terms on the right are the steady state probability of the set of internal states $\{i_r\}$ and $\{k_s\}$, conditioned on $M, N_\v, N_+, N_-$~\cite{qian}, i.e. over the fusion and fission cycles,  $M, N_\v, N_+, N_-$ remain constant \cite{Matsuura-Tokita}. Summing Eq.\,\eqref{eq:Master_all_si} over internal states leads to
\begin{small}
\begin{eqnarray}
&& {\dot P}(M, t, \vert N_\v,N_+,N_-) = \nonumber \\
&& - \sum_{\substack{\{i_r\},\{k_s\} , \\ \{j_r\} \in  \mathfrak{I}_{\hspace{-0.5pt}\{i_r\}}^{\hspace{-0.1pt}+}} }  \; \left(K_{\{i_r \, j_r\}}^{+} \;  P(M,t \vert N_\v,N_+,N_-) P_{ss}(\{i_r\} \vert M, N_\v, N_+)\, P_{ss}(\{k_s\} \vert M, N_-) \right) \nonumber \\ &+& \sum_{\substack{\{i_r\},\{k_s\},\\ \{j_r\} \in  \mathfrak{I}_{\hspace{-0.5pt}\{i_r\}}^{\hspace{-0.1pt}+} }}  \left(K_{\{j_r \, i_r\}}^{+} \;   P(M + \mathbb{M}^+_{\{j_r \, i_r\}},t \vert N_\v,N_+,N_-) \, P_{ss}(\{j_r\} \vert M + \mathbb{M}^+_{\{j_r\, i_r\}}, N_\v, N_+) \, P_{ss}(\{k_s\} \vert M + \mathbb{M}^+_{\{j_r \, i_r\}}, N_-)\right) \nonumber \\ &-& \sum_{\substack{\{i_r\},\{k_s\}, \\ \{l_s\} \in  \mathfrak{I}_{\hspace{-0.5pt}\{i_r\}}^{\hspace{-0.1pt}+}} } \; \left(K_{\{k_s \, l_s\}}^{-} \;  P(M,t \vert N_\v,N_+,N_-) \, P_{ss}(\{i_r\} \vert M, N_\v, N_+)\, P_{ss}(\{k_s\} \vert M, N_-) \right) \nonumber \\  &+& \sum_{\substack{\{i_r\},\{k_s\}, \\ \{l_s\} \in  \mathfrak{I}_{\hspace{-0.5pt}\{k_s\}}^{\hspace{-0.1pt}+} }}  \left( K_{\{l_s \, k_s\}}^{-} \; P(M - \mathbb{M}^-_{\{l_s \, k_s\}},t \vert N_\v,N_+,N_-) \, P_{ss}(\{i_r\} \vert M - \mathbb{M}^-_{\{l_s \, k_s\}}, N_\v, N_+) \,P_{ss}(\{l_s\} \vert M -  \mathbb{M}^-_{\{l_s \, k_s\}}, N_-) \right) \,. 
\label{eq:current_deriv1_si}
\end{eqnarray}
\end{small}
Due to the independence of fission and fusion events assumed above, we can use $\sum_{k_s=1}^4 P_{ss}(k_s \vert M, N_-)=1, \, \sum_{k_s=1}^4 P_{ss}(k_s \vert M + \mathbb{M}^+_{\{j_r\, i_r\}}, N_-)=1 \,\forall \,r$ in  the first two lines, and $\sum_{i_r=1}^4 P_{ss}(i_r \vert M, N_\v, N_+)=1,\,\sum_{i_r=1}^4 P_{ss}(i_r \vert M - \mathbb{M}^+_{\{l_s\, k_s\}}, N_\v, N_+)=1 \, \forall \, s $ in the last two lines to sum these out. With this simplification, consider the terms corresponding to the fusion cycle,
\begin{small}
\begin{eqnarray}
&& {\dot P}(M, t \vert N_\v,N_+,N_-) \Bigg|_{fusion}  = \nonumber \\
&& \left(-\sum_{\substack{\{i_r\} , \{j_r\} \in  \mathfrak{I}_{\hspace{-0.5pt}\{i_r\}}^{\hspace{-0.1pt}+} \\ j_r > i_r} }  \; \left(K_{\{i_r \, j_r\}}^{+} \;  P(M,t \vert N_\v,N_+,N_-) P_{ss}(\{i_r\} \vert M, N_\v, N_+)\right) \right. \nonumber \\ && \left. +\sum_{\substack{\{i_r\}, \{j_r\} \in  \mathfrak{I}_{\hspace{-0.5pt}\{i_r\}}^{\hspace{-0.1pt}+} \\ j_r > i_r}}  \left(K_{\{j_r \, i_r\}}^{+} \;   P(M +  \mathbb{M}^+_{\{j_r \, i_r\}},t \vert N_\v,N_+,N_-) P_{ss}(\{j_r\} \vert M +  \mathbb{M}^+_{\{j_r\, i_r\}}, N_\v, N_+)\, \right) \right) \nonumber \\ &&  \left( -\sum_{\substack{\{i_r\} , \{j_r\} \in  \mathfrak{I}_{\hspace{-0.5pt}\{i_r\}}^{\hspace{-0.1pt}+} \\ j_r < i_r} }  \; \left(K_{\{i_r \, j_r\}}^{+} \;  P(M,t \vert N_\v,N_+,N_-) P_{ss}(\{i_r\} \vert M, N_\v, N_+)\right) \right. \nonumber \\ && \left. +\sum_{\substack{\{i_r\}, \{j_r\} \in  \mathfrak{I}_{\hspace{-0.5pt}\{i_r\}}^{\hspace{-0.1pt}+} \\ j_r < i_r}}  \left(K_{\{j_r \, i_r\}}^{+} \;   P(M - \mathbb{M}^+_{\{j_r \, i_r\}},t \vert N_\v,N_+,N_-) P_{ss}(\{j_r\} \vert M - \mathbb{M}^+_{\{j_r\, i_r\}}, N_\v, N_+)\,\right) \right) \,,
\label{eq:current_deriv2_si}
\end{eqnarray}
\end{small}
where we have segregated the sums based on the order of ${\{i_r , j_r\}}$ and used the fact that for the fusion cycle, ``virtual'' mass is added in forward transition and subtracted in backward transition (reflected in the sign of $\mathbb{M}^+_{\{j_r \, i_r\}}$). Since the ``real'' mass is kept fixed over a cycle  (and is set by the lower index in Eq.\,\eqref{eq:current_deriv2_si}), we can club the first two terms and the last two terms in the above sum. We further assume that each fusion cycle works independently (with size increasing by one unit over each cycle), and the steady state current stays the same over all the cycles during the time scale of cisternal size update. As mentioned above, the growth of a cisterna of size $M$ is conditioned on the availability of fusion enzymes (fusogens) and fission enzymes (fisogens) - $P(N_+)$ and $P(N_-)$ respectively. Hence, the number of fusion events per unit time is given by the available number of fusogens $(N_+)$,

\begin{small}
\begin{eqnarray}
&&{\dot P}(M, t \vert N_\v,N_+,N_-) \Bigg|_{fusion} = \nonumber \\  && P(M,t \vert N_\v,N_+,N_-) \left(-\sum_{\substack{\{i_r\},\{j_r\} \in  \mathfrak{I}_{\hspace{-0.5pt}\{i_r\}}^{\hspace{-0.1pt}+} \\ j_r>i_r }}  \; K_{\{i_r \, j_r\}}^{+} \;  P_{ss}(\{i_r\} \vert M, N_\v, N_+) \right.  \,\left. + \sum_{\substack{\{i_r\},\{j_r\} \in  \mathfrak{I}_{\hspace{-0.5pt}i}^{\hspace{-0.1pt}+} \\ j_r>i_r } }  K_{\{j_r \, i_r\}}^{+} \;  P_{ss}(\{j_r\} \vert M, N_\v, N_+)\, \right) \nonumber \\ &+&  P(M-N_+,t \vert N_\v,N_+,N_-) \left( - \sum_{\substack{\{i_r\},\{j_r\} \in  \mathfrak{I}_{\hspace{-0.5pt}\{i_r\}}^{\hspace{-0.1pt}+} \\ j_r<i_r } }  K_{\{i_r \, j_r\}}^{+} \;   P_{ss}(\{i_r\} \vert M-N_+, N_\v, N_+) \right. \, \left. + \sum_{\substack{\{i_r\},\{j_r\} \in  \mathfrak{I}_{\hspace{-0.5pt}\{i_r\}}^{\hspace{-0.1pt}+} \\ j_r<i_r }}  \; K_{\{j_r \, i_r\}}^{+} \;  P_{ss}(\{j_r\} \vert M-N_+, N_\v, N_+)\, \right) \,, \nonumber \\
\label{eq:current_deriv3_si0}
\end{eqnarray}
\end{small}
which can be written as
\begin{small}
\begin{eqnarray}
{\dot P}(M, t \vert N_\v,N_+,N_-) \Bigg|_{fusion} = - P(M,t \vert N_\v,N_+,N_-)  \, J_{ss}^+ (M,N_\v,N_+)+ P(M-N_+, t \vert N_\v,N_+,N_-)  \, J_{ss}^+ (M-N_+,N_\v,N_+)  \,,
\label{eq:current_deriv3_si}
\end{eqnarray}
\end{small}
where $J_{ss}^{+}$ is the steady state current over the fusion cycle. Note that for the $1$-enzyme model, these would correspond to a fusion cycle for size $M$ and a fusion cycle for size $M-1$. Now, the number of fission events per unit time is given by the available number of fisogens $(N_-)$. Repeating the same steps for the fission terms in Eq.\,\eqref{eq:current_deriv1_si} with the fact that for the fission cycle, ``virtual'' mass is subtracted for the forward transition and added for the backward transition,
\begin{small}
\begin{eqnarray}
&&{\dot P}(M, t \vert N_\v,N_+,N_-) \Bigg|_{fission} = - P(M,t \vert N_\v,N_+,N_-)  \, J_{ss}^- (M,N_-) + P(M+N_-, t \vert N_\v,N_+,N_-) \, J_{ss}^- (M+N_-,,N_-)  \,,
\label{eq:current_deriv4}
\end{eqnarray}
\end{small}
where $J_{ss}^{-}$ is the steady state current over the fission cycle. Using the above result,
\begin{small}
\begin{eqnarray}
{\dot P}(M,t|N_\v,N_+,N_-) &=&  -P(M,t) \, \left(J^{+}_{ss}(M,N_\v,N_+) + \, J^{-}_{ss}(M,N_-) \right)\nonumber \\ &+& P(M-N_+,t) \,J^{+}_{ss}(M-N_+,N_\v,N_+) + P(M+N_-,t) \,J^{-}_{ss}(M+N_-,N_-) \,.
\label{eq:current_eq_si}
\end{eqnarray}
\end{small}
Each of these loop currents $J^{\pm}_{ss}$ can be written in terms of the transition rates in the fusion (fission) cycle, ${\bf K}^{\pm}$ and the steady state probabilities for the internal states $P_{ss}$. These $J_{ss}^{\pm}$ can be evaluated separately,
\begin{small}
\begin{eqnarray}
\label{eq:current_expr_si1}
J_{ss}^+(M,N_\v,N_+)  &=& \left(\sum_{i,j \in  \mathfrak{I}_{\hspace{-0.5pt}i}^{\hspace{-0.1pt}+},\, j>i}  \; K_{ij}^{+} \;  P_{ss}(i \vert M, N_\v, N_+)\, - K_{ji}^{+} \;  P_{ss}(j \vert M, N_\v, N_+)\, \right)  \\
J_{ss}^-(M,N_-) &=& \left(\sum_{l,k \in  \mathfrak{I}_{\hspace{-0.5pt}i}^{\hspace{-0.1pt}-}, \, k>l }  \; K_{lk}^{+} \;  P_{ss}(i \vert M,N_-)\, - K_{kl}^{+} \;  P_{ss}(j \vert M,N_-)\, \right) \,\, .
\label{eq:current_expr_si2}
\end{eqnarray}
\end{small}
As mentioned before, we have assumed that for a given $N_\v,N_+,N_-$, each fusion/fission cycle works independently, and the steady state current stays the same over all cycles during the time scale of cisternal size update. Summing both sides of Eq.\,\eqref{eq:current_eq_si} over the distribution $P(N_\v,N_+,N_-)$ gives the fusion and fission kernel. With all these details, 
we can write the master equation for single cisterna representing a microscopic model for cisternal size update,
\begin{small}
\begin{eqnarray}
{\dot P}(M,t)  &=& \underbrace{- \sum_{k,n_1} \, P_{ss}(N_\v=k,N_+=n_1)  \; J^{+}_{ss}(M,N_\v,N_+)  \; P(M,t)  \; +  \; \sum_{k,n_1} \, P_{ss}(N_\v=k,N_+=n_1) \;  J^{+}_{ss}(M-n_1,N_\v,N_+) \; \;  P(M-n_1,t) }_{\text{Fusion}}  \nonumber \\  && \underbrace{- \sum_{n_2} \, P_{ss}(N_-=n_2)  \;J^{-}_{ss}(M,N_-)  \; P(M,t)  \; + \;  \sum_{n_2} \, P_{ss}(N_-=n_2) \; J^{-}_{ss}(M+n_2,N_-) \;  P(M-n_2,t)}_{\text{Fission}} \,,
\label{eq:Master_final_sup}
\end{eqnarray}
\end{small}
where the first line corresponds to fusion events and the last line corresponds to fission events. It is also clear from the above that $N_+ =  1$, $N_- =  1$ corresponds to a single event ($1$-enzyme) case, namely Eq.\,\eqref{eqn:markov_eff} shown in the main text. Now, we derive the probability distributions for $N_\v,N_+,N_-$ and show how to compute the fusion and fission kernels for a given physical model.
\begin{figure}[t!]
\centering
\includegraphics[width=\textwidth]{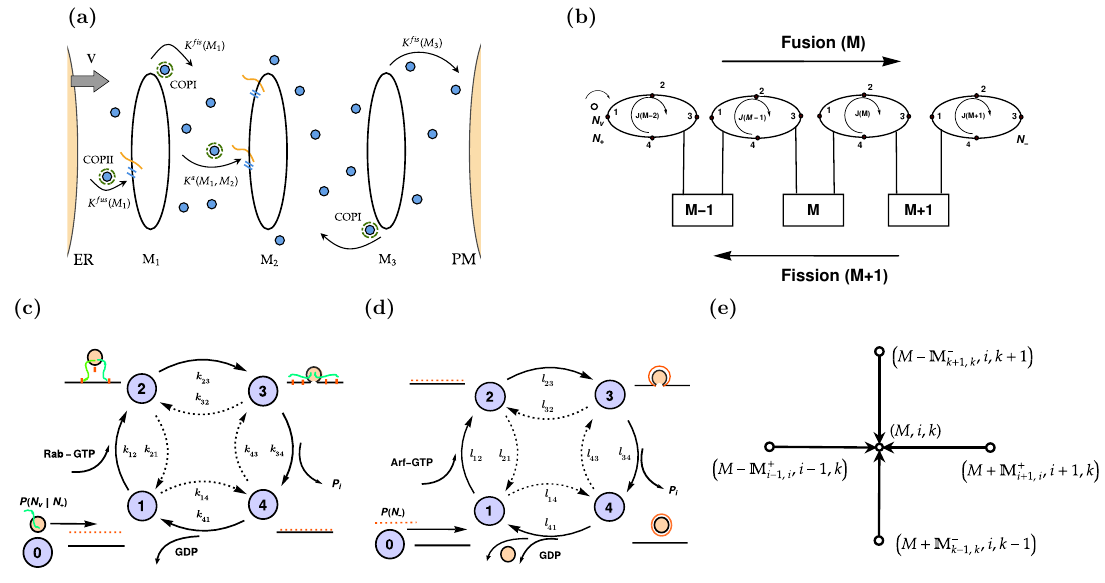}
\caption{{\bf The $n$-enzyme microscopic model for  master equation}: (a) Material flux through multiple cisternae from Endoplasmic reticulum (ER) to plasma membrane (PM). Vesicle transfer from ER to cis-Golgi is facilitated by COPII whereas intercisternal vesicle trasfer is facilitated by COPI protein \cite{glick1}. (b) When summed over fisogens, fusogens and the internal states of fusion-fission cycle, master equation in terms of cisternal size dependent cycle current is obtained, where the fusion cycle increases cisternal size and the fission cycle decreases cisternal size (see chemical master equation, Eq.\,\eqref{eq:Master_final_sup}). Each ellipse represents sum over the four states of fusion and fission cycles.
(c) Four-state Rab-GTPase activation fusion cycle. The input to this Markov cycle is a flux of vesicles (yellow circle) and GTP.  The  forward transitions are represented by full line-arrows, while the reverse arrows are shown as dashed line-arrows. The nonequilibrium cycle clearly supports a current, which results in the increase of cisternal size by one unit. 
(d) Four-state Arf-GTPase activation fission cycle. The input to this Markov cycle is a flux of GTP, while the corresponding output is a flux of GDP, inorganic phosphate P$_i$, and a vesicle. The current through the nonequilibrium cycle results in the decrease of cisternal size by one unit. (e) The chemical master equation graph showing transition into the state $(M,i,k)$ from all the neighborhood states \cite{qian}. For the fusion cycle, ``virtual'' mass, $\mathbb{M}^+_{ji}$ is added in forward transition ($j=i-1$) and subtracted in backward transition ($j=i+1$) and for the fission cycle, ``virtual'' mass, $\mathbb{M}^-_{lk}$ is subtracted for the forward transition ($l=k-1$) and added for the backward transition ($l=k+1$).}
\label{fig:golgi_fluxSI}
\end{figure}

\begin{table}[t!] 
\begin{center}

\centering
\begin{small}
\begin{tabular}{l c}
 \hline
 \hline
 \multicolumn{2}{c}{Microscopic parameters} \\
 \hline
 \hline
Surface area of a cell (fibroblast) & $\sim 1500 \; \mu \text{m}^2$ \cite{griffiths,sable}\\
Diameter of a transport vesicle & d = 60-100 nm ( $0.01 \; \mu \text{m}^2)$ \cite{Rexach,Rhiel}   \\
Diameter of a  Golgi cisternae   & $0.7-1.1\;  \mu \text{m}$  \cite{Klumperman} \\
Total surface area of Golgi   & $20-200\;  \mu \text{m}^2$  \cite{ladinsky, griffiths, Klumperman} \\
Surface area of a Golgi cisterna, $M$ & $1-20\;  \mu \text{m}^2$  \cite{marsh,ladinsky, griffiths, Klumperman} \\
Cell membrane lysis tension   & $1-10 \; Dyn/cm\;  (10^6 \; K_B T/\mu \text{m}^2)$ \cite{upa,sable} \\ 
Golgi cisterna membrane tension, $\gamma$   & $ 0.005 \; Dyn/cm \; (10^3 \; K_B T/\mu \text{m}^2)$ \cite{upa,sable} \\
ER membrane tension   & $0.013 \; Dyn/cm \; (3200 \; K_B T/\mu \text{m}^2)$ \cite{upa,sable} \\
Golgi cisterna rest tension, $\gamma_0$  & $5000 \; K_B T/\mu \text{m}^2$ \cite{upa,sable}\\
Dynamic range, $C_\gamma$  & $5000 \; K_B T/\mu \text{m}^2$ \\
Golgi membrane stretching modulus, $k_s$  & $10^2 \; K_B T/\mu \text{m}^4 \,=\,1 \;  K_B T/(\mu \text{m}^2 \, \text{vesicle})  $ \\
Markov transition rates, $k_{ij},\, l_{ij}$   & $1-1000/s$  \cite{arf1,allin} \\
No. of Fusion model parameters & $|\boldsymbol{\lambda}| = \, 8 + 4 + 3 = 15$ \\
No. of Fission model parameters & $ |\boldsymbol{\mu}| = \, 8 + 4 + 2 = 14$ \\
\hline
\hline
 \multicolumn{2}{c}{Effective parameters} \\
 \hline
 \hline
Golgi Recovery (formation) time & $ \tau \approx 60 \, min (5-10 \, min/\textrm{cisterna})$  \cite{langhans,patterson,dunlop}\\
Estimated avergae vesicle trafficking rate  & $ = \pi d^2 \tau/M \approx 0.1 - 1 \, s $\\
Dimensionless Influx rate &  $v=0-3$ \\
Dimensionless Exit rate &  $d_2 = 0-5$ \\
Dimensionless intercisternal flux rate &  $d_{12} = 0-5$ \\
Michaelis-Menten constant (fusion) &  $C_1 = 100$ \\
Michaelis-Menten constant (fission) &  $C_2 = 20-40$ \\
Macroscopic nucleation rate &  $a_0 = 0.05-0.1$ \\
Cisternal  progression speed  &  $ \approx 0.1 \, \textrm{stacks} \, min^{-1}$ \cite{dmitrieff} \\
\hline
\hline
\end{tabular}
\end{small}

\caption{{\bf Experimental and estimated values of parameters.} The time scale of update of cisternal size (milliseconds-seconds) \cite{arf1,allin} is fixed by the microscopic forward rates. Backward rates are fixed at $1$ unless stated otherwise. Total number of parameters for the fusion process is given by the total sum of kinetic parameters in $4$-state Markov cycle ($8$), Membrane parameters ($4$) and influx parameters ($3$) comprising fusion enzyme availability, vesicle availability and site availability. Similarly, total number of parameters for the fission process is given by the total sum of kinetic parameters in $4$-state Markov cycle ($8$), Membrane parameters ($4$) and influx parameters ($2$) comprising fission enzyme availability and site availability. Membrane parameters are -- Golgi cisterna rest tension ($\gamma_0$), Golgi cisterna rest size ($M_0$), Golgi membrane stretching modulus ($k_s$), and dynamic range ($C_\gamma$).}
\label{tab:tabl1}
\end{center}
\end{table}
\subsection{Physico-chemical basis for nonequilibrium flux kernels}
\label{subsec:markovcycle}
In our analysis, we consider a finite cisternal pool of fusogens and fisogens, set by the expression levels of the cell. This constraint can be imposed in different ways, e.g.,
\begin{enumerate}[label=(\roman*)]

\item \textit{Constant rate ensemble} -- In this ensemble, fusogens/fisogens  are assumed to be available at a constant rate. Specifically, the rate of fusogen/fisogen  arrival per unit time is a strict constant $r_0$, described by $P(r) = \delta(r-r_0)$. As the subunits arrive independently, the probability of observing $n$ arrivals within a given time window follows a Poisson distribution \cite{kampen}.


\item {\it Finite pool ensemble} -- In this ensemble, we assume that at any given time, the cell contains a finite total number of fusogens/fisogens, $N_{total}$. Consequently, each time an enzyme engages in a fusion, fission or transport event, one fewer fusion/fission enzyme remains available at the cisterna. If the size of the Golgi is $M$, the number of available free enzymes is dictated by a function $\mathcal{P}_s(M)$, which represents the number of independent active patches at the cisternal membrane.  The number of free enzymes available is, therefore, $m = N_{total} - \mathcal{P}_s(M)$. The probability of observing exactly $m$ free fusogens/fisogens (which governs the fluctuating arrival rate \cite{kampen}) follows a binomial distribution, given by
\begin{eqnarray}
P(m) = \binom{N_{total}}{m} \,p^m \,(1-p)^{N_{total}-m} \,,
\end{eqnarray}
where $p$ is the probability that a given enzyme is currently free in the pool rather than bound to the Golgi, determined by the specific binding and unbinding rates of the system \cite{banerjee,mukherjee}.

\end{enumerate}





Here, we choose to work with the constant rate ensemble. However, methodologies developed in this work, including the procedure for reduction to a lower dimensional system (see \ref{sec:si2}), readily generalize to other ensemble descriptions.

Now, to construct the fusion-fission flux kernels from  Eqs.\,\eqref{eq:current_expr_si1},\eqref{eq:current_expr_si2},\eqref{eq:Master_final_sup},  we would need the  steady state probabilities $P_{ss}(N\v,N_{\pm})$, the conditional probability in a given internal state $P_{ss}(i \vert M, N\v, N_{\pm}) $, and the transition rates to/from that state. As mentioned above, we will assume that fusion and fission processes are independent events and hence the distributions of fisogens and fusogens are also independent. We consider the fusion and fission flux kernels separately.

{\bf Fusion flux kernel}: 
We first assume that vesicle production and enzyme availability are independent Poisson processes. However, the availability of vesicles primed for fusion depends on the availability of fusion enzymes. Thus,
\begin{equation}
\label{eqn:independence}
P_{ss}(N\v,N_+) = P_{ss}(N\v|N_+) \,P_{ss}(N_+)\, .
\end{equation} 
To obtain $P_{ss}(N\v|N_+)$, we model the injection of vesicles from the ER through a  ``chemical'' process, $\ce{ER <=>[k_+][k_- N\v] v^{}} $, where $k_+$ is the effective rate of arrival of vesicles from the ER, $k_{-} \, N_\v$ is the rate of loss of vesicles and hence $k\v = k_+/k_-$ gives the Poisson rate of the process. 


Close to the cisterna, these secreted vesicles are primed to enter the fusion cycle by associating themselves with fusogens, such as Rab-proteins, which are kept in an inactive GDP-bound form in the cytoplasm by a GDI (GDP dissociation inhibitor) \cite{Alberts2008}. 
This state, the vesicle with cytoplasmic Rab-GDP and its effector proteins is designated as state $|0\rangle$ (see Fig.\ref{fig:golgi_fluxSI}(c)).
The probability of the vesicle being in state $|0\rangle$, which is conditioned on the availability of Rab-GDP (whose number is indexed by $N_+$), is given by,
 \begin{eqnarray}
 P_{ss}(N\v|N_+) =\frac{\binom{N_\v}{N_+}}{k_f^1 + \binom{N_\v}{N_+}}\exp(-k\v) \frac{(k\v)^{N_\v}}{N_\v!} \approx \exp(-k\v) \frac{(k\v)^{N_\v}}{N_\v!} \,.
   \label{eqn:independence1}
 \end{eqnarray}
Here, Rab-proteins lie in the vicinity of vesicles with effective rate $k_f^1$, the ratio of unbinding-binding rates of Rab-proteins to the vesicle. We work in the limit $N_\v \gg N_+$ and $\binom{N_\v}{N_+} \gg k_f^1$ and we can drop the combinatorial factor in Eq.\,\eqref{eqn:independence1}. As the probability of the vesicle being in state  $|0\rangle$ is conditioned on the availability of Rab-GDP, the sum over vesicle number $N_\v$ in Eq.\,\eqref{eq:Master_final_sup} starts from $N_+$. 
Furthermore, the steady state distribution of the fusogens (Rab-GDP), $P_{ss}(N_+)$ is derived from a Poisson process with rate $k_z$,
\begin{eqnarray}
P_{ss}(N_+) \approx \exp(-k_z) \frac{(k_z)^{N_+}}{N_+!} \,.
 \label{fuso_avail}
\end{eqnarray}
Here, $k_z$ is the Poisson rate of availability of fusion enzymes (fusogens), which sets  $N_+$. These fusion-competent vesicles in state $|0\rangle$ then enter the fusion  cycle, characterized here by a 4-state Markov process (see Fig.\ref{fig:golgi_fluxSI}(c)). 
 The states are denoted by $(\mathbb{X}, \mathbb{C})$, the chemical coordinate of the enzyme-vesicle complex and the configurational coordinate of the cisternal membrane. Thus, the state  $|1\rangle$ $\eqqcolon$ (X, C$_1$), 
where  X is the fusion competent vesicle together with the SNARE proteins \cite{Alberts2008} and C$_1$ is the undeformed cisternal membrane. 
The cytoplasmic Rab-GDP gets converted to a membrane bound Rab-GTP by a GEF (Guanine nucleotide exchange factor) and vice-versa by a GAP (GTPase activating protein). Upon binding to GTP via a GEF, state $|1\rangle$
 gets converted with rate $k_{12}$ to state $|2\rangle$ $\eqqcolon$ (X-GTP, C$_2$), which is membrane bound Rab-GTP + membrane inserted SNARE proteins + vesicle bound to cisternal membrane at specific functional domains where several factors -- including t-SNAREs, tethering proteins, and Rab GTPases -- work together to capture the vesicle \cite{loweg}, marking the  ``docking sites''  at the cisterna.
  These activated fusogens then drive the formation of a membrane fusion pore with rate $k_{23}$, while undergoing a conformation change - state $|3\rangle$ $\eqqcolon$ (X$^*$-GTP, C$_3$) -  enroute to complete fusion state $|4\rangle$ $\eqqcolon$ (X$^*$-GDP, C$_4$) with rate $k_{34}$ and the conversion to the inactivated Rab-GDP which then unbinds from the cisternal membrane, thus completing the fusion cycle and making the fusogens available for another round of fusion. The backward transition rates are slower, the energy consuming rates and the consequent breaking of detailed balance ensures that the enzymatic cycle predominantly proceeds in the forward direction.
 This is represented in the reaction scheme;
 
\begin{small} 
\begin{align*}
\begin{aligned}[c]
\ce{X + C_1 <=>[\ce{k_{12}}][\ce{k_{21}}] [X-GTP] + C_2} \\
\ce{[X-GTP] + C_2 <=>[k_{23}][k_{32}] [X^{*}-GTP] + C_3} \\
\ce{[X^{*}-GTP] + C_3  <=>[k_{34}][k_{43}] [X-GDP] + C_4} \\
\ce{[X-GDP] + C_4  <=>[k_{41}][k_{14}] X + C_1}
\end{aligned}
\end{align*}
\end{small}
On completion of this fusion cycle, lipids and lumenal material are pumped into the cisterna, and the cisternal size increases by one unit. The above description of the mechanical states of the cisternal membrane is, of course, a discrete approximation of a continuous membrane deformation and the transition rates depend on mechanical properties of the cisternal membrane such as membrane tension (or curvature) \cite{goud, sens}. Formally, this four-state fusion cycle can be represented by a continuous-time Markov chain (CTMC) \cite{Norris1997}. This is defined via the probability ${p}_i$ of the system being in the internal state $i$: $\dot{p}_i = \sum_{j=1}^{4} k_{ij} p_j$ with $\sum_{i=1}^{4} p_i = 1$ and a transition rate matrix (generator matrix) $k_{ij}$ \cite{Norris1997} given by
\begin{small} 
\begin{eqnarray}
  k_{ij} = \begin{pmatrix}
    -k_{12}(M) - k_{14} & k_{21} & 0 & k_{41}   \\
     k_{12}(M) & - k_{21}  - k_{23} & k_{32} & 0 \\
    0 & k_{23} & - k_{32}- k_{34}(M)  & k_{43}  \\
    k_{14} & 0 & k_{34} (M) & - k_{41}  -k_{43}  
\end{pmatrix} \,.
\label{eq:markov_states}
\end{eqnarray} 
\end{small} 
\begin{figure}[t!]
\centering 
\includegraphics[width=0.95\textwidth]{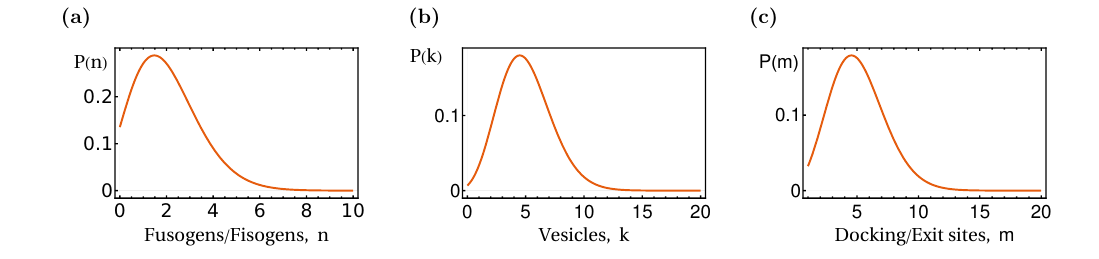} 
\caption{{\bf Poisson rates in  the microscopic model.} The availabilities of fusogens/fisogens, vesicles and docking/exit sites follow a Poisson process, see Eqs.\,\eqref{eqn:independence1},\eqref{fuso_avail},\eqref{eq:combinatorial_dock}. Poisson probabilities for (a) Number of fusogens/fisogens, $n$ (Poisson rate, $k_z=2$) (b) Number of vesicles, $k$ (Poisson rate, $k_\v=5$) (c) Number of docking/exit sites, $m$ (Poisson rate, $k_p \, M=5$).}
\label{fig:poisson_rates}
\end{figure}
Now, as per the master equation Eq.\,\eqref{eq:Master_final_sup}, we need to specify the dependence of the microscopic current on the number of fusion enzymes. As suggested above (and seen in Eq.\,\eqref{eq:markov_states}), some of the transition rates $k_{ij}$ depend on the cisternal size $M$. This arises from two main sources: the first one is the availability of fusogens and cisternal docking sites for the  fusion primed vesicles. For instance, the probability of being in state $|1\rangle$  given by Eq.\,\eqref{eqn:independence1} is conditioned on the availability of fusogens, with $k_z$ being the Poisson rate of finding fusogens. The transition rate $k_{12}$ is taken to be proportional to the  number of ``docking sites" $m$. Also, we have assumed that the probability of finding $m$ ``docking sites" on a cisterna of size $M$ for a given $k$ number of vesicles, $P_d(m|M,k)$, is given via a Poisson point process \cite{Kelly1979}, i.e., 
\begin{equation}
P_d(m|M,k)=  \frac{\binom{m}{k}}{k_f^2 + \binom{m}{k}} \,  \frac{e^{-k_p \, (M+1) } (k_p \, (M+1))^m}{m!} \,,
\label{eq:combinatorial_dock}
\end{equation}
effectively making the transition rate $k_{12}$ cisternal size dependent. Here, $k_f^2$ is the ratio of effective ``docking'' and ``undocking''  rates and the $(M+1)$ term makes sure that the cisterna can evolve from zero size (we can also use a threshold instead of $1$ here). As before, we work in the limit $m \gg k$ and $\binom{m}{k} \gg k_f^2$ and the combinatorial factor can be neglected (hence $k$ only sets the lower limit for the sum over $m$). This also imposes the constraint, $k_\v > k_z$ and $k_p \, M  > k_\v > k_z $ and as we will see, decides the saturation scale of the mean fusion kernel (see Eq.\,\eqref{eq:kernels_analyt_form2_full}). 

The other dependencies on cisternal size $M$ arise from a possible involvement of the mechano-chemistry of fusogens driving some of the transitions. The transition rates can 
depend on compositional properties (e.g., lipid specificity in phase segregated domains) and mechanical properties such as tension or curvature \cite{goud, sens1, sens,  dai, stone, bigay, has} of the cisternal membrane, the latter
 either directly via mechanochemistry or actuated via a chemical pathway~\cite{joseph}. For instance, the transition ($|3\rangle \rightarrow |4\rangle$), which involves flattening of budded vesicle is up-regulated by membrane tension (It is worth noting that other rates may also depend on mechnanochemical parameters. For example, in fusion cycle, transition  $|2\rangle \rightarrow |3\rangle$ will be directly proportional to membrane tension as it involves tethering proteins bringing the vesicle closer to the membrane \cite{Alberts2008,loweg}. Also, the transition $|4\rangle \rightarrow |1\rangle$, which involves decoupling of Rab-enzymes from  cisternae might be proportional to membrane tension. However, these are not considered here). Assuming mechanical properties remain constant throughout the cycle, this can be encoded in the rate $k_{34} = k_{34}^0 \, \frac{2\,\gamma}{C_\gamma + \gamma}$, $k_{34}^0$ is constant, which is the saturation kinetics of the tension $\gamma$ and  $C_\gamma$ sets the dynamic range for the varying tension, $ k_{34}^0$ is constant (Fig.\ref{fig:docks}(c)). In addition, we have taken $k_{43} = k_{34}^0$ to be a constant. We also assume all other rates to be constant. In the adiabatic limit, the cisternal membrane tension depends on the excess cisternal size (akin to excess membrane area \cite{evansT}) as, $\gamma = \gamma_0 \exp{(- (8 \pi \kappa)/(k_B T) \,\Delta M/M_0)}$, $M_0,\gamma_0$ are rest cisterna size and rest membrane tension of the cisterna and $\Delta M$ being positive for fusion events, leads to reduction in the membrane tension. We will work with the linear approximation, $\gamma_0 - k_s(M-M_0(V))$, $k_s$ is the stretching modulus of the cisternal membrane and $k_s = 1/M_0 \,\left(\gamma_0 \,(8 \pi \kappa)/(k_B T \,K_A)\right)$ where choosing $K_A$ fixes the validity and dynamic range of the linear approximation  \cite{evansT}. Since the Golgi cisterna is derived from the ER, we assume the mechanical parameters of the Golgi membrane are on the same order of magnitude as those of the ER (see Table \ref{tab:tabl1}).

With the above ingredients, we can compute the  nonequilibrium fusion flux kernel $K^{fus}$ by averaging the microscopic current in the fusion cycle over the availabilities of vesicles, fusogens and ``docking sites'' (see Eqs.\,\eqref{eq:Master_final_sup} and \eqref{eq:current_expr_si1}),
\begin{small}
\begin{equation}
    \label{eq:fusion_rate_k}
    K^{fus}(M, \boldsymbol{ \lambda})
    = \sum_{n=1}^{N_E} \sum_{k=1}^{N_E} \sum_{m=1}^{N_E} \,P(N_+=n) \, P(N\v = k|N_+) \,
    P_d(m|M,k) \,  J_{ss}^+(M,m)  \,,
\end{equation}
\end{small}
\begin{figure*}[t!]
\centering
\includegraphics[width=0.95\textwidth]{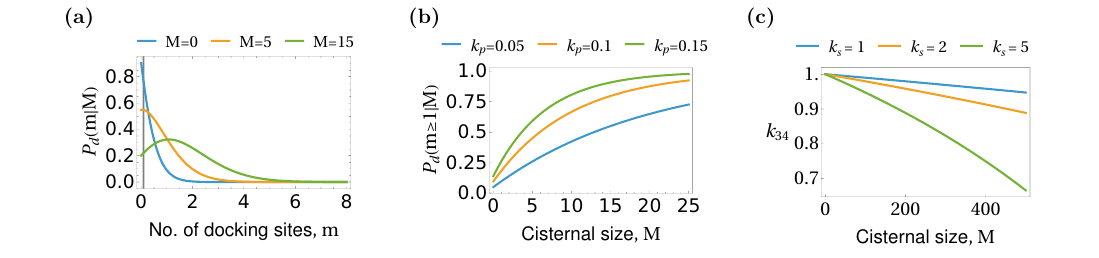} 
\caption{{\bf Cisternal size dependencies of microscopic rates.} Probability of finding a ``docking/exit site'' given cisternal size $M$, $P_d(m|M)$ increases as the cisternal size is increased, see Eq.\,\eqref{eq:combinatorial_dock}. It is plotted (a) as a function of number of docking sites $m$ for different cisternal sizes, $M$ and (b) as a function of cisternal size, $M$ for different site availability rates, $k_p$. (c) Flattening rate for fusion, $k_{34}$ as a function of cisternal size $M$. $k_{34} = k_{34}^0 \, \frac{2\,\gamma}{C_\gamma + \gamma}$, and $\gamma=\gamma_0 - k_s \, M$, where $k_{34}^0$ is constant and $k_s$ is the stretching modulus
of the cisternal membrane. Parameter values, $\gamma_0=C_\gamma= 5000 \; K_B T/\mu \text{m}^2$,  $k_s=1$, $k_{34}^0=1$. The linear regime is set by the dynamic range $C_\gamma$, beyond which the membrane tension effects become relevant. Flattening rate for fission,  $l_{32} = l_{32}^0 \frac{2\,\gamma}{C_\gamma + \gamma}$ as a function of cisternal size $M$ is also given by a similar curve.} 
\label{fig:docks}
\end{figure*}
where $N_E$ is the chosen upper limit for the summation beyond which the availabilities of fusogens, vesicles and the ``docking sites'' are zero. The steady-state current through the fission cycle can be calculated from any of the reaction edges (\textit{flow conservation} at each node), e.g. at the edge ($|4\rangle \rightarrow |1\rangle$),
\bea
J_{ss}^+(M,m) = \left (k_{41} P_{ss}^+ (4\vert M,m,\boldsymbol{\lambda}) - k_{14} P_{ss}^+ (1\vert M,m,\boldsymbol{\lambda}\right) \,,
\eea
$\boldsymbol{\lambda}$ is the vector of parameters that includes the rates in the transition matrix, rates that appear in the Poisson processes governing the availability of the vesicles, fusogens and ``docking sites'', and the mechanochemical parameters (see Table \ref{tab:tabl1}).

{\bf Fission flux kernel}: 
The analysis of the fission cycle follows along similar lines (see Fig.\,\ref{fig:golgi_fluxSI}(d)).  
The fission process is catalysed by fisogens such as Arf-proteins, a GTPase enzyme cycle, which we  describe by a 4-state Markov process \cite{QianE, Norris1997}. 
The external state $|0\rangle$ 
 is a sequestration of fisogens (dashed red line) proximal to the cisternal membrane exit site. 
  State  $|1\rangle$  $\eqqcolon$ (Y, C$_1$) represents Arf-GDP-fisogen and the undeformed cisternal membrane. Upon binding to GTP, $|1\rangle$ $\to$ $|2\rangle$ 
$\eqqcolon$ (Y-GTP, C$_2$), which represents Arf-GTP-fisogens bound to a budding cisternal membrane at specialized domains where molecular machinery -- such as Arf-GTPase, GEFs, and coat proteins -- concentrate to deform the lipid bilayer  \cite{mal1,weiland}, marking the ``exit/budding sites'' on the cisterna.
This drives the formation of a complete bud, while undergoing a conformation change - $|3\rangle$  $\eqqcolon$ (Y$^*$-GTP, C$_3$) - enroute to complete fission
   $|4\rangle$  $\eqqcolon$ (Y-GDP, C$_4$), and the conversion to the inactive GDP-bound form. The Y-GDP then unbinds from the vesiculated bud, and returns to $|1\rangle$ and then to state $|0\rangle$. It can be represented by the reaction scheme,
   
 \begin{small}  
\begin{align*}
\begin{aligned}[c]
\ce{Y + C_1 <=>[\ce{l_{12}}][\ce{l_{21}}] [Y-GTP] + C_2} \\
\ce{[Y-GTP] + C_2 <=>[l_{23}][l_{32}] [Y^{*}-GTP] + C_3} \\
\ce{[Y^{*}-GTP] + C_3  <=>[l_{34}][l_{43}] [Y-GDP] +  C_4} \\
\ce{[Y-GDP] + C_4  <=>[l_{41}][l_{14}] Y + C_1}
\end{aligned}
\end{align*}
\end{small}  
At the end of this fission cycle, the cisternal size decreases by one unit. The  rate $l_{12}$ is a function of availability of cisternal  ``exit (budding) sites'', which is generated via a Poisson point process \cite{Kelly1979} with probability  $P_b(m|M)$  and increases with increasing $M$ before it saturates (as before, we have neglected the combinatorial factor, see Eq.\eqref{eqn:independence1}). Since the budding process ($|2\rangle \rightarrow |3\rangle$) has to work against the membrane tension, the transition rate  $l_{32}$ is facilitated by the membrane tension, $l_{32} = l_{32}^0 \frac{2\,\gamma}{C_\gamma + \gamma}$, $C_\gamma$ sets the saturation scale. In addition, we have taken $l_{23} = l_{32}^0$ to be a constant. We also assume all other rates to be constant.
The fission kernel can be computed by averaging the microscopic current in the fission cycle over the availabilities of fisogens and ``exit sites'' (see Eqs.\,\eqref{eq:current_expr_si2},\eqref{eq:Master_final_sup}),
\begin{equation}
    \label{eq:fission_rate_k}
    K^{fis}(M,\boldsymbol{\mu}) = \sum_{m=1}^{N_E}\sum_{n=1}^{N_E}  \; P(N_- = n)  \, P_b(m|M,n) \, J_{ss}^-(M,m)\,,
\end{equation}
where $N_E$ is a chosen upper limit for the summation beyond which the availabilities of the fisogens and the ``exit sites'' are zero. Also,
\bea
J_{ss}^- = l_{41} P_{ss}^- (4\vert M,m,\boldsymbol{\mu}) - l_{14} P_{ss}^- (1\vert M,m,\boldsymbol{\mu})
\eea
is the steady state current through the fission cycle, and $\boldsymbol{ \mu}$ is the vector of parameters that includes the rates in the transition matrix, rates that appear in the Poisson processes governing the probability of budding, availability of fisogens and ``exit sites'', and the mechano-chemical parameters (see Table \ref{tab:tabl1}).
\subsection{Dynamics of the mean cisternal size}
\label{sec:SiMeanDy}
\label{sec:SiMeanDy}
We derive the dynamical equation for the mean cisternal size. Eq.\,\eqref{eq:Master_final_sup} can be written in the following concise form,
\begin{small}
\begin{eqnarray}
\frac{dP(M,t)}{dt} &=&  \sum_{n_1} \left( K^{fus}(M-n_1,n_1)\,P(M-n_1,t) - K^{fus}(M,n_1) \, P(M,t) \right) \nonumber \\ &&+ \sum_{n_2}  \left( K^{fis}(M+n_2,n_2)\,P(M+n_2,t) - K^{fis}(M,n_2) \,P(M,t) \right)  \nonumber  \\
&=&   \sum_{n_1} (E^{-n_1} - 1)K^{fus}(M,n_1) \,P(M,t)  +  \sum_{n_2}  (E^{n_2} - 1) K^{fis}(M,n_2) \,P(M,t) \,, \nonumber \\
\end{eqnarray}
\end{small}
where in the last line, we have used the step operator $E$ \cite{kampen}. The mean dynamics for this Master equation can be obtained using the generating function $G(s,t) = \sum_{M} s^M \, P(M,t)$, which for the fission term reads
\begin{small}
\begin{eqnarray*}
&& \frac{\partial G}{\partial  t}\bigg|_{fission} =  \sum_{M} s^M \,\frac{dP(M,t)}{d t} \nonumber \\ &&= \sum_{M} s^M \,\sum_{n_2}  \,\left( K^{fis}(M+n_2,n_2)\,P(M+n_2,t) - K^{fis}(M,n_2) \,P(M,t)\right)   \nonumber \\
&=&  \sum_{n_2}  s^{-n_2} \,\sum_{M}\, s^{M+n_2}  \,\left( K^{fis}(M+n_2,n_2)\,P(M+n_2,t) \right) - \sum_{n_2} \, \sum_{M} \,s^{M} \, \left(K^{fis}(M,n_2)\, P(M,t)\right) \\
&=& \sum_{n_2} \,   (s^{-n_2}-1) \, \sum_{M} \,s^M \,\left(K^{fis}(M,n_2)\, P(M,t)\right) \,.
\end{eqnarray*}
\end{small}
Following similar steps for the fusion term and adding them up gives,
\begin{small}
\begin{eqnarray}
\frac{\partial G}{\partial t} &=& \sum_{n_2}  (s^{-n_2}-1) \,\sum_{M} s^M \left(K^{fis}(M,n_2)\, P(M,t)\right) +  \sum_{n_1}  (s^{n_1}-1) \,\sum_{M} s^M \left(K^{fus}(M,n_1) \,P(M,t)\right) \,. \nonumber \\
\label{eq:G_part}
\end{eqnarray}
\end{small}
Taking the partial derivative  of the  above equation Eq.\,\eqref{eq:G_part} w.r.t. the dummy variable, $s$ and setting $s=1$ in the resulting equation gives the mean dynamics,
\begin{small}
\begin{eqnarray}
\frac{d \overline{M}}{d  t} &=& - \sum_{n_2} \sum_{M} \, n_2\, \left(K^{fis}(M,n_2)\, P(M,t)\right) +  \sum_{n_1} \sum_{M} \,n_1 \,\left(K^{fus}(M,n_1)\, P(M,t)\right)  \nonumber \\
&=& - \sum_{n_2} \, n_2 \,\overline{ K^{fis}(M,n_2) } +  \sum_{n_1} \,n_1 \,\overline{ K^{fus}(M,n_1) } \,.
\end{eqnarray}
Assuming mean field decoupling, where the fluctuations are neglected,
\begin{eqnarray}
\frac{d \overline{M}}{d  t} = - \sum_{n_2} \, n_2 \, K^{fis}(\overline{M} ,n_2) +  \sum_{n_1} \,n_1  \,K^{fus}(\overline{M} ,n_1)  \,,
\label{eq:mean_eq_N}
\end{eqnarray}
\end{small}
which gives the expression for mean fusion and fission kernels mentioned in the main text. In terms of the microscopic parameters, it can be written as

\begin{small}
\begin{eqnarray}
{\cal K}^{fus}(\overline{M}) &=& \sum_{m,k,n=1}^{N_E} n \, \frac{e^{-k_z} (k_z)^n}{n!}\, \frac{\binom{k}{n}}{k_f^1 + \binom{k}{n}} \, \frac{e^{-k_\v} (k_\v)^k}{k!} \,  \frac{\binom{m}{k}}{k_f^2 + \binom{m}{k}} \, \frac{e^{-k_p \, (\overline{M}+1) } (k_p \, (\overline{M}+1))^m}{m!} \, J_{ss}^{+}(\overline{M},m) \label{eq:kernels_analyt_form1} \\
{\cal K}^{fis}(\overline{M}) &=& \sum_{m,n=1}^{N_E} n \, \frac{e^{-k_z} (k_z)^n}{n!}\, \frac{\binom{m}{n}}{k_f^2 + \binom{m}{n}} \frac{e^{k_p\, \overline{M}} (k_p\, \overline{M})^m}{m!} \, J_{ss}^{-} (\overline{M},m) \,,
\label{eq:kernels_analyt_form2}
\end{eqnarray}
\end{small}
where for clarity, we use $n$ to denote the availability of both fusion and fission enzymes, and hence the number of fusion and fission cycles (denoted by $n_1,n_2$ in Eq.\,\eqref{eq:mean_eq_N}). As mentioned above, we have neglected the combinatorial factors in the main text (see Eqs.\,\eqref{eqn:independence1},\eqref{eq:combinatorial_dock} and the subsequent discussion). 
\section{Reduction to a low dimensional dynamical system}
\label{sec:si2}
Given that the master equation Eq.\,\eqref{eq:Master_final_sup} involves a high-dimensional parameter space (see Eqs.\,\eqref{eq:fusion_rate_k},\eqref{eq:fission_rate_k}) and features highly non-linear terms, it becomes analytically intractable \cite{grima}. Hence, extracting a macroscopic understanding from it can be arduous. We address this issue using three distinct methods that yield consistent results: (i) a graphical method where we look at the fixed points of the mean field equations, (ii) a phase diagram based method where the phase diagram obtained from the microscopic model consisting of the solution classes of the master equation Eq.\,\eqref{eq:Master_final_sup} is used to confirm or invalidate a plausible macroscopic model. 
(iii) an algebraic method in which we construct an effective function (of cisternal size as well as of fusogens/fisogens availability) from different parts of the sum in Eqs.\,\eqref{eq:kernels_analyt_form1},\eqref{eq:kernels_analyt_form2}.

\subsection{Dimensional reduction through parsimony}
\label{subsec:dimredux}
Eqs.\,\eqref{eq:fusion_rate_k},\eqref{eq:fission_rate_k} provide us with explicit but cumbersome formulae for computing fusion-fission kernels.
We compute the current over fusion and fission cycles and then plug it into the stochastic Master equation
Eq.\,\eqref{eq:Master_final_sup}. Taking the first moment of the resulting equation, and using a {\it mean field decoupling}, we derive Eqs.\,\eqref{eq:kernels_analyt_form1},\eqref{eq:kernels_analyt_form2} and find that the mean cisternal size, for which we continue to use the symbol $M$ satisfies the equation (see Eq.\,\eqref{eq:mean_eq_N}),
\begin{equation}
\label{eqn:markov_mean}
\frac{dM}{dt}  =  \; \; {\cal K}^{fus}(M, \boldsymbol{\lambda}) - {\cal K}^{fis}(M, \boldsymbol{\mu}) \,,
\end{equation}
where ${\cal K}^{fus/fis}$ now has dimensions of a mass current. This is similar to state-dependent assembly-disassembly rates in the case of microtubule formation, discussed in \cite{Mohapatra2016}. 
This represents a nonlinear dynamical system for $M$, with the  microscopic fusion flux kernel described by 
 $15$ model parameters and the microscopic fission flux kernel described by $14$ model parameters (Table \ref{tab:tabl1}) - a total of $29$ parameters. To reduce this high dimensional parameter space to a low dimensional one, 
 it suffices to understand the generic form of  ${\cal K}^{fus}$ and ${\cal K}^{fis}$, their asymptotic behaviour and the nature of intersections when plotted against 
 $M$. Taking into consideration the limiting values of the fusogens and fisogens, the availability of fusion-fission sites (which scale with $M$) and the dependence of the fusion-fission rates on membrane tension~\cite{dai, goud}, we expect that the microscopic fusion and fission kernels 
 are increasing functions of $M$ and must asymptote to a constant.
 
This is verified by explicitly plotting the kernels with physically reasonable choices of microscopic parameter values. The summation in Eqs.\,\eqref{eq:kernels_analyt_form1},\eqref{eq:kernels_analyt_form2}, can be carried out numerically and mean fission and fusion kernels can be plotted.
Fig.\,\ref{fig:golgi_flux_123} illustrates various intersections of these mean microscopic kernels as one varies the kinetic parameters, giving rise  to stable and unstable fixed points.  We observe that for a wide range of physical parameters, fission and fusion kernels have zero, one, two or three intersections. Varying microscopic rates changes number of stable (filled circle) and unstable (open circle) fixed points. Assuming enzyme kinetics, such a behaviour can be approximated to arbitrary accuracy by the low dimensional forms

\begin{figure*}[t!]
\centering
\includegraphics[width=\textwidth]{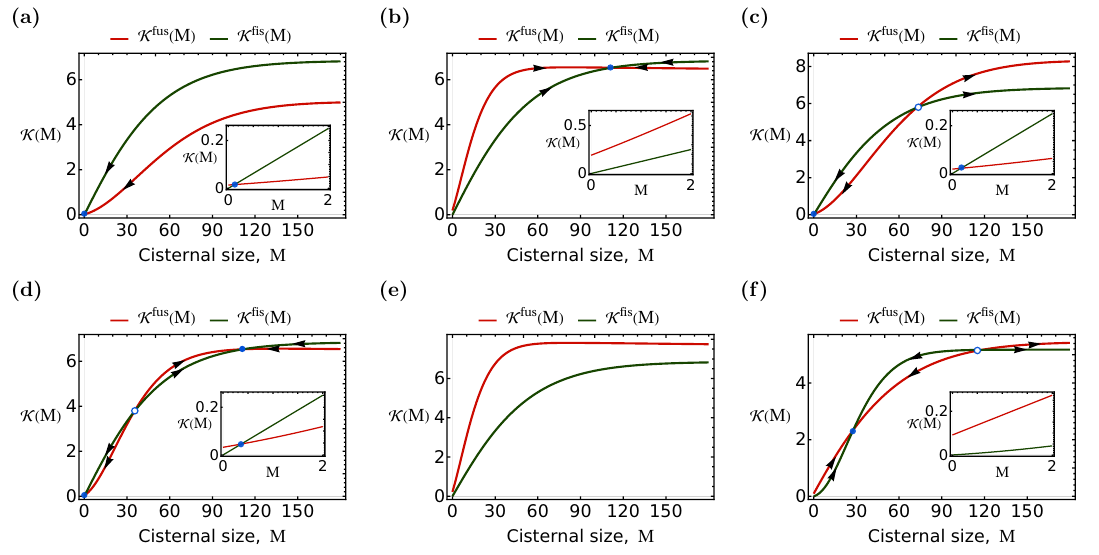}
\caption{{\bf Intersections of mean fusion and fission kernels.} (a-e) Mean fission and fusion flux kernels $\mathcal{K}^{fus}(M)$ and $\mathcal{K}^{fis}(M)$, respectively (together represented as $\mathcal{K}(M)$) can be obtained by computing the mean field equations  (see \ref{sec:SiMeanDy}, \ref{subsec:dimredux} and  Eqs.\,\eqref{eq:kernels_analyt_form1},\eqref{eq:kernels_analyt_form2}). Parameter values are in the range, $(k_p)_{fus}=0.05-0.1$, $(k_p)_{fis} = 0.1$, $(k_z)_{fus}=1-4$, $(k_z)_{fis} = 1$, $l_{23}^0 =1-10$ , $k_{34}^0 =1-10 $, $\gamma_0 =5000 \; K_B T/\mu \text{m}^2 $, $C_\gamma = 5000 \; K_B T/\mu \text{m}^2 $, see Table \ref{tab:tabl1} for notational details. Varying microscopic parameters in Eqs.\,\eqref{eq:kernels_analyt_form1},\eqref{eq:kernels_analyt_form2}, one gets the scenarios -- where mean fission and fusion kernels have one (a,b) or three (d) intersections for $(k_p)_{fus}>(k_p)_{fis}$ or two (c) intersections for $(k_p)_{fus}\sim(k_p)_{fis}$ or zero (e) intersections. Changing the microscopic parameters changes the number of stable (filled circle) and unstable (open circle) fixed points. This result can be used to map the microscopic model, Eq.\,\eqref{eq:Master_final_sup} to a dynamical system problem Eqs.\,\eqref{eq:Rfus1},\eqref{eq:Rfis1}. (a-e, insets) The intersections of mean fission and fusion  kernels near $M=0$. (f) If the fission events are more cooperative than the fusion events, i.e., $(k_p)_{fis}>(k_p)_{fus}$, the largest positive root is an unstable one, and the system grows unboundely. It can't ensure stable biogenesis for generic parameters (see \ref{sec:single_cist_si}). (inset) The mean fission and fusion kernels has no intersections near $M=0$.} 
\label{fig:golgi_flux_123}
\end{figure*}

\begin{eqnarray}
\label{eq:Rfus1}
{\cal K}^{fus} &=& v  \left( a_0 + \frac{M^\alpha}{C_1 +  M^\alpha} \right)  \\
{\cal K}^{fis} &=& \frac{d\,M^\beta }{C_2 +  M^\beta}  \,.
\label{eq:Rfis1}
\end{eqnarray}

The fusion flux kernel, is written in terms of {\it influx rate}, $v$ and {\it effective nucleation constant}, $a_0$ while the fission flux kernel is written in terms of the {\it effective peak fission rate}, $d$.
The parameters $C_1, C_2$ represent the levels of fusogens and fisogens (or {\it Hill saturation constants} for fusion and fission), respectively, while the Hill-exponents $\alpha,\beta$, characterize the {\it cooperativity} of the fusion and fission processes, respectively.  This Michaelis-Menten (and Hill) form of the fusion-fission kernels has been used in previous studies~\cite{Sachdeva2016,Vagne2018}. The dimensional reduction strategy outlined here (from $29$ parameters to $7$), and extended to many cisternae, makes the form of the fusion-fission kernels tractable, and leads to a low dimensional dynamical system for the sizes (masses) of the individual cisternae. \\

\noindent \textit{Cooperativity in fusion and fission processes}: Note that in Fig.\,\ref{fig:golgi_flux_123}(a-e), while plotting the mean fission and fusion kernels, we have taken the  microscopic rates $(k_p)_{fus}>(k_p)_{fis}$ and  $(k_z)_{fus}>(k_z)_{fis}$, which leads to the fusion events being more cooperative than the fission events. In the following discussion, we propose an argument for this.

The cooperativity in fusion and fission processes can arise through the following mechanisms:
\begin{enumerate}[label=(\roman*)]
\item {\it Enzyme allostery}: It is molecular and structural in nature. For example, Arf-GTP recruits the COPI coat subunits and they exhibit lateral protein-protein cooperativity. The initial recruitment of a coatomer provides binding interfaces that thermodynamically enhance the recruitment of adjacent subunits \cite{weiland} via elastic interactions.
\item {\it Substrate allostery}: The cisternal membrane undergoes localized physical alterations that lower the energy barrier for subsequent events. For example, the initial binding of COPI coatomers induces local membrane bending. This locally curved lipid membrane physically matches the inherent curvature of the COPI complex, acting as a highly favorable physical "substrate" that thermodynamically enhances the recruitment of subsequent coat subunits to continue the invagination \cite{goud}.
\item {\it Enhanced availability due to a prior proximal event}: This is an instance of dynamic cooperativity. For example, in a Rab cascade, membrane-bound Rab GTPases recruit GEFs to activate more Rabs \cite{Stenmark2009}. This cascading expansion recruits tethers that cluster resident t-SNAREs together, creating a hyper-reactive patch that ensures vesicle fusion \cite{Jahn2006, Alberts2008}.
\item  {\it Trans-membrane cooperativity}: Unlike fission, which occurs on a single continuous membrane, fusion requires the cooperative coupling of enzymes on the incoming vesicle (v-SNAREs) and the cisternal membrane (t-SNAREs) \cite{Jahn2006}.
\end{enumerate}

Now, mechanisms (i) and  (ii) occur during the $4$-state fission cycle and mechanism (iv) occurs during the $4$-state fusion cycle. Mechanism (iii) is generic and related to the initiation of fusion and fission cycles and we will concentrate on it here.

As mentioned before, the initiation of fusion is an enzymatically driven, cooperative process where  membrane-bound Rab GTPases recruit GEFs to activate more Rabs \cite{Stenmark2009}. The recruitment rate is proportional to the cisterna size $M$, but subject to the depletion of cytosolic resources.

Let $P_0$ represent an inactive (cytosolic) Rab protein in GDP-bound form and $P_1$ represent an active (membrane-bound) Rab protein in GTP-bound form. A bound Rab protein recruits more membrane-bound Rab proteins \cite{Jahn2006, Alberts2008} and this process can be represented by the following chemical reactions:
\begin{align}
    \text{Basal Activation:} \quad & P_0 \xrightarrow{k_1 I_{\text{eff}}(M)} P_1 \nonumber \\
    \text{Positive Feedback} \quad & P_0 + P_1 \xrightarrow{k_2 I_{\text{eff}}(M)} 2P_1 \nonumber\\
    \text{Basal Decay:} \quad & P_1 \xrightarrow{k_{-1}} P_0    \,,
     \label{eq:cascade}
\end{align}
where we have defined the effective activation function $I_{\text{eff}}(M)$, assumed to have a  Michaelis-Menten saturation form, $ I_{\text{eff}}(M) = I_{\text{max}} \left( \frac{M}{K_I + M} \right)$, and  represents the availability of upstream activators such as Rab-GEFs \cite{Stenmark2009} that subsequently drive the clustering of resident t-SNAREs on the cisternal membrane \cite{Jahn2006}. Let $N$ be the total number of Rab proteins (assumed to be a constant)  and $n$ be the number of active Rab proteins ($P_1$). Applying the law of mass action to the above reaction scheme \cite{kampen}, the activation propensity $W_+(n)$ combines both basal activation and the positive feedback:
\begin{align}
    W_+(n) = \Big[ k_1 I_{\text{eff}}(M) + k_2 I_{\text{eff}}(M) n \Big] \left( N - n \right) \hspace{1cm}
    W_-(n) = k_{-1} n \,,
\end{align}
which gives the master equation for the  time evolution of the probability $P(n,t)$ of having exactly $n$ active Rabs:
\begin{equation}
    \frac{\partial P(n,t)}{\partial t} = W_+(n-1) P(n-1, t) + W_-(n+1) P(n+1, t) - \Big[ W_+(n) + W_-(n) \Big] P(n, t) \,.
\end{equation}
Assuming mean field decoupling, we can write the equation for the mean fraction ($x = \langle n \rangle / N$) for the active Rabs, 
\begin{equation}
    \frac{dx}{dt} = k_1 \, I_{\text{eff}}(M) \,(1 - x) + k_2 \, I_{\text{eff}}(M) \, N \, x \,(1 - x) - k_{-1} \, x \,.
    \label{eq:mean_fus_coop}
\end{equation}
\begin{figure*}[t!]
\centering
\includegraphics[width=\textwidth]{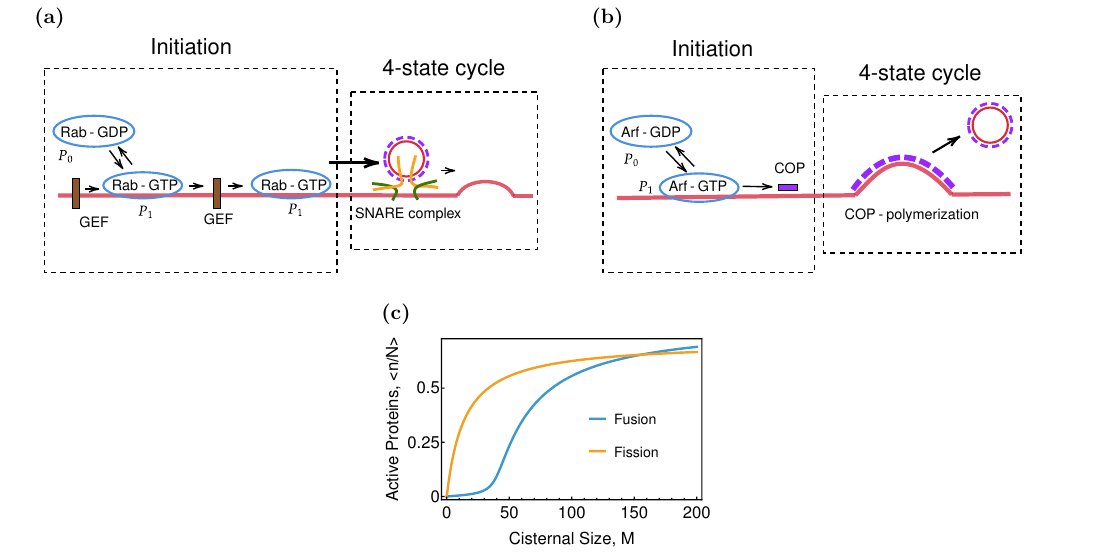}
\caption{{\bf Recruitment of fusogens and fisogens on the cisternal membrane.} (a) Initiation of fusion at the target cisterna begins with the basal conversion of inactive cytosolic Rab-GDP ($P_0$) to active membrane-bound Rab-GTP ($P_1$) via a resident GEF \cite{Alberts2008}. This membrane-bound Rab ($P_1$) subsequently recruits additional GEFs to activate more cytosolic Rabs, establishing a highly cooperative positive feedback loop known as a Rab cascade \cite{Stenmark2009}, see Eq.\,\eqref{eq:cascade}. This cascade ultimately scaffolds the trans-membrane SNARE complexes, which drive the fusion of the incoming vesicle with the target membrane \cite{Jahn2006}. (b) In contrast, the initiation of fission relies on the independent anchoring of Arf proteins to the membrane \cite{DSouzaSchorey2006}, lacking a  positive feedback loop,  see Eq.\,\eqref{eq:fis_coop}. Here, $P_0$ represents inactive cytosolic Arf-GDP and $P_1$ represents active membrane-anchored Arf-GTP. The active Arf anchors ($P_1$) subsequently recruit COPI coat proteins. The lateral polymerization of these coat proteins drives the physical budding and eventual fission of the membrane, which occur temporally downstream as a $4$-state fission cycle and are therefore not included in this initiation model. (c) Plot showing the mean fraction of active Rab-proteins for fusion (blue) and the mean fraction of active Arf-proteins for fission (yellow) in the steady-state as a function of cisternal size, $M$ (see Eqs.\,\eqref{eq:mean_fus_coop},\eqref{eq:mean_fis_coop}). It clearly indicates the initiation (which involves recruitment of fusion or fission enzyme on the cisternal membrane) of fusion process is more cooperative than that of the fission process.}
\label{fig:Master_cooperativity}
\end{figure*}
Now, unlike fusion, the initiation of fission does not feature enzymatic self-recruitment \cite{DSouzaSchorey2006}. Arf anchors to the membrane and later recruits COPI coat-proteins which we consider to be part of the $4$-state fission cycle, which occurs subsequently and is not part of the initiation process. Therefore, initiation of fission lacks the positive feedback. As before, let $P_0$ represent an inactive Arf protein in GDP-bound form and $P_1$ represent an active, membrane-anchored Arf-protein in  GTP-bound form, giving the reaction scheme:
\begin{align}
    \text{Binding:} \quad & P_0 \xrightarrow{k_1 I(M)} P_1 \nonumber\\
    \text{Basal Decay:} \quad & P_1 \xrightarrow{k_{-1}} P_0 \,,
         \label{eq:fis_coop}
\end{align}
with activation function $I(M)$. Let $N$ be the total number of Arf-proteins available to the membrane and $n$ be the number of active Arf-proteins ($P_1$). For this reaction scheme, propensities are given by
\begin{align}
    W_+(n) = k_1 I(M) \Big( N- n \Big) \hspace{1cm}
    W_-(n) = k_{-1} n \,.
\end{align}
The master equation describing the time evolution of the probability $P(n,t)$ of having $n$ number of active Arf proteins is given by
\begin{align}
    \frac{\partial P(n,t)}{\partial t} = k_1 I(M) \Big( N - (n-1) \Big) P(n-1, t) + k_{-1} (n+1) P(n+1, t) 
   - \Big[ k_1 I(M) \Big( N- n \Big) + k_{-1} n \Big] P(n, t) \,,
\end{align}
and the mean fraction ($x = \langle n \rangle / N$) of active Arf-proteins is given by
\begin{equation}
    \frac{dx}{dt} = k_1 I(M) (1 - x) - k_{-1} x \,.
    \label{eq:mean_fis_coop}
\end{equation}

In Fig.\,\ref{fig:Master_cooperativity}(c), we plot the steady-state mean fraction number of active Rab-proteins and Arf-proteins as a function of cisternal size. It clearly indicates fusion being more cooperative than fission. We will see  in \ref{sec:single_cist_si} that this is a necessary condition for de novo biogenesis \cite{Misteli2001}.

We can further argue that fusion is functionally more cooperative because it benefits from both dynamic cooperativity (Rab cascades \cite{Stenmark2009}) and robust inter-membrane scaffolding (trans-SNARE cooperativity \cite{Jahn2006}). This continuous expansion is only terminated by delayed Rab-GAP feedback  \cite{Stenmark2009}.
In contrast, fission is geometrically self-limiting; although COPI assembly exhibits initial substrate allostery, the mature bud's extreme curvature acts as a strict physical timer, triggering Arf-GAP1 to abruptly halt polymerization \cite{bigay,DSouzaSchorey2006,weiland}.

\begin{figure}[t!]
\centering
 \includegraphics[width=\textwidth]{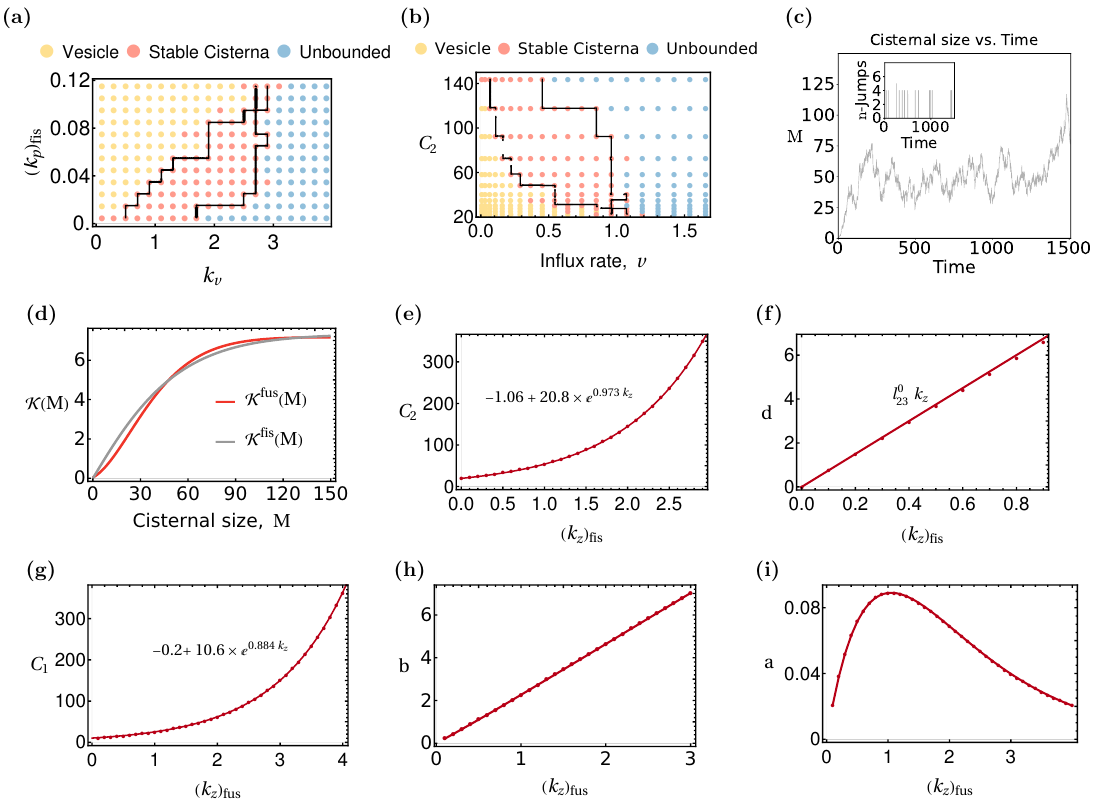}
\caption{{\bf Microscopic parameters to rates in dynamical system.} (a) Phase diagram obtained using stochastic simulations from $n$-enzyme master equation Eq.\,\eqref{eq:Master_final_sup} for single cisterna in ($k_{\v},(k_p)_{fis}$) plane, $k_{\v}$ and $(k_p)_{fis}$ are the Poisson rates of availability for vesicles and fisogens, respectively. Parameter values, $(k_p)_{fus} = 0.2$, $(k_z)_{fus} = 4$, $(k_z)_{fis} = 1$, $l_{23}^0 =3$, $k_{34}^0 =3$, see Eqs.\,\eqref{eq:kernels_analyt_form1},\eqref{eq:kernels_analyt_form2} and Table \ref{tab:tabl1} for notational details. (b) The phase diagram in the microscopic parameter space in (a) is projected to the macroscopic parameter space, ($v$, $C_2$) plane, where $v$ is the influx rate  and $C_2$ is the Hill saturation constant for fission,  using the mapping Eq.\,\eqref{eq:mapm}. 
(c) Stochastic trajectory for single cisterna obtained from $n$-enzyme master equation Eq.\,\eqref{eq:Master_final_sup}. Parameter values, $k_\v = 1$, $(k_p)_{fus} = 0.1$, $(k_p)_{fis} = 0.05$,  $(k_z)_{fus} = 1$, $(k_z)_{fis} = 1$, $l_{23}^0 =3$, $k_{34}^0 =6$. The inset shows the number of events (jumps) in a time step in Gillespie simulation. (d) Cisternal size-dependent fusion and fission rates rates for the reduced dynamical system Eq.\,\eqref{eq:veq_H}, $\mathcal{K}^{fus}(M),\,\mathcal{K}^{fis}(M)$, respectively (together represented as $\mathcal{K}(M)$) intersect each other thrice for the parameter regime used. (e-i) Parameters in the reduced dynamical system model can be estimated from the microscopic parameters. Taking the Poisson rate of finding the fusogens and fisogens, $(k_z)_{fus}$ and $(k_z)_{fis}$, respectively as independent parameters and using the mapping Eq.\,\eqref{eq:mapm}, one can compute the exit rate ($d = k^{fis}_m \kappa_2$), Hill saturation constant for fission ($C_2$), Hill saturation constant for fusion ($C_1$), Fusion rate ($b = v_m^{fus} \, \kappa_1$) and  nucleation rate ($a = a_0^{fus} \, v_m^{fus}$). $C_1,\,b$  are increasing functions of availability of fusogens, $(k_z)_{fus}$ and and $C_2,\,d$ are increasing functions of availability of fisogens, $(k_z)_{fis}$. Since the nucleation rate $a$ depends on the availability of ``docking'' sites for zero cisternal size, it first increases and then decreases for larger Poisson rate of availability of fusogens, $(k_z)_{fus}$. (d-i) Default parameter values, $(k_p)_{fus}=0.1$, $(k_p)_{fis} = 0.05$, $(k_z)_{fus}=3$, $(k_z)_{fis} =1$, $l_{23}^0 =7.5$, $k_{34}^0 =2.45$,  $k_{\v} =20$, $\gamma_0 = 5000 \; K_B T/\mu \text{m}^2 $, $C_\gamma = 5000 \; K_B T/\mu \text{m}^2 $, see Table \ref{tab:tabl1}.
} 
\label{fig:1cist_Master_sim}
\end{figure}

\subsection{Dimensional reduction through simulating stochastic trajectories\label{sec:macmap}}
Another way to construct an effective description of the rates in the master equation is via simulating the trajectories for the master equation and then analysing the asymptotic behaviour (\textit{asymptotic time domain analysis}). We use the Gillespie algorithm \cite{gillespieE} to compute the system trajectories and then calculate the average over these stochastic trajectories at large time (steady-state average). Hence, based on the steady-state behaviour, we demarcate the solution classes of the master equation Eq.\,\eqref{eq:Master_final_sup} in the space of microscopic parameters. Then, assuming a plausible macroscopic model, we project this phase diagram to the space of parameters of this macroscopic model. The comparison of this projected phase diagram with the phase diagram obtained directly from the macroscopic model confirms or invalidates the assumed macroscopic model. The mapping between the single cisterna microscopic model and a macroscopic Hill-type model is given below. \\

\noindent \textit{Mapping the microscopic mean kernels to dynamical system rates}:
In this work, we have mapped the mean field equation for the microscopic model Eq.\,\eqref{eq:mean_eq_N} to the following dynamical system model,
\begin{small}
\begin{equation}
\frac{dM}{dt} =   \underbrace{v^{fus}_m \, \left(a_0^{fus} +  \frac{\kappa_1 \, M^\alpha}{C_1 + M^\alpha} \right)}_{{\cal K}^{fus}}- \underbrace{v^{fis}_m \left(\frac{\kappa_2 \, M^\beta}{C_2 + M^\beta}\right)}_{{\cal K}^{fis}}\,,
\label{eq:veq_H0}
\end{equation}
\end{small}
which is a detailed version of Eqs.\,\eqref{eq:Rfus1},\eqref{eq:Rfis1}, where we have separated the size-dependent terms $v^{fus}_m, \, v^{fis}_m$ and the rate constants, shown in parentheses. We choose to do so because for the microscopic model, we do not scale the propensities w.r.t. the fission rate, unlike the rates in the macroscopic model (see \ref{sec:si4}). As before, for clarity, we have dropped the overbar for the mean cisternal size ($M$). To compare the phase diagram in the space of microscopic and macroscopic parameters, we need the mapping between these two parameter spaces. We notice that once we fix $\alpha,\beta$, mapping between the microscopic mean flux kernels Eqs.\,\eqref{eq:kernels_analyt_form1},\eqref{eq:kernels_analyt_form2} and macroscopic parameters can be established quite easily (also see Eqs.\,\eqref{eq:Rfus1},\eqref{eq:Rfis1}),
\begin{small}
\begin{eqnarray}
v^{fis}_m \kappa_2 &=&  {\cal K}^{fis}(M\rightarrow\infty) = d\,,\hspace{1cm}
C_2 =\frac{{\cal K}^{fis}(M\rightarrow\infty)}{{\cal K}^{fis}(M\rightarrow 1)} -1  \,,\hspace{1cm}
a_0^{fus} \, {v^{fus}_m} =  {\cal K}^{fus}(M\rightarrow 0) = a \,,\nonumber  \\
C_1  &=& \frac{{\cal K}^{fus}(M\rightarrow \infty)-{\cal K}^{fus}(M\rightarrow 0)}{{\cal K}^{fus}(M\rightarrow 1)-{\cal K}^{fus}(M\rightarrow 0)} \,,\hspace{1cm}
\kappa_1 \, v^{fus}_m= {\cal K}^{fus}(M\rightarrow \infty)-{\cal K}^{fus}(M\rightarrow 0)  = b\,,
\label{eq:mapm}
\end{eqnarray}
\end{small}
and $v^{fus}_m$ is given by Eq.\,\eqref{eq:kernels_analyt_form1} by setting $J_{ss}^{+} =1$. Hence, we have a mapping between microscopic fusion and fission kernels and macroscopic rates. 
 
 In Fig.\,\ref{fig:1cist_Master_sim}(a), we plot the phase diagram in the space of microscopic parameters by computing the average over the stochastic trajectories obtained using Gillespie simulations for the $n$-enzyme master equation Eq.\,\eqref{eq:Master_final_sup}. We use the mapping Eq.\,\eqref{eq:mapm} to project the phase diagram into the space of macroscopic parameters in the dynamical system Eq.\,\eqref{eq:veq_H0} (Fig.\,\ref{fig:1cist_Master_sim}(b)). We then compare the projected phase diagram to the phase diagram in Fig.\,\ref{fig:p_roots1}(b) (Fig.\ref{fig:phases_1d}(b) in the main text) obtained directly using fixed point analysis (see \ref{sec:si4}). We observe that these two phase diagrams are qualitatively similar, validating the dynamical system model.

\subsection{An algebraic criterion for dimensional reduction \label{sec:Markov_current}}
Here, we present an algebraic method to arrive at the macroscopic rates from the microscopic model. To this end, we derive the functional form that the steady state microscopic current takes for given microscopic transition rates in the $4$-state Markov cycles and the probability distributions for enzyme/site availabilities. \\

\noindent\textbf{Steady state current in a Markov cycle:}\\
We will compute steady-state current in a general Markov cycle with $3$ and $4$ nodes for any given transition rates. Then we will derive the effective dependence of this current on the cisternal size, $M$, for given cisternal size dependencies of the transition rates.\\
\noindent\textbf{3-nodes:}

Consider a Markov system with three nodes (Fig.\ref{fig:nodes_sites_sum}(a)), with rates given by $(a,a',b,b',c,c')$, where unprimed indices correspond to forward transitions and primed indices correspond to reverse transitions. With the transition rate matrix \cite{Norris1997},
\begin{small}
\begin{eqnarray}
\mathcal{M} = \begin{bmatrix}
 -a-c' & a' & c \\
  a & -b-a' & b'  \\ 
  c' & b & -c-b' 
   \end{bmatrix}\,.
   \label{eq:transition_M}
\end{eqnarray}
\end{small}
One can show that the current associated with the cycle is given by the ratio of two terms, the numerator $\mathcal{A}$ is given by the Kolmogorov current condition \cite{Kelly1979}, 
\begin{small}
\bea
\mathcal{A} =  a  b c - a'b'c' \,, 
\eea
\end{small}
whereas denominator $\mathcal{B}$ is given by the sum of the determinants of the cofactors of the transition rate matrix Eq.\,\eqref{eq:transition_M}, with one of the rows replaced by (1,1,1).
\begin{small}
\begin{eqnarray}
\mathcal{M}' = \begin{bmatrix}
 1 & 1 & 1 \\
  a & -b-a' & b'  \\ 
  c' & b & -c-b' 
   \end{bmatrix}\,.
   \label{eq:mark_3}
\end{eqnarray}   
\end{small} 
 This comes from the probability normalization condition, and this modified matrix $\mathcal{M}'$ has full rank unlike the transition rate matrix, Eq.\,\eqref{eq:transition_M}. From Eq.\,\eqref{eq:mark_3}, one can compute  $\mathcal{B}$ and see that it cannot have terms that are products of rates corresponding to the outflux rates at the same node, such as $ac', \,aa', \,ba', \, bb'$, etc. 
\begin{small}
\bea
\mathcal{B}  = \underbrace{ab + bc + ca}_{\textrm{forward sum}} + \underbrace{a'b' + b'c' + c' a'}_{\textrm{backward sum}} + \underbrace{ab'  + bc' + ca'}_{\textrm{mixed sum}}\,.
\label{eq:kolmogorov3}
\eea
\end{small}
Here, we have broken the above sum in three parts. The second term is obtained from the first by priming two of its indices, while the third term is obtained from the forward sum by priming a single index, given that the multiplying factors do not correspond to outfluxes at the same node.
Hence, sums like $(a'b  + b'c + c'a)$ are ruled out. 

\noindent \textbf{4-nodes:} 

The above exercise can be repeated for $4$ nodes (Fig.\ref{fig:nodes_sites_sum}(b)), and the numerator and the denominator are given by,
\begin{small}
\begin{eqnarray}
\mathcal{A} &=&  a  b c d - a'b'c'd'  \nonumber \\
\mathcal{B} &=&  \underbrace{abc + bcd + cda + dab}_{\textrm{forward sum}} + \underbrace{a'b'c' + b'c'd' + c'd'a' + d'a'b'}_{\textrm{backward sum}}  + \underbrace{abc' + bcd' + cda' + dab'}_{\textrm{mixed sum 1 }}  + \underbrace{ab'c' + bc'd' + cd'a' + da'b'}_{\textrm{mixed sum 2}} \,,\nonumber \\
\label{eq:kolmogorov}
\end{eqnarray}
\end{small}
where terms in third and fourth sum can be obtained from the terms in the first sum by sequentially adding primed variables, with the condition that the transition rates do not correspond to the outfluxes at the same node. 

\begin{enumerate}[label=(\alph*)]
\item The numerator $\mathcal{A}$ is a difference of two terms : a product of four forward rates and a product of four backward rates.

\item The denominator $\mathcal{B}$ consists of a total of $16$ terms, which are products of three rates.

\item  The functional dependence of steady state current in a Markov cycle on variables such as cisternal size can be written
down from the functional dependencies of various transition rates. For example, if the transition rates have polynomial or Hill-kind dependencies on the cisternal size, steady state current can be written as $\frac{C_\alpha \, M^\alpha + C_\beta M^\beta + \ldots + C_0}{C'_\alpha \, M^\alpha + C'_\beta M^\beta + \ldots + C'_0}$.

\item The magnitude of the current in a Markov cycle is dominated by the smallest forward rate in the cycle. For example, assuming that all backward (primed) rates are small compared to forward (unprimed) rates, above current can be approximated from Eq.\,\eqref{eq:kolmogorov},
 
$$\frac{a  b c d}{abc + bcd + cda + dab} \approx \frac{a  b c d}{dab} \approx c \,, $$

if $c$ is the smallest forward rate. In the main text, we have assumed this to be the rates dependent on membrane parameters such as Golgi membrane tension, and all other forward rates are kept very high ($1000$ per second) \cite{arf1,allin}.
\end{enumerate}

For the $4-$state fusion and fission cycle, the transition $|3\rangle \rightarrow |4\rangle$ for fusion and  the transition $|3\rangle \rightarrow |2\rangle$ for fission correspond to the domain flattening rates and are proportional to the tension ($c \equiv c(\gamma)$ for fusion and $b' \equiv b'(\gamma$) for fission respectively). Also, the transition $|1\rangle \rightarrow |2\rangle$ is proportional to the number of docking and budding sites $m$ ($a \equiv a \, m $). Assuming that all other reverse transition rates are small ($1$ per second, see Table \ref{tab:tabl1}), currents over fusion and fission cycles are given by

\begin{small}
\begin{eqnarray}
J^{-}_{ss}(M,m) &=& \frac{a b c d  \, m }{(a  b c \, m  + a b d \, m   + b c d  + c d a \, m )  + a d b' (\gamma) } \\
J^{+}_{ss}(M,m) &=& \frac{a b c (\gamma) d \, m}{a  b d \, m  + (a b c (\gamma) \, m + c (\gamma) d a \, m  + b c (\gamma) d) } \,.
\label{eq:jm}
\end{eqnarray}
\end{small}
As before, note that if the backward rates are small, $J_{ss}^{-}(M,m) \approx b$ ($b$ being the forward transition rate $|2\rangle  \rightarrow |3\rangle $ for fission), and $J_{ss}^{+}(M,m) \approx c(\gamma)$. Furthermore, if we assume the Golgi membrane tension to be $\gamma = \gamma_0 -k_s(M-M_0)$ and $b'(\gamma), c(\gamma)$ to be saturated Hill functions of the membrane tension, we can plug these functions into Eq.\,\eqref{eq:kolmogorov} and from here we can read the hill exponents and coefficients. We now present an approximation, when the currents in the Markov cycles have a weak dependence on the cisternal size $M$.\\
\begin{figure*}[t!]
\centering
\includegraphics[width=0.9\textwidth]{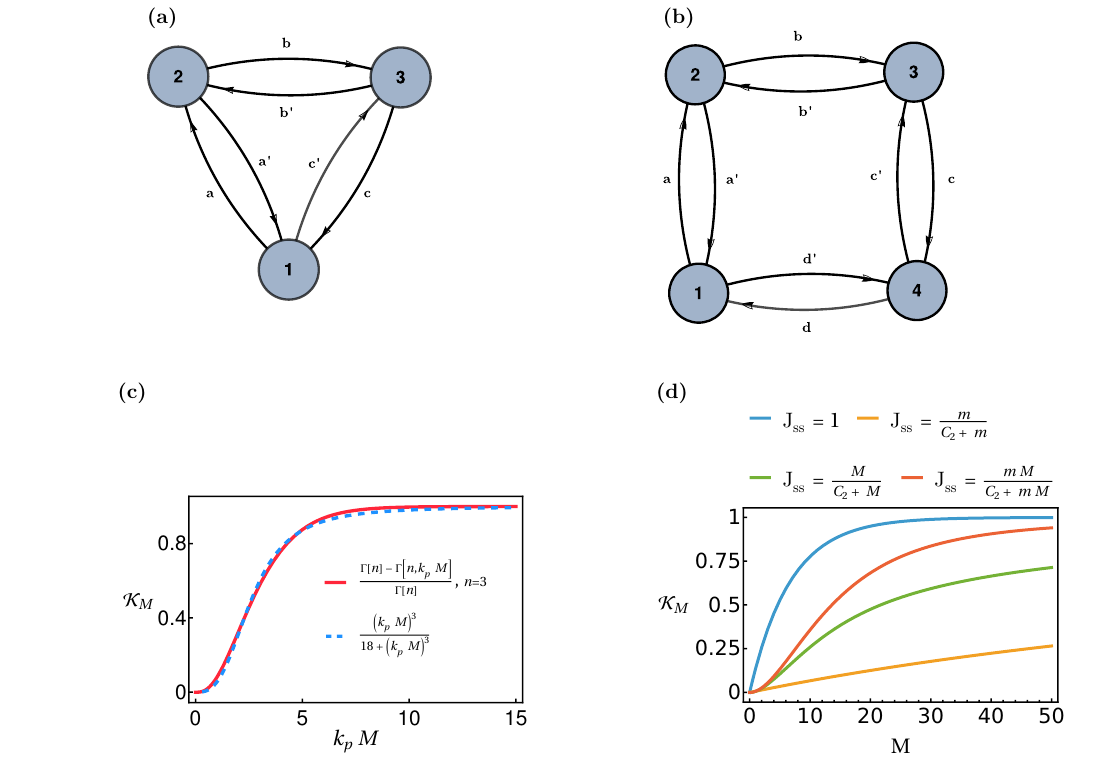}
\caption{{\bf Parameter space reduction using algebraic expressions for currents over fusion and fission cycles}: (a,b) Markov cycles for $3$-node and $4$-node systems. Eqs.\,\eqref{eq:kolmogorov3},\eqref{eq:kolmogorov} outline the calculation for the nonequilibrium steady-state current over these cycles. (c) Hill function (shown as blue dashed line) is a reasonably good approximation for the summation $\mathcal{K}(M)$ (shown as red thick line) in Eq.\,\eqref{eq:kernels_analyt_form2_red} with current, $J_{ss}^-=1$. (d) For various model functions $J_{ss}^-$, sum over site availability, $m$ in Eq.\,\eqref{eq:kernels_analyt_form2_red} with $n=1$, gives saturating behaviour as a function of cisternal size $M$ (Poisson rate for site availability, $k_p=0.15$, Hill saturation constant, $C_2=20$). These saturating functions $\mathcal{K}_M$, when summed over enzyme availability (see Eq.\,\eqref{eq:kernels_analyt_form2_full}) yields mean fission (and fusion) kernel, also a saturating function of cisternal size $M$, see Eq.\,\eqref{eq:Rfis1M}. The qualitative behavior observed in panels (b,c) remains valid for both microscopic fission  and microscopic fusion kernels (see Eqs.\,\eqref{eq:kernels_analyt_form2_red},\eqref{eq:kernels_analyt_form_red}).}
\label{fig:nodes_sites_sum}
\end{figure*}

\noindent \textbf{Mapping to dynamical system}

From the mean fission kernel Eq.\,\eqref{eq:kernels_analyt_form2}, sum over $m$ for a particular $n$ reads
\begin{small}
\begin{eqnarray}
{\cal K}^{fis}_M(M) &=& \sum_{m=n}^{N_E} \, \frac{e^{-k_p\, M} (k_p\, M)^m}{m!} \, J_{ss}^{-} (M,m) \,.
\label{eq:kernels_analyt_form2_red}
\end{eqnarray}
\end{small}

The cisternal size, $M$-dependence in Eq.\,\eqref{eq:kernels_analyt_form2_red} comes from two factors, one from site availability, given by Poisson distribution with rate $k_p \, M$ and the other from the steady state current over the Markov cycle $J_{ss}^{-} (M,m)$. As shown above, $J_{ss}^-  (M,m)$ can be written in Hill-form, with exponents and saturation constants that arise from the size dependencies of the transition rates in the Markov cycles. If we set the  current $J_{ss}^{-} (M,m)$ to  a constant ($=J^-$, which would be the case if  $m \approx 1$ and the transition rates are slow varying functions of tension, i.e., $\gamma \approx \gamma_0$), then Eq.\,\eqref{eq:kernels_analyt_form2_red} is a sum over tail $(m>n)$ of a Poisson distribution, which  saturates as a function of $M$ and can be approximated by a Hill type function, $\frac{M^n}{\left(\mathcal{C}_n /(k_p)^n+ M^n\right)}$. Here, $n$ (which depends on the Poisson rate of enzyme availability, $k_z$) sets the dominant term in the sum and characterizes the cooperativity (Hill-exponent) and $\mathcal{C}_n$ is an increasing function of $n$ that can be obtained from the fitting, Fig.\ref{fig:nodes_sites_sum}(c). Moreover, $\mathcal{C}_n/(k_p)^n$ sets the Hill saturation constant for this sum. Using the approximated Hill function, the full fission kernel is given by
\begin{small}
\begin{eqnarray}
{\cal K}^{fis}(M) = \sum_{n=1}^{N_E} n \frac{e^{-k_z} (k_z)^n}{n!} \frac{M^n}{\left(\mathcal{C}_n /(k_p)^n+ M^n\right)} \,,
\label{eq:kernels_analyt_form2_full}
\end{eqnarray}
\end{small}
which is a weighted average over $n$ for a given cisternal size $M$. Note that for large $M$, the above sum reduces to ${\cal K}^{fis}(M) \rightarrow k_z$ (in the limit $N_E\rightarrow \infty$), and this, along with the  multiplicative constant, $J^-$ gives the peak rate of the weighted sum $\mathscr{d}(k_z)$. We can use the mapping Eq.\,\eqref{eq:mapm} to extract an effective Hill constant, and it  is a slightly cumbersome expression,
\bea
\mathcal{C}(k_z) =  k_z e^{k_z}\bigg/\left(\sum_{n=1}^{\infty} \frac{(k_z)^n}{(n-1)!} \cdot \frac{1}{C_n/(k_p)^n+1}\right)-1\,,
\label{eq:sum_C_kz}
\eea
which can be shown to be an increasing function of $k_z\, (\partial\mathcal{C}(k_z)/\partial k_z >0)$. Hence, the above sum Eq.\,\eqref{eq:kernels_analyt_form2_full} can be approximated by
\begin{small}
\begin{eqnarray}
{\cal K}^{fis}(M) = \frac{\mathscr{d}(k_z) M^\eta}{\mathcal{C}(k_z) + M^\eta}\,,
\label{eq:kernels_analyt_form2_Ckz}
\end{eqnarray}
\end{small}

where $\eta$ is the effective cooperativity for the sum and $\mathscr{d}(k_z)$ can be rewritten as $\mathscr{d}(k_z) = \nu_m^{fis} \kappa$, segregated into a size term, $\nu_m^{fis}$ and a time scale, $\kappa$. Similarly, for the mean fusion kernel  Eq.\,\eqref{eq:kernels_analyt_form1}, the sum over the tail $(m>n)$ of the Poisson distribution for site availability gives ($N_E\rightarrow \infty$),

\begin{small}
\begin{eqnarray}
{\cal K}^{fus}_M(M) &=& \sum_{m=n}^{\infty} \, \frac{e^{-k_p\,( M+1)} (k_p\, (M+1))^m}{m!}   = \left(\sum_{m=n}^{\infty} \frac{e^{-1}}{m!} \right)  + \left( \sum_{m=n}^{\infty} \left( \frac{e^{-k_p ( M+1)} (k_p \,(M+1))^m}{m!} - \frac{e^{-1}}{m!} \right) \right)\,,
\label{eq:kernels_analyt_form_red}
\end{eqnarray}
\end{small}
where the lowest order term is a constant due to the $(M+1)$ term and following similar arguments sketched for the fission kernel, one can approximate the sum Eq.\,\eqref{eq:kernels_analyt_form1} with a Hill function plus a constant. 
\begin{small}
\begin{eqnarray}
{\cal K}^{fis}(M) = \mathscr{a}_0 (k_z) + \frac{\mathscr{d}(k_z) M^\alpha}{\mathcal{C}(k_z) + M^\alpha}\,,
\label{eq:kernels_analyt_form2_C1kz}
\end{eqnarray}
\end{small}
where, $\mathscr{d}, \mathcal{C}$ above are different from that of the fission kernel. 

Note that if the tension dependence of the Markov transition rates are of the Hill or Michaelis-Menten type with a large dynamic range $C_{\gamma}$, the membrane tension is a slowly varying function of the cisternal size. Hence, the Hill-exponent and saturation constant for the summations Eqs.\,\eqref{eq:kernels_analyt_form1},\eqref{eq:kernels_analyt_form2} are largely determined by the site availability factor, whereas the Markov cycle contributes to the rate in the numerator (which is of the order of the smallest forward rate in the Markov cycle). Furthermore, for more complicated polynomial or Hill-kind dependencies of microscopic rates on the cisternal size (see discussion on the computation of current in  a Markov cycle above), fission/fusion kernel can be approximated by a product of two Hill functions with different exponents, summed over the site and enzyme probabilities (i.e., over $m,n$). This sum can be further approximated by a Hill function. In Fig.\ref{fig:nodes_sites_sum}(d), we plot this sum  over enzyme availability for various model functions for $J_{ss}$.  

The above discussion confirms that the microscopic mean fusion and fission kernels can be approximated by the following macroscopic fission and fusion rates,

\begin{small}
\begin{eqnarray}
{\cal K}^{fus} = v_m^{fus}  \left( a_0^{fus} + \frac{\kappa_1 \,(M-M_0)^\alpha }{C_1 +  \, (M-M_0)^\alpha } \right)  \hspace{1cm}
{\cal K}^{fis} = v_m^{fis} \left(\frac{\kappa_2 \,(M-M_0)^\beta }{C_2 +  \, (M-M_0)^\beta } \right)  \,,
\label{eq:Rfis1M}
\end{eqnarray}
\end{small}
which with $M_0=0$, can be reparameterized in the form of Eqs.\,\eqref{eq:Rfus1},\eqref{eq:Rfis1} (Eqs.\,\eqref{eq:Rfisfus} in the main text). 

\subsection{Variations in the form of kernels do not change the topology of fixed points} 
The functional form of microscopic mean kernels would change if the following changes: (a) The probability distribution of the enzyme availability (assumed to be a Poisson process here) changes. This can occur if the Poisson rate changes, or if the underlying dynamics are inherently non-Poissonian. For instance, a finite pool ensemble \cite{banerjee, mukherjee} would give non-Poissonian availability of active fusogens/fisogens. (b) The availability of the ``docking'' and/or ``exit'' site has a different dependence (say, stronger or weaker) on the cisterna size. (c) Transition rates in the $4$-state Markov cycle have different dependencies on size and tension. We will explore the third point now.

Suppose that the size dependence in the fusion kernel is not dominated by the site availability but the tension dependence of the fusion-cycle transition rate. If these transition rates are non-monotonic in cisternal size, microscopic mean fusion and fission kernels may not remain monotonically increasing functions of cisternal size. For instance, when the budding and the flattening rates are given by Helfrich free energy \cite{nelson}.
\begin{small}
\begin{eqnarray}
\mathcal{H}(\beta) = \kappa^B \, (\sqrt{\beta} -1)^2 + \gamma \, \beta \,,
\label{eq:Helf_H}
\end{eqnarray}
\end{small}
where $\gamma = \gamma_0 - k_s(M-M_0)$, is the Golgi cisterna membrane tension; $\gamma_0,M_0$ are the cisterna membrane rest tension and size and  $k_s$ is the membrane stretching modulus for the Golgi cisterna, $\kappa^B$ is the bending rigidity and $\beta$ is  the budding parameter ($0 \leq \beta \leq 1$), $1$ being a fully budded domain and $0$ being a flattened domain, the line tension is assumed to be zero \cite{sens1}.   Domain budding and flattening events occur in the transitions $|3\rangle\leftrightharpoons |4\rangle$ for fusion cycle and  $|3\rangle \leftrightharpoons |2\rangle$ for fission cycle. We choose the  bending rigidities for the fusion and fission processes to be ($\kappa^B_{fus} = 1950,\, \kappa^B_{fis} = 2000, \, k_s = 1, \, \gamma_0 = 5000$, in the units of $K_BT$ \cite{sens}). With these parameter values, the budding and flattening rates -- calculated via Kramers' rate theory \cite{sens,kampen} for the Helfrich free energy \cite{nelson} -- exhibit non-monotonic behavior across the physically relevant range of cisternal sizes. We can then use these membrane tension dependent transition rates to compute the current $J_{ss}^+$ and $J_{ss}^-$ over the fusion and fission cycle, respectively. Assuming the availabilities of fusogens, fisogens, vesicles and sites (docking as well as exit) are given by Poisson processes (as before), we can sum the fusion and fission currents over these probability distributions to compute the mean fusion and fission kernels (see Eqs.\,\eqref{eq:kernels_analyt_form1},\eqref{eq:kernels_analyt_form2}). As we can see in Fig.\ref{fig:golgi_roots_helfrich}, the obtained mean fusion and fission kernels are non-monotonic. However, the topology of the fixed points remains the same for one cisterna (see Fig.\,\ref{fig:golgi_flux_123}). We will explore the consequences of such kernels in future work.

\begin{figure}[t!]
\centering
\includegraphics[width=0.95\textwidth]{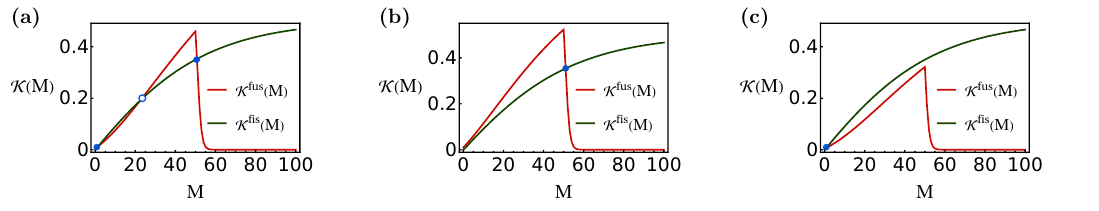}
\caption{{\bf Intersections of mean fusion and fission kernels with  Helfrich Hamiltonian \cite{nelson}} (a-c) In the model Eq.\,\eqref{eq:Helf_H}, where domain budding and flattening rates are derived from the Helfrich Hamiltonian, varying the microscopic parameters alters the intersections of the fission and fusion kernels. This changes the number of stable and unstable fixed points. In this model, size dependence in fusion kernel is not dominated by the site availability, but the tension dependence of the fusion-cycle rates and the fusion kernel is non-monotonic. However, the fixed point structure remains the same for one cisterna (see Fig.\,\ref{fig:golgi_flux_123}). Parameters values, $(k_p)_{fus}=0.05$, $(k_p)_{fis} = 0.05$, $(k_z)_{fus}=2-3$, $(k_z)_{fis} = 1$, $l_{23}^0 = 1$, $k_{34}^0 = 2-5 $, $C_\gamma = 5000 \; K_B T/\mu \text{m}^2 $, bending rigidities, $\kappa^B_{fis} = 5000,\, \kappa_{fus}^B = 1950$, membrane stretching modulus, $k_s = 1$, and membrane rest tension, $\gamma_0 = 5000$, in the units of $k_B T$, see Eqs.\,\eqref{eq:kernels_analyt_form1},\eqref{eq:kernels_analyt_form2} and Table \ref{tab:tabl1} for notational details.}
\label{fig:golgi_roots_helfrich} 
\end{figure}

\section{Stochastic master equation for two cisternae}
\label{sec:Mastereqtwocisternae}

In this section, we sketch how to extend the stochastic master equation Eq.\,\eqref{eq:Master_final_sup}, with microscopic parameters for two and more cisternae. This essentially involves specifying how material is transferred from cisterna $1$ to $2$ and vice versa. We discuss the following two schemes that precisely address this,

\begin{figure}[t!]
\centering 
\includegraphics[width=0.95\textwidth]{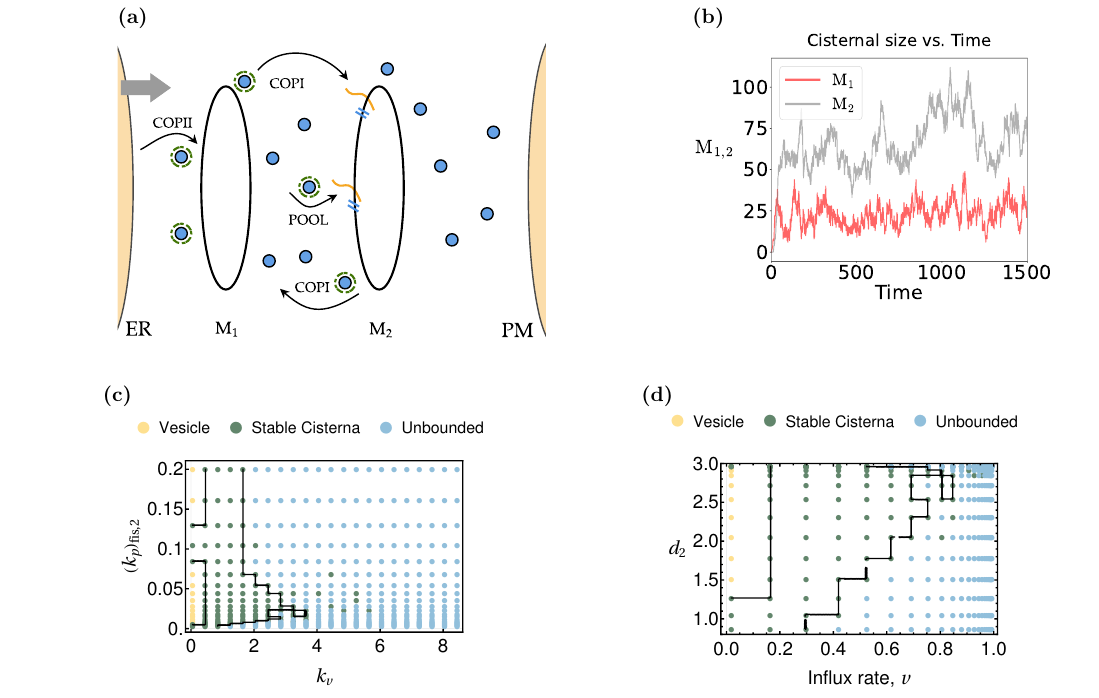}
\caption{{\bf Simulating the master equation for two cisterna.} (a) Schematic of vesicle transfer between cisternae via COPI mediated direct transfer and COPI mediated indirect transfer through the intercisternal pool of vesicles.
(b) Stochastic trajectories $2$-cisternae system, $M_1$(red),  $M_2$(grey), obtained by extending the stochastic master equation Eq.\,\eqref{eq:Master_final_sup} assuming the intercisternal vesicle transfer occurs via the \textbf{Direct} scheme Eq.\,\eqref{eq:direct_scheme}. Parameter values, $k_\v = 2.1$, $(k_p)_{fus, (1,2)} = (0.2,0.15)$, $(k_p)_{fis, (1,2)} = (0.05,0.04)$, $(k_z)_{fus,1} = 1$, $(k_z)_{fis,(1,2)} = 1$, $T_{12} = 0.5$, $T_{21} = 0$, $(k_{34}^0)_{fus,(1,2)} = 8$, $(l_{32}^0)_{fis,(1,2)} = 6$. (c) Phase diagram for $2$-cisternae system in the microscopic parameter space, ($k_\v,(k_p)_{fis,2}$) plane generated using the \textbf{Direct} scheme depicting vesicle, stable cisterna and unbounded growth regimes.  Here, $k_\v$: Poisson rate of vesicle availability at cisterna $1$, $(k_p)_{fis,2}$: fission enzyme availability at Cisterna $2$. (d) The phase diagram in  the microscopic parameter space in (c) is projected to ($v,d_2$) plane, using the mapping Eq.\,\eqref{eq:mapm}, $v$: macroscopic influx rate at cisterna $1$, $d_2$:  macroscopic exit rate at cisterna $2$. This projected phase diagram is qualitatively similar to the one obtained for the $2$-cisternae dynamical system, Fig.\ref{fig:fp_lc_sol}(a) in the main text, therefore validating the dynamical system model. Here, $k_\v$: Poisson rate of vesicle availability at cisterna $1$,  $(k_z)_{fus,(1,2)}$: fusion enzyme availability at cisterna $1$ and $2$,  $(k_z)_{fis,(1,2)}$: fission enzyme availability at cisterna $1$ and $2$, $(k_p)_{fus,(1,2)}$: ``docking'' site availability at cisterna $1$ and $2$,  $(k_p)_{fis,(1,2)}$: ``exit'' site availability at cisterna $1$ and $2$, $v$: macroscopic influx rate at cisterna $1$, $d_2$:  macroscopic exit rate at cisterna $2$, $(k_{34}^0)_{fus,(1,2)}$: flattening rates in the fusion cycles at cisterna $1$ and $2$, $(l_{32}^0)_{fis,(1,2)}$: flattening rates in the fission cycles at cisterna $1$ and $2$. }
\label{fig:2cist_Master_sim}
\end{figure}

\begin{enumerate}[label=(\alph*)]
\item \textbf{Direct} scheme :  vesicles are directly transferred from cisterna $1$ to $2$ via COPI proteins, and this material transfer from cisterna $1$ to cisterna $2$ happens within a given intercisternal time $T_{12}$, at the end of which cisternal sizes are updated. In the Gillespie code, this can be implemented as follows:  

Suppose the sizes of the cisterna $1$ and $2$ are $M_1$ and $M_2$ respectively. Given the values of all the microscopic rates, we need to compute the probability of, say, $n$ number of vesicles fissing out at cisterna $1$, and then fusing to cisterna $2$. This involves finding $n$ fission enzymes, then at least $n$ exit sites at cisterna $1$, and then finding at least $n$ docking sites at cisterna $2$.  This multiplied by the intercisternal time $T_{12}$ and the probability of finding $n$ COPI proteins, gives the propensity of the event that $n$ number of vesicles is transferred from cisterna $1$ to $2$. Hence, in the Gillespie code propensity of n-vesicles transfer is given by,
\bea
\mathcal{P}_{12}(n)=  T_{12} \, Pois(n,k_z) \,  \left( \sum_{k=n}^{N_E} \, Pois(k,k_p \, M_1) \,J_{ss}^{-}(M_1,k)  \sum_{m=k}^{N_E} \, Pois(m,k_p'\, M_2) \, J_{ss}^+ (M_2,m)\right)\,,
\label{eq:direct_scheme}
\eea
where, $Pois(i,q)$ above is the Poisson probability of $i$ events for a given rate $q$, ($k_p \, M_1,k_p' \,M_2$) are Poisson rates of finding exit and docking sites respectively, $k_z$ is the Poisson rate of finding fusion enzyme and $T_{12}$ is the intercisternal time. Here, in the numerical scheme, we truncate the sum a large $N_E$, beyond which the Poisson probabilities are zero. We can use a similar procedure to compute transfer of vesicles from cisterna $2\rightarrow 1$.

\item \textbf{Indirect} scheme : vesicles are transferred from cisterna $1$ to the intercisternal space during a given intercisternal time $T_{12}$. At the end of this time period, this pool of intercisternal vesicles are available for fusion, from where vesicles are taken by fusion enzymes to cisterna $2$. In the Gillespie code, this can be implemented as follows:  

Suppose the size of cisterna $1$ is $M_1$. Given the values of all the microscopic rates, compute the probability of finding at least $m$ exit sites at cisterna $1$. One can compute the mean number of fission events at cisterna $1$ by summing over the probability of finding $m$ COPI proteins.
This multiplied by the intercisternal time $T_{12}$ gives the mean number of vesicles available in the intercisternal pool. These vesicles are then carried over to cisterna $2$ via COPI enzymes for fusion, just like single cisterna case. In the Gillespie code, the propensity of n-vesicles transfer can be computed as,
\bea
\label{eq:indirect_scheme1}
k_\v  &=& T_{12} \,  \sum_{m=1}^{N_E} Pois(m,k_z) \, \left(\sum_{k=m}^{N_E} \,Pois(k,k_p \, M_1) \, J_{ss}^{-}(M_1,k) \right)\\
\mathcal{P}_{12}(n) &=& Pois(n,k_z') \left( \sum_{k=n}^{N_E} \, Pois(k,k_\v) \, \sum_{m=k}^{N_E} \, Pois(m,k_p' M_2) \, J_{ss}^+ (M_2,m)\right)\,,
\label{eq:indirect_scheme2}
\eea
where, as above, $Pois(i,q)$ above is the Poisson probability of $i$ events for a given rate $q$, ($k_p \, M_1,k_p' M_2$) are Poisson rates of finding exit and docking sites respectively, ($k_z,k_z'$) are the Poisson rates of finding fission and fusion enzymes respectively and $T_{12}$ is the intercisternal time. We can use a similar procedure to compute transfer of vesicles from cisterna $2\rightarrow 1$.



\end{enumerate}

The above two schemes can be extended to multiple cisternae via adding intercisternal times $T_{ij}$ between cisterna $i$ and $j$. This along with $n$-enzyme master equation Eq.\,\eqref{eq:Master_final_sup} can be used to compute the stochastic trajectories and the phase diagram for two (or multiple) cisternae system, Fig.\,\ref{fig:2cist_Master_sim}. Since the stability of the solutions is difficult to assert for the stochastic case, we use a simple criteria: if $M<M^{high}_{th} = 300$, it is a stable phase (vesicle if $M\leq1$), and unbounded phase otherwise. 

\section{Detailed dynamical system analysis of a single cisterna\label{sec:si4}}
With various methods outlined in \ref{sec:si2}, the deterministic mean dynamics for the cisternal size $M$ can be written as
\begin{small}
\begin{eqnarray}
\frac{dM}{dt} &=&  {\cal K}^{fus}(M)- {\cal K}^{fis}(M)  =  v^{fus}_m \, \left(a_0^{fus} +  \frac{\kappa_1 \, M^\alpha}{C_1 + M^\alpha} \right)- v^{fis}_m \left(\frac{\kappa_2 \, M^\beta}{C_2 + M^\beta}\right)\,,
\label{eq:veq_H}
\end{eqnarray}
\end{small}
 where we have separated the size (mass) terms $v^{fus}_m,v^{fis}_m$ and the rate constants (shown in parentheses). We scale the above equation in terms of $\kappa_1,\, \kappa_2$ to rewrite the above equation in terms of influx rate $v = v^{fus}_m \, \kappa_1$, exit rate $d = v^{fis}_m \, \kappa_2$ and a nucleation constant, $a_0 = a_0^{fus}/\kappa_1$,
\begin{small}
\begin{equation}
\frac{dM}{dt} =   v \, \left(a_0 + \frac{M^\alpha}{C_1 + M^\alpha} \right)- \left(\frac{d \, M^\beta}{C_2 + M^\beta}\right)\,.
\label{eq:sveq_H}
\end{equation}
\end{small}
The parameters $C_1, C_2$ represent the saturation constants for the fusion and fission kernels, respectively, while the Hill-exponents $\alpha,\beta$, define the {\it cooperativity} of the fusion and fission processes. The influx rate, $v$  and the nucleation constant, $a_0 > 0$ ensure that there is a nucleation seed for the cisterna. In the main text, we have converted this equation into dimensionless form by setting $d=1$. Now, we study the root structure of this dynamical system. 
\subsection{Number of positive roots}
\label{sec:single_cist_si} 
Given the Hill type kernels, that we have assumed above, one can have at most three fixed points, regardless of the value of Hill exponents. To see this, we notice that the fixed points of the above equation are  given by
\begin{small}
\begin{equation}
 C_1 C_2 a_0 v  + C_1 ( a_0 v - d)  M^{\beta} + C_2 (v + a_0 v ) M^{\alpha} + ( v + a_0 v  - d)M^{\alpha+\beta} = 0 \,.
\label{eq:cubic_alp}
\end{equation}
\end{small}
The above polynomial can have at most three sign changes in its coefficients, and hence at most three positive roots, by Descartes’ rule of signs \cite{burnside,katz} and exactly that many positive roots if all roots are real. Various scenarios for $\alpha =2,\, \beta =1$ are detailed in Table \ref{Tab:Tcr}. Since the number of roots does not depend on the values of $\alpha,\beta$; we may fix $\alpha=2,\,\beta=1$ from here onwards without loss of generality, resulting in the dynamical equation,
\begin{small}
\begin{equation}
\frac{dM}{dt} =   v \, \left(a_0 + \frac{M^2}{C_1 + M^2} \right)- \left(\frac{d \, M}{C_2 + M}\right)\,.
\label{eq:sveq_H1}
\end{equation}
\end{small}

\begin{table}[t!]
\centering
\renewcommand{\arraystretch}{1.3} 
\small
\begin{tabular}{|c|c|c|c|c|c|c|c|c|}
\hline
\multicolumn{9}{|c|}{\vspace*{0.15cm} \textbf{Number of Positive roots of cubic polynomial Eq.\eqref{eq:cubic_alp}} \vspace*{-0.1cm}} \\
\hline
\parbox{1.6cm}{\centering \vspace*{0.1cm} $\pm()x^3$ \\[-0.1cm] $a_0 v + v-d$ \vspace*{0.1cm}} & 
\parbox{1.6cm}{\centering \vspace*{0.1cm} $+()x^2$ \\[-0.1cm] $a_0 v +v$ \vspace*{0.1cm}} & 
\parbox{1.6cm}{\centering \vspace*{0.1cm} $\pm()x$ \\[-0.1cm] $a_0 v -d$ \vspace*{0.1cm}} & 
\parbox{1.2cm}{\centering \vspace*{0.1cm} $+()x^0$ \\[-0.1cm] $a_0 v$ \vspace*{0.1cm}} & 
\parbox{0.8cm}{\centering \vspace*{0.1cm} $\Delta$ \vspace*{0.1cm}} & 
\parbox{1.5cm}{\centering \vspace*{0.1cm} N(real \\[-0.1cm] roots) \vspace*{0.1cm}} & 
\parbox{1.8cm}{\centering \vspace*{0.1cm} N(possible \\[-0.1cm] +ve roots) \vspace*{0.1cm}} & 
\parbox{1.8cm}{\centering \vspace*{0.1cm} N(possible \\[-0.1cm] -ve roots) \vspace*{0.1cm}} & 
\parbox{1.4cm}{\centering \vspace*{0.1cm} N(+ve \\[-0.1cm] roots) \vspace*{0.1cm}} \\
\hline
$+$ & $+$ & $+$ & $+$ & $+$ & 3 & 0 & 3/1 & 0 \\
$+$ & $+$ & $-$ & $+$ & $+$ & 3 & 2/0 & 1 & 2 \\
$-$ & $+$ & $-$ & $+$ & $+$ & 3 & 3/1 & 0 & 3 \\
$-$ & $+$ & $+$ & $+$ & $+$ & NF & NF & NF & NF \\
$+$ & $+$ & $+$ & $+$ & $-$ & 1 & 0 & 3/1 & 0 \\
$+$ & $+$ & $-$ & $+$ & $-$ & 1 & 2/0 & 1 & 0 \\
$-$ & $+$ & $-$ & $+$ & $-$ & 1 & 3/1 & 0 & 1 \\
$-$ & $+$ & $+$ & $+$ & $-$ & NF & NF & NF & NF \\
\hline
\multicolumn{9}{|l|}{\vspace*{0.1cm} *NF = non-feasible \vspace*{-0.1cm}} \\
\hline
\end{tabular}
\caption{{\bf Root structure of cubic polynomial Eq.\,\eqref{eq:cubic_alp}}: According to Descartes' rule of signs \cite{burnside}, the number of positive roots of a polynomial equals the number of sign changes in its coefficients, or is less by an even integer. The cubic discriminant gives the maximum number of real roots for a cubic polynomial ($\Delta >0$ for three real roots and $\Delta <0$ for one). Combined with the fact that only cubic and linear term of polynomial Eq.\,\eqref{eq:cubic_alp} can change sign, we can enumerate number of roots as we vary the parameters (Fig.\,\ref{fig:golgi_flux_123}(a-e)). This can also be deduced from say, Cardano's formula \cite{katz} for cubic polynomials, but the method used here works for higher order polynomials as well.}
\label{Tab:Tcr}
\end{table}
However, it should be noted that the nature of roots depends on the values of $\alpha$ and $\beta$. Eq.\,\eqref{eq:cubic_alp} with $\alpha =2, \beta=1$ gives
   \begin{equation}
 C_1 C_2 a_0 v  + C_1 ( a_0 v - d)\,  M + C_2 (v + a_0 v ) \, M^2 + ( v + a_0 v  - d)\,M^3 = 0 \,.
 \label{eq:eq1}
\end{equation}
Now, interchanging  $\alpha$ and $\beta$ ( $\alpha=1$, $\beta=2$) in the above equation changes the order of the coefficients
   \begin{equation}
 C_1 C_2 a_0 v  + C_2 (v + a_0 v )\, M  + C_1\, ( a_0 v - d) \, M^2 + ( v + a_0 v  - d) \, M^3 = 0 \,.
  \label{eq:eq2}
\end{equation}
Given that the parameters are positive valued, in Eq.\,\eqref{eq:eq1}, the order of sign changes for the coefficient of the polynomial is $(+,\pm,+,\pm)$ allowing maximum three sign changes and hence maximum three positive roots. However, in Eq.\,\eqref{eq:eq2}, the order of sign changes for the coefficient of the polynomial Eq.\,\eqref{eq:eq2} is $(+,+,\pm,\pm)$ allowing for a maximum of two sign changes and hence a maximum of two positive roots. So, for $a_0>0,\, v>0$, Eq.\,\eqref{eq:eq2} has at most two positive roots and since Eq.\,\eqref{eq:sveq_H} corresponds to a  positive system (which is ensured by the flows of the dynamical system), the larger root for the case with $\alpha<\beta$ is an unstable one. Therefore, for stable biogenesis for the generic values of the parameters, we assume $\alpha>\beta$.

\subsection{Phase diagram for single cisterna}
\begin{figure}[t!]
\centering
\includegraphics[width=\textwidth]{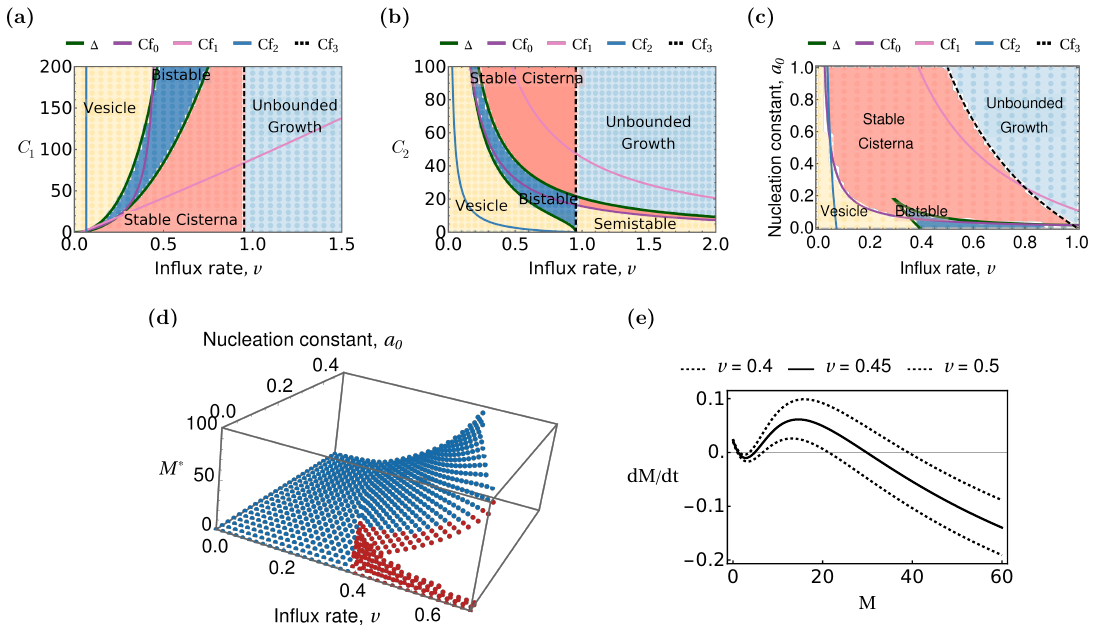}
\caption{\textbf{Phase diagram and boundaries for single cisterna.} (a,b,c) Solutions of  a single cisterna system Eq.\,\eqref{eq:sveq_H1} can be classified based on the stability of the fixed points -- stable, semistable, bistable and unbounded. A bistable system has two stable fixed point solutions (see Fig.\,\ref{fig:golgi_flux_123}(d)). Stable phases can be further classified based on the value of cisterna size, into  vesicle phase (M $< 1$, or a given threshold $M_0$ and stable fixed point, yellow region) and cisterna phase (M $\ge 1$ and stable fixed point). A semistable phase, that correspond to a scenario when fusion and fission kernels intersect twice (Fig.\,\ref{fig:golgi_flux_123}(c)), and the system is stable or unstable depending on the initial condition. At higher influx rate $v$, there are no fixed points, and the system grows unboundedly (light blue region). The analytic phase boundaries can be computed from the cubic discriminant $\Delta$ and the coefficients of cubic polynomial Eq.\,\eqref{eq:cubic}. Default parameter values, $C_1= 100$, $C_2=40$,  $d= 1$, $a_0=0.05$ where $v$ is the influx rate and $C_1$ is the Hill saturation constant for fusion, $C_2$ is the Hill saturation constant for fission, and $a_0$ is the nucleation constant. (a) Phase diagram in $(v,C_1)$ plane. (b) Phase diagram in $(v,C_2)$ plane.
(c) Phase diagram in $(v,a_0)$ plane. For a given influx rate $v$, a minimum value of the nucleation constant $a_0$ is needed for formation of a stable cisterna.
(d) Fixed point structure of Eq.\,\eqref{eq:cubic} can also be represented as a catastrophe curve by plotting the surface $\Sigma_f$ from Eq.\,\eqref{eq:cats_sigma}. Number of fixed points change as one moves in ($v,a_0)$) plane. Red plot points represent parameter regime with multistability whereas blue plot points represent parameter regime with single stable fixed point. Parameter values,  $C_1= 100$, $C_2=40$, $d= 1$. (e) Dynamical system Eq.\,\eqref{eq:cubic} ensures neighborhood stability, i.e., the fixed point point structure of the system remains invariant under small changes of parameters. Parameter values, $C_1= 100$, $C_2=40$,  $d= 1$, $a_0=0.05$. 
}
\label{fig:p_roots1}
\end{figure}
The roots of the above equation Eq.\,\eqref{eq:sveq_H1} are given by the polynomial equation,
\begin{small}
\begin{equation}
 C_1 C_2 a_0 v  + C_1 ( a_0 v - d)  M + C_2 (v + a_0 v ) M^2 + ( v + a_0 v  - d)M^{3} = 0 \,.
\label{eq:cubic}
\end{equation}
\end{small}
We use this polynomial equation to study the nature of intersections of the fission and fusion kernels, i.e., existence and stability of the fixed points for $M$,  by varying the effective parameters $v$ (influx rate),  $C1, \,C_2$ (levels of fisogens and fisogens), $a_0$ (nucleation constant) and $d$ (exit rate). This allows us to construct a phase diagram. As shown in Fig.\,\ref{fig:p_roots1}(a-c), there are stable phases corresponding to one or two stable fixed points, semi-stable phase corresponding to one stable and one unstable fixed point, with asymptotic behaviour depending on the initial condition, and unbounded growth corresponding to no stable fixed points. Phases can be further classified based on the value of the steady state (fixed point) -- vesicle phase ($M  \leq M_{0}$, a threshold size and stable fixed point) and stable cisterna phase ($M  \ge M_{0}$ and stable fixed point), see Fig.\,\ref{fig:p_roots1}(a-c). 


 Phase boundaries can be constructed analytically by applying the transformation $M \rightarrow M+M_0$, to the cubic polynomial Eq.\,\eqref{eq:cubic} -- and then analysing its coefficients ($Cf_i$'s) and the cubic discriminant ($\Delta$) of the transformed polynomial (Fig.\,\ref{fig:p_roots1}(a-c)). To segregate the vesicle phase from the cisterna phase, we have taken $M_0 =1$, the unit vesicle size, which gives the transformed polynomial,
\begin{small}
\begin{eqnarray}
&& \underbrace{(-d_1 - C_1 d_1 + v + a_0 v + a_0 C_1 v + C_2 v + a_0 C_2 v + a_0 C_1 C_2 v)}_{Cf_0} \, + \, M \underbrace{(-3 d_1 - C_1 d_1 + 3 v + 3 a_0 v + a_0 C_1 v + 2 C_2 v + 2 a_0 C_2 v)}_{Cf_1} \nonumber \\ && +  \, M^2 \underbrace{(-3 d_1 + 3 v + 3 a_0 v + C_2 v + a_0 C_2 v)}_{Cf_2} \, + \, M^3 \underbrace{(-d_1 + v + a_0 v)}_{Cf_3} = 0\,.
\label{eq:transform_cf}
\end{eqnarray}
\end{small}
If the coefficient $Cf_0 <0$, the above polynomial cannot have a root for $M\leq 1$ (as the original polynomial Eq.\,\eqref{eq:cubic} will only cross the y-axis for $M>1$). This is also dictated by the contour line, $Cf_0 = 0$ (thick purple line) in  Fig.\,\ref{fig:p_roots1}(a-c). Furthermore, Fig.\,\ref{fig:p_roots1}(c) suggests that a minimum value of the nucleation constant $a_0$ is needed to move to the stable cisterna phase from the vesicle phase.

A closer inspection of Eq.\,\eqref{eq:cubic} reveals that the system has cubic (cusp) singularity. 
In fact, equations such as Eq.\,\eqref{eq:cubic} correspond to a stable family of functions (unfolding) under small perturbations of the parameters, Fig.\,\ref{fig:p_roots1}(e) \cite{thom}.
Above equation Eq.\,\eqref{eq:cubic} loses (or gains) a pair of roots through \textit{saddle-node} (SN) bifurcation \cite{strogatz} leading to change in the number of stable fixed points (one to two). A better way to represent such fixed-point structures is to look at catastrophe curves \cite{thom}, which is the study of the appearance of singularities as one moves continuously in parameter space (Fig.\,\ref{fig:p_roots1}(d)). The cubic polynomial Eq.\,\eqref{eq:cubic} above, can be mapped to canonical form of a cusp catastrophe curve (universal unfolding of $x^4$-singularity with codimesion $2$, i.e., the control parameters, see \cite{thom,golubitzky} for details). More elaborately, one needs to compute the following surfaces \cite{germs},
\begin{small}
\begin{eqnarray}
\text{Unfolding} &:&   f(M)   \in  \mathbb{F}_k(M,\pmb{\mu}), \text{polynomial of degree k} \nonumber \\
\Sigma_f &:& \{(M,\pmb{\mu}) |\;  \frac{\partial f(M)}{\partial  M} =0\}  \nonumber \\
\Delta_f &:& \{(M,\pmb{\mu}) \in \Sigma_f |\; \frac{\partial ^2f(M)}{\partial M^2} =0\}  \nonumber \\
D_f &:& \{(\pmb{\mu}) |\; \exists \; \;  M \; \text{s.t.} \;  \; \frac{\partial f(M)}{\partial M} =0 \; \text{and} \;  \frac{\partial  ^2f(M)}{\partial M^2} =0\}\,,
\label{eq:cats_sigma}
\end{eqnarray}
\end{small}
where $\mathbb{F}_k(M, \pmb{\mu})$ is a polynomial of degree-k with n-parameters $\pmb{\mu}$, which in case of cusp catastrophe is $4$ and $2$ respectively, $\Sigma_f$ is defined by the gradient and $\Delta_f$ by the Hessian determinant. $\Sigma_f$ for the single cisterna dynamical system is plotted in Fig.\ref{fig:p_roots1}(d).

\section{Dynamics of the two cisternae system\label{sec:si5}}
\subsection{Constructing dynamical system for two cisternae \label{sec:2cist_T}}
Here we derive the dynamical system for two cisternae. For this, we derive motivation from \ref{sec:Mastereqtwocisternae}, where we have discussed master equation for two cisternae, and introduced intercisternal time over which intercisternal vesicle flux is integrated via \textbf{Direct} or \textbf{Indirect} scheme. 

The intercisternal flux depends on the availability of cognate pairs of v-SNAREs and t-SNAREs at the donor and the acceptor cisterna respectively  \cite{Munro2004, Bonifacino2003, Robinson2004, Traub2009, Yu2010, Jahn2006, Wickner2008, DSouzaSchorey2006, Stenmark2009, itocis, mani}. Consider the fission flux $J^{fis}_{i,i+1}$ at time $t$, from the donor cisterna $i$ to target cisterna $i+1$. The appropriate v-SNAREs need to be transported along with the cargo vesicles  destined to the target cisterna, leading to a depletion of the SNARE-pool and a reduced probability for subsequent fission events at $i$, unless replenished by fusion events at $i$ from $i+1$ \cite{willett}.  This implies that  the fission flux from the cisterna $i$ to $i+1$ at time $t$ will depend on the local fisogen availability and the size $M_i$ of the donor cisterna $i$  as well as on the local fusgogen availability and size $M_{i+1}$ of the acceptor cisterna $i+1$ at earlier times within a time window, functionally represented as the flux $\mathcal{J}_{i,i+1} (\phi_i^x,M_i,t;\phi_{i+1}^y,M_{i+1},t-t_w)$. Here, $\phi$ variables keep track of SNARE identities $x,y$, and hence give a cisternal SNARE-matching condition while $t_w$ is intercisternal time window \cite{storrie,jackson,dmitrieff}, which is drawn from a given probability distribution $P(t_w)$.
Furthermore, the fusogen and fisogen availability may depend on the local cisterna size as well. Using same arguments, fission flux from cisterna $i+1$ to $i$ at time $t$ will depend on the local fisogen availability and the size of the donor cisterna $i+1$  as well as on the local fusgogen  availability and size of the acceptor cisterna $i$ at earlier times within a time window, functionally represented as the flux $\mathcal{J}_{i+1,i} (\phi_{i+1}^y,M_{i+1},t;\phi_i^x,M_i,t-t_w)$. Summed over the SNARE identities and over the intercisternal times $t_w$, gives the total intercisternal flux between two cisterna, i.e.,
\bea
\mathcal{J}^{fis}_{i,i+1} (M_i,M_{i+1},t) &=& \sum_{x,y} \int dt_w \, P(t_w) \, \mathcal{J}_{i,i+1} (\phi_i^x,M_i,t;\phi_{i+1}^y,M_{i+1},t-t_w) \\
\mathcal{J}^{fis}_{i+1,i} (M_i,M_{i+1},t) &=& \sum_{x,y} \int dt_w \, P(t_w) \, \mathcal{J}_{i+1,i} (\phi_{i+1}^y,M_{i+1},t;\phi_i^x,M_i,t-t_w) \,.
\eea
Similarly, fusion flux from cisterna $i$ to $i+1$ depends on the fisogen availability at cisterna $i$  at earlier times and so on. The above integrals can be simplified by substituting mean fusion and fission kernels, $\mathcal{J}_{i,i+1}^{fis} (\phi_i^x,M_i,t;\phi_{i+1}^y,M_{i+1},t-t_w) = \phi_i^x(t) \,{\cal K}^{\sfi} ( M_i(t)) \, \phi_{i+1}^y(t-t_w) \, {\cal K}^{\sfu} (M_{i+1}(t-t_w))$. We further assume $P(t_w)$ to be uniform over an interval $(0,T)$ and a SNARE matching condition $\sum_{x,y} \phi_{i}^x(t) \phi_{i+1}^y(t-t_w)  = C_{i,i+1}(t,t-t_w) \, \delta^{x,y}$, $C_{i,i+1}$ is a constant that encodes the coupling strength for different cargo identities and can depend on the intercisternal time window $t_w$. One can also have more intricate matching conditions, depending on whether one or both homotypic and heterotypic fusions are allowed \cite{Vagne2018}. These assumptions lead to mean field integro-differential equations for the cisternal sizes $M_i$.
\bea
\mathcal{J}_{i,i+1}^{fis} (M_i,M_{i+1},t)&=&\int_{t_w=0}^T dt_w  \, C_{i,i+1}(t,t-t_w) \,{\cal K}^{\sfi} ( M_i(t)) \, \,{\cal K}^{\sfu} (M_{i+1}(t-t_w)) \\
\mathcal{J}_{i+1,i}^{fis}(M_i,M_{i+1},t) &=&  \int_{t_w=0}^T dt_w  \, C_{i,i+1}(t,t-t_w) \,{\cal K}^{\sfi} ( M_i(t-t_w)) \, \, {\cal K}^{\sfu} (M_{i+1}(t))\,,
\eea
 With the further assumption that ``size'' variables vary slowly over the intercisternal time window $t_w$, these integro-differential equations transform into ordinary differential equations.
\bea
J_{i,i+1}^{fis} (M_i,M_{i+1},t) &=& C \, T \, {\cal K}^{\sfi} ( M_i(t)) \, {\cal K}^{\sfu} (M_{i+1}(t)) \\
J_{i+1,i}^{fis} (M_i,M_{i+1},t) &=&  C \, T  \, {\cal K}^{\sfi} ( M_i(t)) \, {\cal K}^{\sfu} (M_{i+1}(t))\,,
\eea
where $C_{i,{i+1}} =  C$ is assumed to be a constant over cisternae identities and the time window. Following the same steps, similar expressions can be derived for $J_{i,i+1}^{fus},J_{i+1,i}^{fus}$. Using the functional forms of the fusion and fission kernels derived for the single cisterna, Eq.\,\eqref{eq:veq_H}, the balance of fluxes at the cisterna $1$ and $2$, with the above arguments gives
\begin{small}
\begin{eqnarray}
 \label{eq:retrograde_si1}
\dot{M}_1 
&=& v^{fus}_m \left(a_1^{fus} + \frac{ \kappa_{11} M_1(t)^2}{C_{11} + M_1(t)^2}\right) - \int_{t-T_1}^{t} dt' \, \left[\frac{v^{fis}_m \kappa_{12}  \, M_1 (t)}{C_{12}  + M_1(t)}\left(a_2^{fus} + \frac{  \kappa_{21} M_2(t')^2}{C_{21}+M_2(t')^2}\right) \right]\nonumber \\ &&+ \,  \int_{t-T_2}^{t} dt' \, \left[\frac{v'^{fis}_m \kappa_{22}  \, M_2(t')}{C_{22} + M_2(t')}\left(a_1^{fus} + \frac{\kappa_{11} \, M_1(t)^2}{C_{11} + M_1(t)^2}\right) \right] - \, \frac{v^{fis}_m \kappa_{12}  \, M_1(t)}{C_{12}  + M_1(t)}   \\
\dot{M}_2
 &=&  \int_{t-T_3}^{t} dt' \, \left[\frac{v^{fis}_m \kappa_{12}  \, M_1(t')}{C_{12}  + M_1(t')}\left(a_2^{fus} + \frac{ \kappa_{21} M_2(t)^2}{C_{21}+M_2(t)^2}\right)\right]  \nonumber \\ &&- \int_{t-T_4}^{t} dt' \, \left[\frac{v'^{fis}_m\kappa_{22}  \, M_2 (t)}{C_{22(t)} + M_2}\left(a_1^{fus} + \frac{\kappa_{11} \, M_1(t')^2}{C_{11} + M_1(t')^2}\right)\right]  -\, \frac{v'^{fis}_m \kappa_{22}  \, M_2(t)}{C_{22} + M_2(t)} \,,
 \label{eq:retrograde_si2}
\end{eqnarray}
\end{small}
where there are intercisternal fluxes, and we also include leak fluxes.
$v^{fus}_m$, $v'^{fis}_m$ are characteristic cisterna size scales set by fusion and fission enzymes (see Eqs.\,\eqref{eq:kernels_analyt_form2_Ckz},\eqref{eq:kernels_analyt_form2_C1kz}) and 1/$\kappa_{ij}$'s are the time scales associated with fusion and fission kernels at the cisterna $i=1,2$ ($j=1$ for fusion and $j$=2 for fission). 
$T_1,T_2,T_3,T_4$ are intercisternal times over which mass is transferred from from one cisterna to  another (or glycosylation enzymes are delivered to the cisterna via retrograde transport \cite{storrie,jackson}, which might act as a time marker).
By scaling the integration time w.r.t. $\kappa_{ij}$, one can see that if the intercisternal times are much smaller than fission and fusion times scales, $T_1,T_2,,T_3,T_4  \ll 1/\kappa_{ij}$, the integral terms might vanish and there is no  transfer of material from cisterna $1$ to $2$. Hence, to have a bounded as well as non-trivial contribution, the inverse time scale $k_{ij}$ and the window $T_1,T_2,T_3,T_4$ must be properly balanced so that their product $\lambda_{\alpha} = k_{ij} \, T_{\alpha}$ remains finite. This is physically plausible, as $T_{\alpha}$ can be seen to encode unaccounted degrees of freedom (such as spatial variables) while $k_{ij}$ increases as these degrees of freedom are condensed (i.e., $T_{\alpha}\rightarrow0$). 
Furthermore, the condition $T_{\alpha} \ll 1/\lambda_{\alpha}$ ensures that the integration window is small relative to the characteristic evolution time of the system. Under these conditions, the memory (time delay) effects become negligible (can also be seen from taylor expansion of the integral), and the integro-differential dynamics reduce to a standard ODE.
\begin{small}
\begin{eqnarray}
\label{eq:retrograde1_si}
\dot{M}_1 &=& v \left(a_1 + \frac{ M_1^2}{C_{11} + M_1^2}\right) -\frac{ d_{1}  \, M_1}{C_{12}  + M_1}  - \frac{d_{12} \, M_1}{C_{12}  + M_1}\left(a_2 + \frac{ M_2^2}{C_{21}+M_2^2}\right) + \frac{d_{21}\,M_2}{C_{22} + M_2}\left(a_1 + \frac{ M_1^2}{C_{11} + M_1^2}\right)  \\ 
\dot{M}_2 &=&   \frac{d_{12} \,M_1}{C_{12}  + M_1}\left(a_2 + \frac{M_2^2}{C_{21}+M_2^2}\right) -\frac{ d_{21}  \, M_2}{C_{22} + M_2}\left(a_1 + \frac{ M_1^2}{C_{11} + M_1^2}\right)    - \frac{ d_{2} \, M_2}{C_{22} + M_2}  \,,
\label{eq:retrograde2_si}
\end{eqnarray}
\end{small}
where similar to the one cisterna case, we use fusion and fission cycle time scales  $\kappa_{ij}$'s to write the above equation in terms of influx rates, intercicsternal rates, exit rates, and leak rates at cisterna $1$ and $2$. For example, $v  \equiv  \, v^{fus}_m \, \kappa_{11}$ is the influx rate and the grouped variables, $d_{1} \equiv  \, v^{fis}_m \, \kappa_{12}$ is the leak rate at the cisterna $1$ and so on. The Peak rates for anterograde and retrograde intercisternal flux transfer are parameterized by $d_{12}$ and $d_{21}$ respectively along with the exit rate $d_{2}$ at cisterna $2$. As in the single cisterna case, We set $d_1=1$ to write these equations in dimensionless form. This dynamical system can be extended to a system of multiple cisterna by using the \textit{flux matrix} ${\cal K}$ with terms corresponding to fission, fusion and intercisternal transfer at cisterna $i$,
\bea
 {\cal K}^{fus} (M_i) &&= v_i\,\left(a_{i}+\frac{M_i^2}{C_{i1} + M_i^2}\right) \label{eq:fusi} \\  {\cal K}^{fus} (M_i) &&= \frac{d_{i} \, M_i}{C_{i2} + M_i} \label{eq:fisi}  \\ {\cal K}^{tr}_{a} (M_i) &&= \frac{d_{i2} \, M_i}{C_{i2} + M_i}\left(a_{i+1}+\frac{M_{i+1}^2}{C_{(i+1)1} + M_{i+1}^2}\right) \label{eq:trai}  \\ {\cal K}^{tr}_{r} (M_i) &&= \frac{d_{(i+1)2} \, M_{i+1}}{C_{(i+1)2} + M_{i+1}}\left(a_{i}+\frac{M_i^2}{C_{i1} + M_i^2}\right) \,,\label{eq:trri} 
\eea
where fusion and retrograde fluxes ${\cal K}^{tr}_{r}$ contribute positively whereas fission and anterograde fluxes ${\cal K}^{tr}_{a}$ contribute negatively and $v_i$ can be put to zero for $i>1$.

\subsection{Existence of limit cycle solutions in a flux system \label{sec:prop_reduced}} 

We prove the Proposition \ref{prop:prop} stated in the main text and discuss related results. In the processes, we derive the minimal structure needed for existence of limit cycle solutions in a flux system.
\begin{prop}\label{prop:prop1}
Consider the general flux system,
\begin{eqnarray}
\label{eq:yx_si}
 \dot{M_1} &=& J_{in}(M_1) - \mathcal{J}^{fis}_{12} (M_1,M_2) + \, \mathcal{J}^{fus}_{21} (M_1,M_2) \\
 \dot{M_2} &=&  \,  \mathcal{J}^{fus}_{12} (M_1,M_2) - \mathcal{J}^{fis}_{21} (M_1,M_2) - J_{ex}(M_2) \,,
  \label{eq:xy_si}
\end{eqnarray}
where $J_{in}$ and $J_{ex}$ are the influx and the exit flux, respectively and the rest are intercisternal fluxes, all of which are positive for $M_{1,2} \geq 0$. We consider the following two cases: 
\begin{enumerate}[label=(\roman*)]
\item 
If the intercisternal fluxes
$\mathcal{J}^{fis}_{12}(M_1,M_2) = \mathcal{J}^{fis}_{12}(M_1)$, $\mathcal{J}^{fus}_{21}(M_1,M_2) = \mathcal{J}^{fus}_{21}(M_1)$, $\mathcal{J}^{fus}_{12}(M_1,M_2) = \mathcal{J}^{fus}_{12}(M_2)$ and $\mathcal{J}^{fis}_{21}(M_1,M_2) = \mathcal{J}^{fis}_{21}(M_2)$, the above dynamical system can only have real eigenvalues. 
\item Let the intercisternal fluxes are dependent on the size of the donor cisterna alone, i.e., $\mathcal{J}^{fis}_{12}(M_1,M_2) = \mathcal{J}^{fis}_{12}(M_1)$, $\mathcal{J}^{fus}_{21}(M_1,M_2) = \mathcal{J}^{fus}_{21}(M_2)$, $\mathcal{J}^{fus}_{12}(M_1,M_2) = \mathcal{J}^{fus}_{12}(M_1)$ and $\mathcal{J}^{fis}_{21}(M_1,M_2) = \mathcal{J}^{fis}_{21}(M_2)$. If $ \mathcal{J}^{fus}_{12}(M_1),  \,\mathcal{J}^{fus}_{21}(M_2)$ are co-monotonic as functions of their arguments, the above dynamical system can only have real eigenvalues and hence cannot have closed orbit solutions. We refer to such a flux systems as a {\it reduced flux system}. 
\end{enumerate}
\end{prop}
\begin{proof}
The proof to (i) is straightforward as it leads to two independent one dimensional system, that can only have real eigenvalues. 

Proof of (ii) is as follows: in the reduced form (ii), above equation modifies to,
\begin{eqnarray}
 \dot{M_1} &=& J_{in}(M_1) - \mathcal{J}^{fis}_{12} (M_1) + \, \mathcal{J}^{fus}_{21} (M_2) \\
 \dot{M_2} &=&  \,  \mathcal{J}^{fus}_{12} (M_1) - \mathcal{J}^{fis}_{21} (M_2) - J_{ex}(M_2) \,.
\end{eqnarray}
The Jacobian matrix around any fixed point of the above system will have the form
\begin{equation*}
J = \begin{pmatrix}
a-k & d \\
k_1 & -c
\end{pmatrix}\,,
\label{eq:jacb}
\end{equation*}
where $a = \frac{\partial J_{in}}{\partial M_1}, k = \frac{\partial \mathcal{J}^{fis}_{12}}{\partial M_1}, k_1 =  \frac{\partial \mathcal{J}^{fus}_{12}}{\partial M_1}, c = \frac{\partial (\mathcal{J}^{fis}_{21} + J_{ex})}{\partial M_2}, d = \frac{\partial \mathcal{J}^{fus}_{21}}{\partial M_2} $. Since  $\mathcal{J}^{fus}_{12}(M_1)$  and $\mathcal{J}^{fus}_{21}(M_2)$ are co-monotonic, the above system $J$ can be transformed to a non-negative matrix $M =\pm J +\beta I$, $I$ being the identity matrix and $\beta$, a large enough positive number to make M non-negative. This implies from Perron-Frobenius theorem \cite{farina2000positive}, that $J$ has at least one real eigenvalue. Since complex roots occur in pairs, for the above reduced two-dimensional flux system, all eigenvalues are real.
\end{proof}
\begin{remark}
For our purpose, Perron-Frobenius theorem \cite{farina2000positive} states that if $J$ is a non-negative irreducible (strongly connected) matrix, then it has at least one positive real eigenvalue.
\end{remark}
\begin{remark}
The above proposition \ref{prop:prop1} for the $2$-cisternae system can also be proved in a simpler way by computing the discriminant of the characteristic polynomial for the Jacobian $J$. For Jacobian Eq.\,\eqref{eq:jacb}, the discriminant is $(a+c-k)^2 + 4 \,d \, k_1$, which is always positive, since  $\mathcal{J}^{fus}_{12}(M_1)$  and $\mathcal{J}^{fus}_{21}(M_2)$ are co-monotonic (sign($d$) = sign($k_1$)). Hence, the reduced system above can only have real eigenvalues.
\end{remark}
\begin{figure}[t!]
\centering

\includegraphics[width=0.35\textwidth]{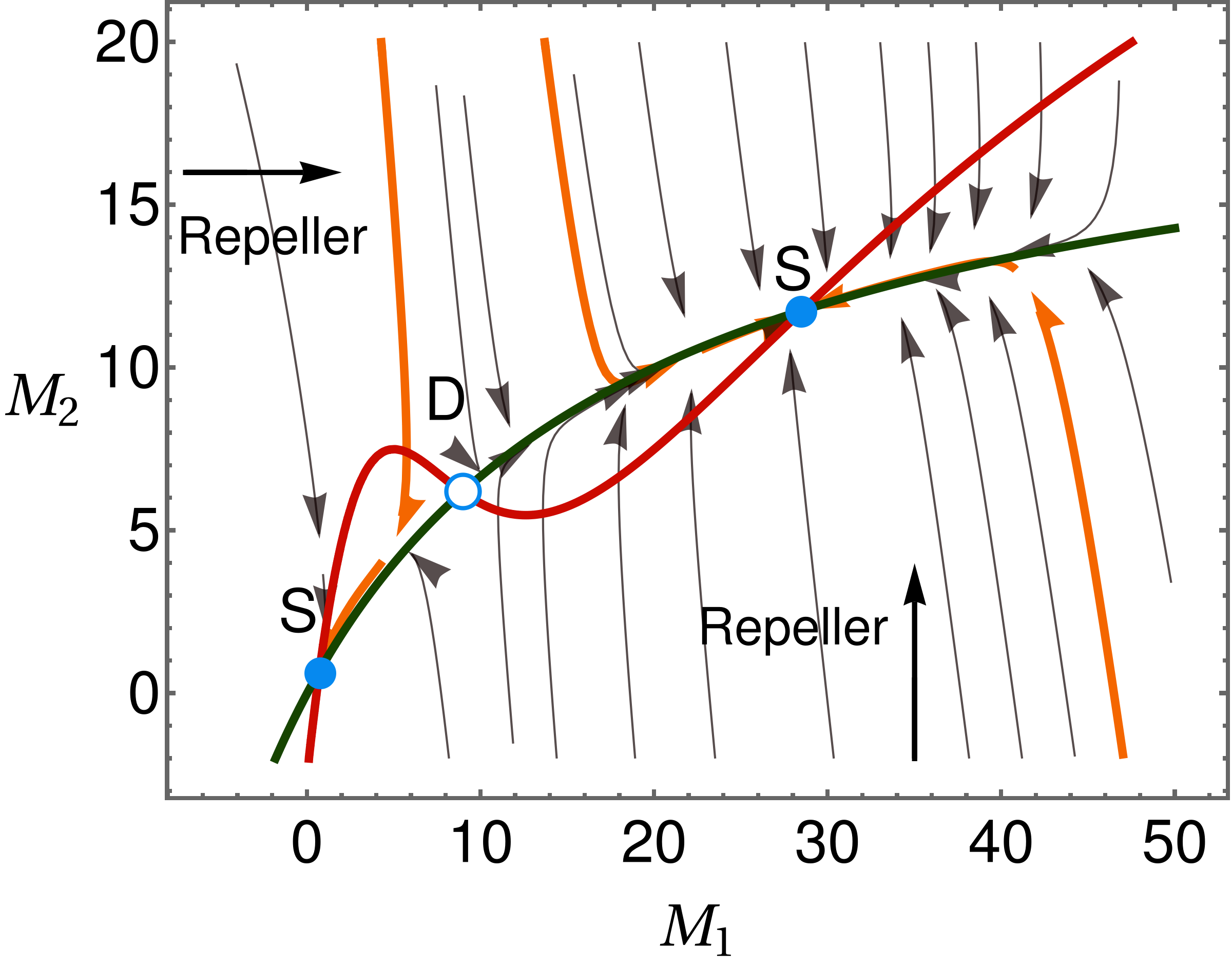}
\caption{\textbf{A positive system with only real eigenvalues.} Dynamics of system defined by Eqs.\eqref{eq:retrograde1_si},\eqref{eq:retrograde2_si} (or Eqs.\eqref{eq:retrograde1},\eqref{eq:retrograde2} in the main text) remains positive for all the time ($v>0$) for positive initial conditions, with positive $M_1$ and $M_2$ axes acting as repellers. For the reduced system (Proposition \ref{prop:prop1}), the fixed points (intersections of $M_1$ (red) and $M_2$ (green) nullclines), correspond to real eigenvalues. Fixed points are shown -- stable (S, filled circle) and saddle (D, open circle). A few trajectories with different initial conditions are also shown (thick orange line).}
\label{fig:positive_sys}
\end{figure}

\begin{remark}
Proposition \ref{prop:prop1} provides a way to reduce a general flux system to one that only admits real eigenvalues \text{(see Fig.\,\ref{fig:positive_sys})}.
\end{remark}
\begin{remark}
Adding the intercisternal fluxes that depend both on $(M_1,M_2)$, can facilitate the existence of complex roots, and hence the appearance of limit cycles. In the main text, we have segregated the anterograde fission flux as, $\mathcal{J}^{fis}_{12} (M_1,M_2) = \mathcal{J}^{leak}_{12} (M_1) + \mathcal{J}^{fus}_{12} (M_1,M_2)$ and retrograde fission flux as, $\mathcal{J}^{fis}_{21} (M_1,M_2) = \mathcal{J}^{leak}_{21} (M_1) + \mathcal{J}^{fus}_{21} (M_1,M_2)$ and we absorb the leak flux, $\mathcal{J}^{leak}_{21} (M_1)$  to the exit flux, $J_{ex}(M_2)$ at cisterna $2$.
\end{remark}
\begin{corollary}\label{cor:1main}
Let us consider the system Eqs.\,\eqref{eq:yx_si},\eqref{eq:xy_si} with anterograde flux depending on both $M_1,M_2$. This system can admit limit cycle solutions.
\end{corollary}
\begin{proof}
The above dynamical system can be effectively written in the form,
\begin{small}
\begin{eqnarray}
\dot{M_1} &=& F(M_1) - H(M_1,M_2)  \nonumber \\
\dot{M_2} &=&  H(M_1,M_2) - G(M_2) \nonumber \,.
\end{eqnarray}
\end{small}
The Jacobian matrix takes the form,
\begin{eqnarray}
J = \begin{pmatrix}
 \frac{\partial F}{\partial M_1}-\frac{\partial H}{\partial M_1} & -\frac{\partial H}{\partial M_2} \\
 \frac{\partial H}{\partial M_1}&  \frac{\partial H}{\partial M_2} -\frac{\partial G}{\partial M_2}
\end{pmatrix}\,.
\label{eq:Jacob_M1M2_retro}
\end{eqnarray}
Here, if $H(M_1,M_2)$ is co-monotonic in $M_1$ and $M_2$, the off-diagonal elements of the Jacobian have opposite signs. More explicitly, the Jacobian can have a sign signature $\begin{pmatrix} + & - \\ + & -  \end{pmatrix}$ or $\begin{pmatrix} + & + \\ - & - \end{pmatrix}$, and therefore the characteristic polynomial of the Jacobian can have a negative discriminant leading to complex roots and the possibility of limit cycles \cite{tyson}.
\end{proof}
\begin{corollary}\label{cor:original_name}
 For the above system, Eq.\,\eqref{eq:yx_si},\eqref{eq:xy_si}, let the anterograde fluxes depend upon the size of the donor cisterna alone, i.e., $\mathcal{J}^{fis}_{12}(M_1,M_2)  = \mathcal{J}^{fis}_{12}(M_1)$, $\mathcal{J}^{fus}_{12}(M_1,M_2)  = \mathcal{J}^{fus}_{12}(M_1)$, whereas the retrograde fluxes  $\mathcal{J}^{fus}_{21}(M_1,M_2),\,\mathcal{J}^{fis}_{21}(M_1,M_2)$ depend on both $M_1,M_2$. This system generally has real eigenvalues.
\end{corollary}
\begin{proof}
Without loss of generality, we can assume that $\mathcal{J}^{fis}_{21}(M_1,M_2) = \mathcal{J}^{fus}_{21}(M_1,M_2)$. For the above assumption, the Jacobian matrix has the form, 
$$
J = \begin{pmatrix}
a-k + l & d \\
k_1 - l & -c
\end{pmatrix}\,,$$
where $l = \partial \mathcal{J}^{fis}_{21}/\partial M_1$ and all other entries are as above. The discriminant for the characteristic polynomial for this Jacobian is  $(a + c - k + l)^2  + 4 \,d \, (k_1 - l) $, which is mostly positive, unless $l$ is very large. Hence, this system generally has real eigenvalues.
\end{proof}
\subsection{Nullclines and bifurcations \label{sec:nullc2}}
The changes in the solution structure of one cisterna dynamical system can be captured completely by \textit{saddle-node bifurcations}. However, as the dimensionality of the system increases, other possibilities also arise. For instance, the stability of a fixed point might change via \textit{Hopf bifurcation}, and the system can admit closed orbit solutions \cite{strogatz}.

\begin{figure}[t!]
\centering
 \includegraphics[width=\textwidth]{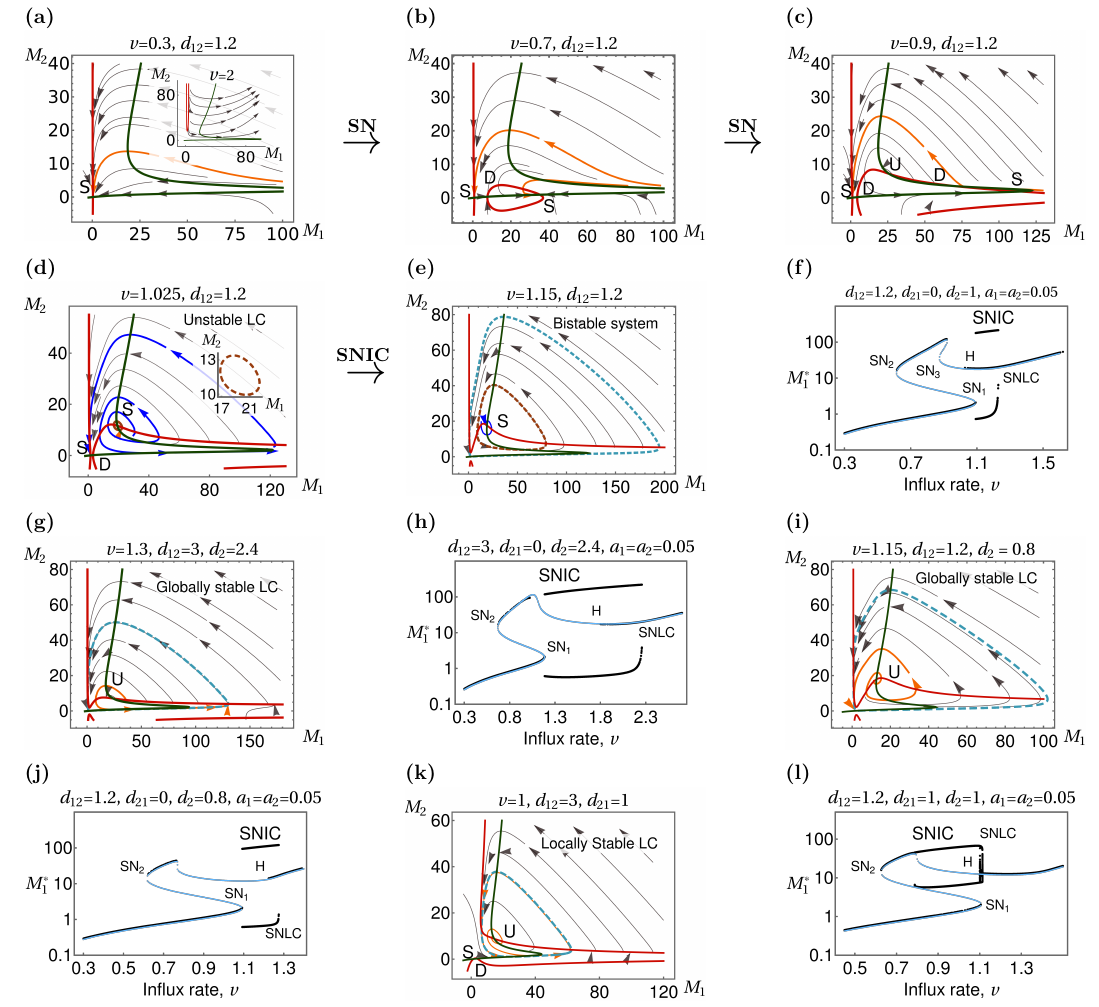}
\caption{{\bf Solution structure for two cisternae}. Solution structure for $2$-cisternae dynamical system Eqs.\,\eqref{eq:retrograde1_si},\eqref{eq:retrograde2_si} is characterized by the stability of the fixed points -- stable fixed point (S), saddle fixed point (D) and unstable fixed point(U) as well as via bifurcations -- saddle-node bifurcations (SN$_i$'s), Hopf bifurcation (H), SNIC bifurcation and SNLC bifurcation (see Glossary\,\ref{box:definitions} for details). As influx rate is increased from (a) to (c), the system undergoes a series of saddle-node bifurcations, thereby altering the number of fixed points. (a-inset) For high influx rate ($v= 2$), system has no fixed points and it grows unboundedly. (d) A fixed point can change its stability via Hopf bifurcation, and an unstable local limit cycle arises. (e)  As $v$ increases, a saddle-node pair near zero annihilate each other and since they lie on a close loop, it gives rise to an outer limit cycle solution via SNIC (light blue, dashed). Further increasing $v$, the inner unstable limit cycle (brown, dashed) grows in size and annihilates with the outer stable limit cycle via SNLC and the system is left with a stable fixed point. (f) Bifurcation diagram depicting the series of events (a-e). Other parameters are kept fixed at, $C_{11} = C_{21} = 100$, $C_{12} = C_{22} = 20$, $a_1 = a_2  = 0.05$, $d_{2} = 1$, $d_{12} = 1.2$, $d_{21} = 0$. (g) For a different set of parameters, $d_{12} = 3$, $d_{2} = 2.4$, the generated limit cycle via SNIC bifurcation is globally stable, as the enclosed fixed point is unstable, as observed in the bifurcation diagram (h). The corresponding phases are shown in the phase diagram Fig.\ref{fig:overlap}(a). (i,j) A globally stable limit cycle via SNIC can also be obtained by lowering the value of exit rate, $d_{2} = 0.8$. (k,l) Adding a $M_1$-independent retrograde flux from cisterna $2$ to $1$ ($d_{21}=0.6$) modifies the nullclines and a locally stable limit cycle appears via SNIC.
}
\label{fig:2cist_roots}
\end{figure}

Unlike the one-cisterna case, finding the root structure of high-dimensional systems requires quite advanced mathematical tools (see \textit{Intersection theory}\cite{fulton}, for instance), and we will rely on graphical and numerical techniques. In this regard, we plot the \textit{nullclines}, which are the curves of zero dynamics of different state variables in the state space. The intersections of these curves yield the system's fixed points, as all the rates of change are simultaneously zero at those locations. We draw the nullclines for the system Eqs.\,\eqref{eq:retrograde1_si},\eqref{eq:retrograde2_si} (or Eqs.\,\eqref{eq:retrograde1},\eqref{eq:retrograde2} in the main text) in Fig.\,\ref{fig:2cist_roots} and  analyse how it changes as we vary the system parameters.

Let us discuss the detailed root structure (and the phases) for the $2$-cisternae dynamical system. Suppose we fix the values of the parameters at, $C_{11} = C_{21} = 100, C_{12} = C_{22} = 20 $ $ , a_1 = a_2  = 0.05, d_{2} = 1$, $d_{12} = 1.2$, $d_{21} = 0$ and draw the flows and nullclines as we increase the influx rate $v$, see Fig.\,\ref{fig:2cist_roots} (a-f). With these parameter values and small influx rate, $v$, there is only one stable fixed point near zero, i.e., there is no cisterna formation. As the influx rate increases, a pair of fixed points appear via \textit{saddle node bifurcation}, and the number of fixed points goes from $1\rightarrow 3 \rightarrow 5$ in Fig.\ref{fig:2cist_roots} (a-c). As $v$ increases further, the second cisterna starts to grow in size and the fixed points disappear via another \textit{saddle node bifurcation}. At this point, locally unstable limit cycle solution appear via \textit{Hopf bifurcation} (Fig.\ref{fig:2cist_roots} (d), shown as brown dashed line). Thereafter, increments in $v$, leads to disappearance of fixed points via \textit{saddle node bifurcation} near zero, and if the saddle-node pair happens to lie on a closed loop (facilitated by the positivity and nonlinearity of flux kernels), their disappearance leads to the rise of limit cycles via \textit{SNIC bifurcation} in  Fig.\,\ref{fig:2cist_roots}(e) (shown as light blue dashed line) and we have a nested bistable limit cycle solution with stable central fixed point. On further increasing $v$, the inner unstable circle Fig.\,\ref{fig:2cist_roots}(e) keeps on growing in size, eventually coalescing and annihilating with the outer stable limit cycle via {saddle-node bifurcation of limit cycles} (SNLC) and the system is left with a stable fixed point. The series of events (a-e) is depicted in a comprehensive way in  the bifurcation diagram, Fig.\,\ref{fig:2cist_roots}(f).

Note that for the system Eqs.\,\eqref{eq:retrograde1_si},\eqref{eq:retrograde2_si}, the stability of the limit cycle solution generated via SNIC bifurcation depends upon the fixed point structure remaining after the saddle-node annihilation. Crucially, if only one enclosed fixed point remains after the bifurcation, the stability of the limit cycle generated via SNIC can be determined by the stability of the enclosed fixed point and the global flows in the system. The stability of the obtained limit cycle solutions changes if other parameters in the system are varied. For example, if we set $d_{12}=3, d_2=2.4$ (Fig.\,\ref{fig:2cist_roots}(g,h), then the SNIC bifurcation precedes the Hopf bifurcation as we increase the influx rate, the enclosed fixed point is unstable and the obtained limit cycle is globally stable, as shown in the bifurcation diagram, Fig.\,\ref{fig:2cist_roots}(h). A globally stable limit cycle via SNIC bifurcation is also achieved by lowering the value of the exit rate $d_2$, Fig.\,\ref{fig:2cist_roots}(i,j).
Futhermore, adding a $M_1$-independent retrograde flux from cisterna $2$ to $1$ can lead to the formation of locally stable limit cycles via SNIC bifurcation, Fig.\ref{fig:2cist_roots}(k,l). 

Based on the insight above, in the main text, we have explored the following solution classes and the bifurcations leading to the transitions between them \cite{strogatz}:

\begin{enumerate}[label=(\alph*)]
\item Appearance or annihilation of a pair of fixed points via saddle node bifurcation.
\item Change in the  stability of a fixed point via Hopf bifurcation: This might lead to appearance of local oscillations.
\item Appearance of limit cycle via SNIC bifurcation: This might lead to the appearance of locally stable, bistable or globally stable oscillations. SNIC bifurcations are identified by a saddle-node (SN) ghost \cite{strogatz} near the point of saddle-node annihilation, where the system spends a significant amount of time compared to the rest of the trajectory, contributing significantly to the time period of the oscillation.
\item Coalescence via SNLC: The \textit{Saddle-Node Bifurcation of Limit Cycles} (SNLC) involves the collision and mutual annihilation of two nested limit cycles. Typically, an inner unstable cycle expands to merge with an enclosing stable cycle; upon coalescence, both periodic solutions vanish.  Analogous to the saddle-node bifurcation of fixed points, this transition is accompanied by a \textit{ghost delay} where trajectories linger in the vicinity of the "phantom" orbits before eventually drifting away.
\item Note that when the system has multistability, i.e., the  solutions are only locally stable, the asymptotic behaviour of the system depends on the initial conditions (Fig.\ref{fig:overlap}(d,e,f)), resulting in overlapping regions in the phase diagram over the parameter space of the system.
\end{enumerate}

\begin{figure}[t!]
\centering
 \includegraphics[width=\textwidth]{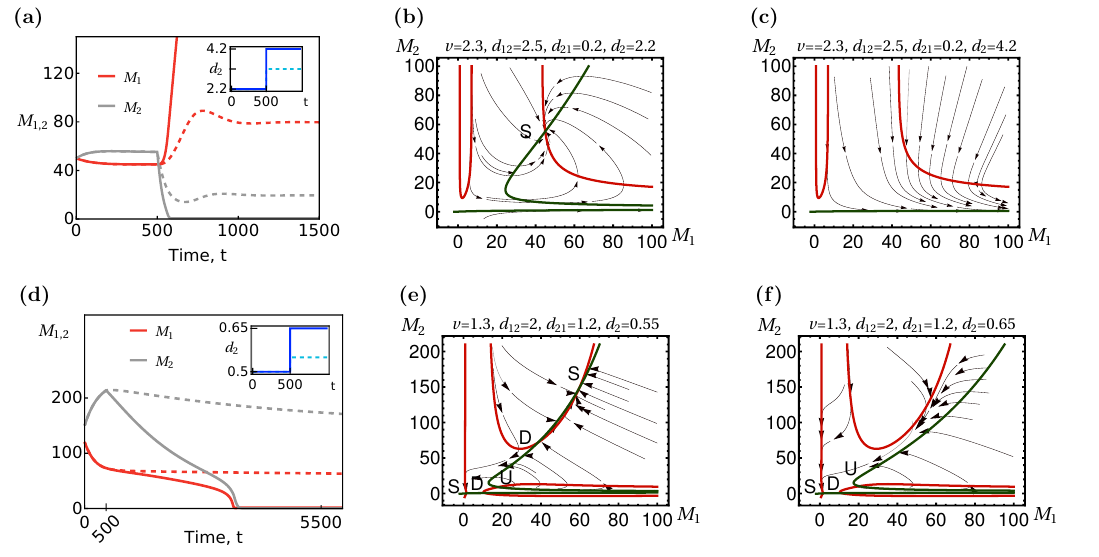}
\caption{{\bf Effect of retrograde flux on the response to systematic perturbations of the $2$-cisternae system.} (a) For the $2$-cisternae system with low retrograde rate $d_{21}$ compared to the anterograde rate $d_{12}$ (parameter values shown in (b,c)), if the exit rate at the second cisterna, $d_2$ is increased ($2.2 \rightarrow 4.2$)), the second cisterna is depleted. As the anterograde flux from the first cisterna increases with the size of the second cisterna (see Eqs.\,\eqref{eq:retrograde1_si},\eqref{eq:retrograde2_si}), depletion of the second cisterna causes the influx rate dominating the outflux rate at the first cisterna, leading to the increase in the size of the first cisterna and eventual unbounded growth. Solutions before (dashed line) and after (thick line) the perturbation is applied, are shown. The corresponding nullclines and flows are plotted, (b) $\rightarrow$ (c). (d) If the retrograde rate $d_{21}$ is comparable to the anterograde rate $d_{12}$  (parameter values shown in (c,d)), increasing the exit rate, $d_2$ leads to depletion of the second cisterna, as before. However, in this regime, even when the second cisterna is depleted by increasing $d_2$ ($0.55 \rightarrow 0.65$), outflux flux ($=$ anterograde flux $+$ leak flux) at the first cisterna can balance the influx rate, $v$ supporting a stable vesicle phase (see Eqs.\,\eqref{eq:retrograde1_si},\eqref{eq:retrograde2_si}). Solutions before (dashed line) and after (thick line) the perturbation is applied, are shown. The corresponding nullclines and flows are plotted, (e) $\rightarrow$ (f). Other parameters are kept fixed at $C_{11} = C_{21} = 100$, $C_{22} = 20$, $C_{12} = 20$, $d_{1} = 1$.}
\label{fig:2cist_retro_fig}
\end{figure}

\subsection{Effects of retrograde flux on the solution space of the two-cisternae system\label{sec:2cist_retro}}
To understand the effects of the retrograde flux on the solution space of $2$-cisternae system, let us look at the system of equation,  Eqs.\,\eqref{eq:retrograde1_si},\eqref{eq:retrograde2_si}. It can be written as,
\begin{small}
\begin{eqnarray}
\dot{M_1} &=& F(M_1) - H_1(M_1,M_2) + H_2(M_1,M_2)  \nonumber \\
\dot{M_2} &=&  H_1(M_1,M_2) - H_2(M_1,M_2)  - G(M_2) \,. \nonumber 
\end{eqnarray}
\end{small}
the Jacobian matrix takes the form,
\begin{eqnarray}
J = \begin{pmatrix}
 \frac{\partial F}{\partial M_1}-\frac{\partial H_1}{\partial M_1} + \frac{\partial H_2}{\partial M_1} & -\frac{\partial H_1}{\partial M_2} + \frac{\partial H_2}{\partial M_2} \\ \\
 \frac{\partial H_1}{\partial M_1} - \frac{\partial H_2}{\partial M_1}&  \frac{\partial H_1}{\partial M_2} - \frac{\partial H_2}{\partial M_2}-\frac{\partial G}{\partial M_2}
\end{pmatrix}\,.
\label{eq:Jacob_M1M2_retro2}
\end{eqnarray}
Note that in Eqs.\,\eqref{eq:retrograde1_si},\eqref{eq:retrograde2_si}, the retrograde flux $H_2(M_1,M_2) $ is an increasing function of the sizes of the first and the second cisterna, $M_1$ and $M_2$, respectively. We can see that adding retrograde flux might cause the off-diagonal elements of the Jacobian, Eq.\,\eqref{eq:Jacob_M1M2_retro} to have the same sign and therefore restricts the existence of limit cycle solutions, see Fig.\,\ref{fig:overlap}(c).

We have seen before that adding a $M_1$-independent retrograde flux from cisterna $2$ to $1$ can lead to the formation of a locally stable limit cycles via SNIC bifurcation (Fig.\,\ref{fig:2cist_roots}). Furthermore, adding a retrograde flux that is dependent on both $M_1,M_2$ adds to the nonlinearities in the system and in general, would lead to an increase in the number of fixed points. We discuss a consequence of this below.

In the main text, we have discussed the response of the $2$-cisternae system to systematic perturbations that involves the dissolution of the second cisterna by increasing the exit rate, $d_2$ at the second cisterna. We observed that the $2$-cisternae system has a different response based on the retrograde rate. We explain these solutions by analysing the nullclines and flows (Fig.\,\ref{fig:2cist_retro_fig}) in these two parameter regimes.
\begin{enumerate}[label=(\roman*)]
\item At first, let us consider the case with small retrograte rate. In this case, if the exit rate $d_2$ is large (keeping the influx rate, $v$ at a constant value), the second cisterna is depleted. However, as a result of this, not only the retrograde flux decreases, but the anterograde flux decreases as well, see Eqs.\,\eqref{eq:retrograde1_si},\eqref{eq:retrograde2_si}. This leads to the first cisterna becoming unstable, see Fig.\,\ref{fig:2cist_retro_fig}(a-c).
\item Now, consider the case with large retrograte rate. In this case, if the exit rate $d_2$ is large (keeping the influx rate, $v$ at a constant value), the second cisterna is depleted, same as case (i). However, stable phases for the $2$-cisternae system with large retrograde flux occur at smaller values of the influx rate, $v$. In this regime, if the retrograde flux is reduced by depleting the second cisterna, the anterograde flux and the leak flux at the first cisterna (see Eqs.\,\eqref{eq:retrograde1_si},\eqref{eq:retrograde2_si}) are large enough to support a stable vesicle phase. This effectively leads to depletion of the first cisterna, see Fig.\,\ref{fig:2cist_retro_fig}(d-f).
\end{enumerate}

\subsection{Arriving at the phase diagram and phase boundaries for the two-cisternae system\label{sec:RH}} 

\begin{figure}[t!]
\centering
\includegraphics[width=\textwidth]{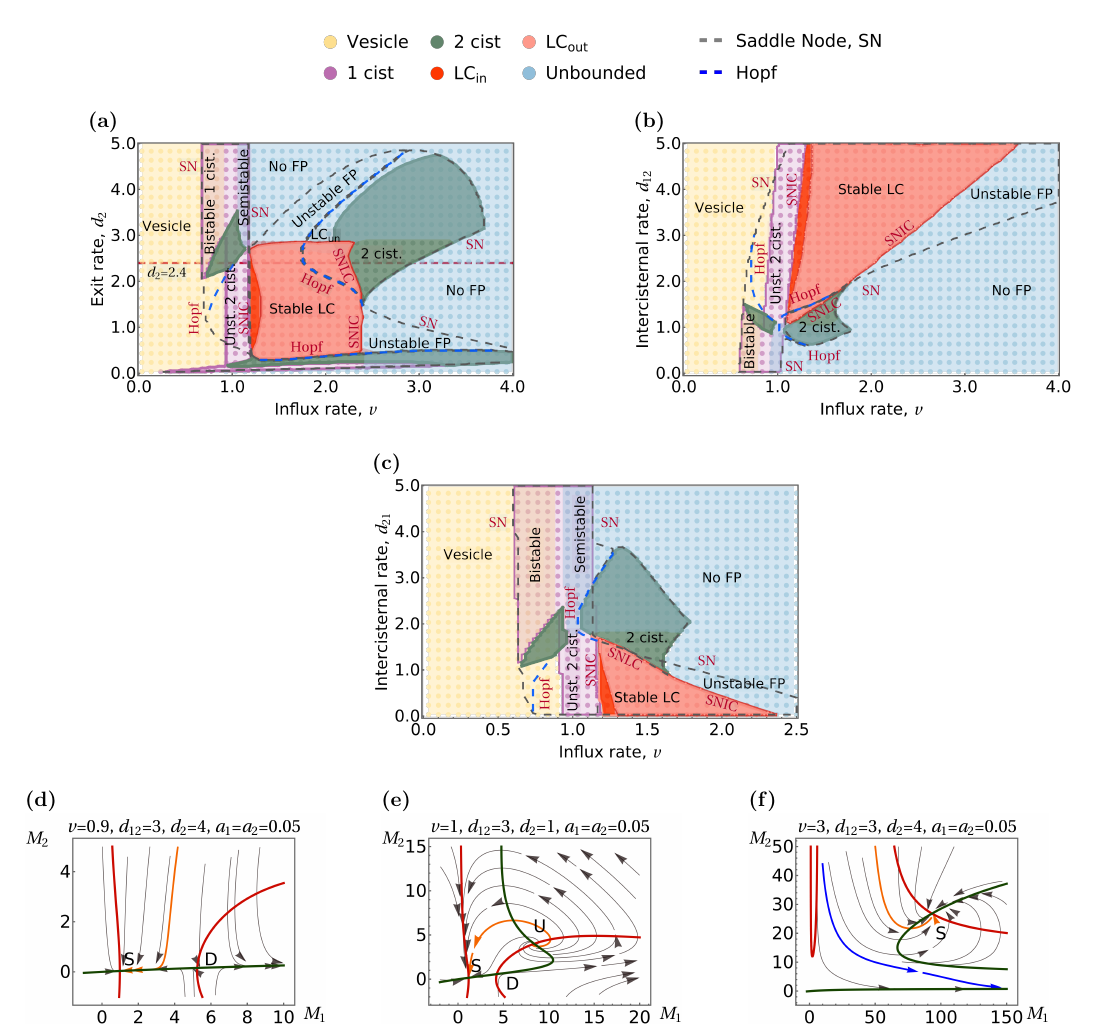}
\caption{\textbf{Phase diagram for two cisternae.}  (a) Phase diagram for two cisternae system in (influx rate $v$, exit rate $d_2$) plane. The horizontal line (dark red, dashed) corresponds to the bifurcation diagram in Fig.\ref{fig:2cist_roots}(h). Phase boundaries, or bifurcation lines, are constructed 
with saddle-node (SN) bifurcation indicated with gray dashed lines and 
Hopf bifurcation indicated with  blue dashed lines, see Table \ref{tab:bifurcation_criteria} \cite{strogatz, GH}. (b) Phase diagram for two cisternae system in (influx rate $v$, anterograde intercisternal flux rate $d_{12}$) plane. (c) Phase diagram for two cisternae system in (influx rate $v$, retrograde intercisternal flux rate $d_{21}$) plane. For the flux kernels used in Eqs.\,\eqref{eq:retrograde1_si},\eqref{eq:retrograde2_si} (and equivalently, Eqs.\,\eqref{eq:retrograde1},\eqref{eq:retrograde2}), oscillatory solutions are observed for small retrograde rate, $d_{21}$. For (a,b,c), default values of the parameters are,  $C_{11} = C_{21} = 100$, $C_{12} = C_{22} = 20$, $a_1 = a_2  = 0.05, d_{2} = 1$, $d_{12} = 3$, $d_{21} = 0$. (d,e,f) Nullclines and flows showing multistability, with mixed phases where system converges to one phase or the other depending on the initial condition. (d) semistable single cisterna phase, with $M_1>1, M_2<1$ (e) unstable two cisterna phase with stable fixed point lying at  $M_1>1, M_2<1$. (f) System grows unboundedly (blue) or goes to a stable fixed point (orange) depending on the initial condition. These cases give rise to mixed (overlapping) phases in the phase diagram (a,b,c) over the parameter space of the system. 
Note that in the phase  diagram (a,b,c), locally stable $2$-cisternae phase (green-blue overlap regions) as well as  globally stable $2$-cisternae phase are shown (dark green).}
\label{fig:overlap}
\end{figure}

As previously noted, this work focuses on both local (saddle-node, Hopf) and global (SNIC, SNLC) bifurcations. In this section, we present a method for predicting these transitions based on the eigenvalue analysis of the fixed points. Although local eigenvalue analysis alone cannot definitively confirm a global bifurcation, tracking the stability of these fixed points across the parameter space provides critical clues. Specifically, the manner in which fixed points change stability often indicates the underlying global bifurcation responsible for shifts in the solution structure.


We first establish the eigenvalue analysis for predicting saddle-node and Hopf bifurcations. For the $2$-cisternae dynamical system, $\dot{M_1} =f(M_1,M_2),\, \dot{M_2} = g(M_1,M_2)$, the Jacobian matrix around a fixed point ($M_1^*,M_2^*$) is given by
\begin{eqnarray}
J = \begin{pmatrix}
   \frac{df}{dM_1} &
   \frac{df}{dM_2}  \\
   \frac{dg}{dM_1} &
   \frac{dg}{dM_2} 
   \end{pmatrix} \Bigg|_{(M_1^*,M_2^*)} \,,
   \label{eq:jacob_bn}
 \end{eqnarray}  
which yields the characteristic equation $\lambda^2 - \text{Tr}(J) \lambda + \text{Det}(J)=0$, where $\text{Tr}(J)$ and $\text{Det}(J)$ denote the  trace and determinant of $J$, respectively. For the fixed point to be linearly stable (i.e., both eigenvalues have negative real parts), the system must satisfy $\text{Tr}(J) <0$ and $\text{Det}(J)>0$. A saddle-node bifurcation occurs when $\det(J) = 0$ provided $\text{Tr}(J) \neq 0$. Conversely, the steady state loses stability when the trace changes sign from negative to positive while $\text{Det}(J) >0$; in this regime, periodic solutions typically emerge via a Hopf bifurcation \cite{tyson}.

Mapping these changes in the number and stability of fixed points across parameter space defines the bifurcation boundaries. For a 2D system, these conditions are straightforward \cite{liu} and are summarized in Table \ref{tab:bifurcation_criteria}. Note that these stability requirements correspond to the two-dimensional Routh–Hurwitz criterion \cite{DorfBishop2017}. In higher-order systems, this criterion imposes conditions on the coefficients of the characteristic polynomial to detect the emergence of purely imaginary roots \cite{hairer}, providing a generalized tool for mapping stability boundaries and the onset of limit cycles.

\begin{table}[t!]
\centering
\renewcommand{\arraystretch}{1.5}
\begin{tabular}{|c|c|c|}
\hline
\textbf{Bifurcation Type} & \textbf{Eigenvalue Condition} ($\lambda$) & \textbf{Jacobian Criteria (2D)} \\ 
\hline
\textbf{Saddle-Node (SN)} & 
\parbox{5cm}{\vspace*{0.2cm}One simple real zero eigenvalue. \\ $\lambda_1 = 0, \lambda_2 \neq 0$\vspace*{0.2cm}} & 
\parbox{5cm}{\vspace*{0.2cm}$\det(J) = 0$ \\ $\text{Tr}(J) \neq 0$\vspace*{0.2cm}} \\ 
\hline
\textbf{Hopf} & 
\parbox{5.5cm}{\vspace*{0.2cm}A pair of purely imaginary eigenvalues. \\ $\lambda_{1,2} = \pm i\omega \quad (\omega \neq 0)$ \\ \textit{i.e.,} $\text{Re}(\lambda) = 0, \, \text{Im}(\lambda) \neq 0$\vspace*{0.2cm}} & 
\parbox{6cm}{\vspace*{0.2cm}$\text{Tr}(J) = 0, \, \det(J) > 0$ \\ \textit{Transversality:} $\frac{d(\text{Tr}(J))}{d \alpha_i} \neq 0$ \\ $\alpha_i$'s: bifurcation parameters \\ such as influx rate, $v$\vspace*{0.2cm}} \\ 
\hline
\end{tabular}
\caption{Detection criteria for saddle-node and Hopf bifurcations based on the Jacobian matrix $J$, Eq.\eqref{eq:jacob_bn}. The transversality condition, $\frac{d(\text{Tr}(J))}{d \alpha_i} \neq 0 \implies \frac{d\,\text{Re}(\lambda)}{d\alpha_i} \neq 0$, ensures $\text{Re}(\lambda)$ genuinely crosses zero, guaranteeing the birth of an isolated closed orbit (limit cycle) with definite local stability (attracting or repelling).}
\label{tab:bifurcation_criteria}
\end{table}

As mentioned before, global bifurcations such as SNIC cannot be strictly predicted using only local stability analysis. However, for the $2$-cisternae system considered here, the nullcline structure and bifurcation diagrams (Fig.\,\ref{fig:2cist_roots}) reveal that the stability of the limit cycle generated via SNIC bifurcation is topologically linked to the fixed point it encloses. Specifically, a stable limit cycle arising via a SNIC bifurcation encloses an unstable fixed point, as shown in Fig.\,\ref{fig:2cist_roots}(g,h). Conversely, a central stable fixed point can generate a concentric bistable solution, characterized by an inner unstable limit cycle and an outer stable limit cycle (Fig.\,\ref{fig:2cist_roots}(e,f)). Consequently, the local instability of the central fixed point serves as a necessary, though not sufficient, precursor for the onset of stable oscillations via this global mechanism, with sufficiency ultimately dictated by the global nonlinear flow.

Fig.\ref{fig:overlap}(a,b,c), depicts the boundaries of saddle-node (shown as grey, dashed lines) and Hopf bifurcation (shown as blue, dashed lines) in the parameter space. However, the eigenvalue analysis does not predict the \textit{basin of attraction} \cite{strogatz} of a fixed point or a limit cycle, and  the complete numerical phase diagram has to be computed numerically.  We use the following numerical scheme for finding these phases:
\begin{enumerate}
\item Solve the differential equations Eqs.\,\eqref{eq:retrograde1_si},\eqref{eq:retrograde2_si} (Eqs.\,\eqref{eq:retrograde1},\eqref{eq:retrograde2} in the main text) for $(M_1(t),M_2(t))$  for given values of parameters and initial conditions.
\item Choose, say, five time points $\tau_a < \tau_b < \tau_c < \tau_d< \tau_e$. These time points are chosen to be large enough numbers to make sure that the system has arrived at a steady-state. We compute the following quantities : $\lvert M_1(\tau_a)-M_1(\tau_c) \rvert, \, \lvert M_1(\tau_b)-M_1(\tau_d) \rvert , \,  \lvert M_2(\tau_a)-M_2(\tau_c) \rvert, \, \lvert M_2(\tau_c)-M_2(\tau_d) \rvert ,\, M_1(\tau_e), M_2(\tau_e)$, $\lvert \rvert$ denotes the absolute value.
\item The steady-state is a fixed point solution if $0<\lvert M_1(\tau_a)-M_1(\tau_c) \rvert<\delta_1$, $0<\lvert M_1(\tau_b)-M_1(\tau_d) \rvert<\delta_1$, $0<\lvert M_2(\tau_a)-M_2(\tau_c) \rvert<\delta_1$, $0<\lvert M_2(\tau_b)-M_2(\tau_d) \rvert<\delta_1$, $M_1(\tau_e) < M^{high}_{th}$ and $M_2(\tau_e) < M^{high}_{th}$, $\delta_1>0$ is a chosen tolerance and $M^{high}_{th} = 300$ is the threshold size assumed for the stable state.



\item The steady-state is a limit cycle solution if $\delta_1\leq\lvert M_1(\tau_a)-M_1(\tau_c) \rvert<M^{high}_{th}$, $\delta_1\leq\lvert M_1(\tau_b)-M_1(\tau_d) \rvert<M^{high}_{th}$, $\delta_1\leq\lvert M_2(\tau_a)-M_2(\tau_c) \rvert< M^{high}_{th}$, $\delta_1\leq\lvert M_2(\tau_b)-M_2(\tau_d) \rvert < M^{high}_{th}$,  $M_1(\tau_e) < M^{high}_{th}$ and $M_2(\tau_e) < M^{high}_{th}$.

\item The system is unbounded if $M_1(\tau_e) \geq M^{high}_{th}$ or $M_2(\tau_e) \geq M^{high}_{th}$. 


\item While the numerical scheme presented above is adequate for the present purposes, it is subject to the following limitations -- (i) it would not detect limit cycle solutions with a radius, $r_\tau<\delta_1$ as well as $r_\tau>M_{th}$ -- and we would assume them to be fixed point and unbounded solutions, respectively. (ii) Fixed point solutions that reach the asymptotic value at a time much larger than $\tau_e$, will be categorized as unbounded solutions. 

\end{enumerate}

\noindent 
\textbf{Computing the phase difference:} \\
Since the limit cycle solutions of the above dynamical system  Eqs.\,\eqref{eq:retrograde1_si},\eqref{eq:retrograde2_si} have nonlinear waveform, and frequency of the system is amplitude dependent, computing the phase difference between the $(M_1(t),M_2(t))$ time-series is not straightforward and the usual methods such as Fourier or Hilbert transform  might give inaccurate results (Bedrosian theorem might not hold \cite{bedrosian1962}). Fortuitously, If the system has well-defined peaks, we can use peak to peak analysis to compute the phase difference using the following procedure --  For each time-series, detect consecutive peaks and assign phase linearly between them $(0,2\pi)$. The phase difference between $(M_1(t),M_2(t))$ can be computed as, $\Delta \phi(t) = 2 \pi$ (time between peaks/time period).

We use a variant of this method -- we compute the time period, $\tau$ from the numerical solution and fit a sinusoidal function $A_1 + A_2 \, sin(2 \pi t /\tau + c)$ to the  solutions obtained for $M_1$ and $M_2$ from the differential equations Eqs.\,\eqref{eq:retrograde1_si},\eqref{eq:retrograde2_si}. The phase difference can then be derived from the correlation of the fitted functions. Note that the sinusoidal fitting extracts phase from the entire waveform's distribution and is a better representative of the true temporal relationship between broad or asymmetric signals than peak to peak's single-point estimate. It essentially captures the ``mean behaviour'' or the macroscopic phase shift of the oscillators.
\section{Numerical computation of structural stability of the 2-cisternae dynamical system \label{sec:cist_struct}}
In this section, we prove the structural stability for $2$-cisternae system Eqs.\,\eqref{eq:retrograde1_si},\eqref{eq:retrograde2_si}. To compute structural stability, we follow the numerical scheme given below. For clarity, we explain these steps when applied to cubic polynomial/cusp singularity (which is the single cisterna system Eqs.\,\eqref{eq:sveq_H1}). Readers interested in further details can consult \cite{thom,golubitzky,Cox2015}.
\begin{enumerate}
    \item Compute the Gröbner basis from the system of equations/vector fields (e.g., $x^3+ b\,x +c$ for  cusp singularity \cite{golubitzky}, where x is the state variable and b and c are constant parameters). We have used \textit{Mathematica} for this (which uses Buchberger algorithm and DegreeReverseLexicographic order, see \cite{Cox2015} for details).
    \item Extract the leading monomials from each polynomial in the Gröbner basis (e.g., $\{x^3\}$). 
    \item Generate all monomials up to a chosen maximum polynomial Degree: $\{1, x, x^2, x^3, x^4, x^5, \ldots\}$.
    \item Reject all monomials divisible by any leading monomial (e.g., reject $x^3, x^4, x^5, \ldots$).
    \item The surviving monomials form the standard monomial basis (e.g., $\{1, x, x^2\}$).
    \item Algebraic capacity = dimension of standard monomial basis (e.g., capacity = 3).
    \item Choose specific numerical values for the parameters and numerically solve the system to count the number of distinct complex roots (geometric roots).
    \item Compare: if geometric roots = algebraic capacity, the system is structurally stable.
    \item If geometric roots $<$ algebraic capacity, the system is structurally unstable.
\end{enumerate}
The above steps for the $2$-cisternae system are tabulated in Table \ref{tab:algebraic_roots}. The number of distinct complex roots and the algebraic capacity are plotted in Fig.\,\ref{fig:robust_algebraic}.

\begin{table}[htpb]
\centering
\renewcommand{\arraystretch}{1.5}
\begin{tabular}{|c|c|}
\hline
\parbox{4.5cm}{\centering \vspace*{0.2cm} \textbf{Vector field} \vspace*{0.2cm}} & 
\parbox{13cm}{\vspace*{0.2cm} 
$\displaystyle \left\{ v \left(a_1 + \frac{ M_1^2}{C_{11} + M_1^2}\right) -\frac{ d_{1} \, M_1}{C_{12} + M_1} - \frac{d_{12} \, M_1}{C_{12} + M_1}\left(a_2 + \frac{ M_2^2}{C_{21}+M_2^2}\right) + \frac{d_{21}\,M_2}{C_{22} + M_2}\left(a_1 + \frac{ M_1^2}{C_{11} + M_1^2}\right) \right.,$ \\[0.4cm] 
$\displaystyle \left. \frac{d_{12} \,M_1}{C_{12} + M_1}\left(a_2 + \frac{M_2^2}{C_{21}+M_2^2}\right) -\frac{ d_{21} \, M_2}{C_{22} + M_2}\left(a_1 + \frac{ M_1^2}{C_{11} + M_1^2}\right) - \frac{ d_{2} \, M_2}{C_{22} + M_2} \right\}$ \vspace*{0.2cm}} \\ 
\hline

\parbox{4.5cm}{\centering \vspace*{0.2cm} \textbf{Leading monomials} \vspace*{0.2cm}} & 
\parbox{13cm}{\centering \vspace*{0.2cm} $\displaystyle \left\{(M_1)^3 (M_2)^2,\,(M_1)^4 (M_2),\,(M_1)^5,(M_2)^6,\,(M_1) (M_2)^5,\,(M_1)^2 (M_2)^4\right\}$ \vspace*{0.2cm}} \\ 
\hline

\parbox{4.5cm}{\centering \vspace*{0.2cm} \textbf{Standard monomial basis} \vspace*{0.2cm}} & 
\parbox{13cm}{\vspace*{0.2cm} 
$\displaystyle \left\{1,\,(M_1),\,(M_2),\,(M_1)^2,2 \,(M_1) \,(M_2),\,(M_2)^2,\,(M_1)^3,3 \,(M_1)^2 \,(M_2) \right.$, \\[0.3cm] 
$\displaystyle 3 \, (M_1) \, (M_2)^2,(M_2)^3, (M_1)^4,4 \,(M_1)^3 \,(M_2),6\, (M_1)^2 \,(M_2)^2,4 \,\,(M_1)\, (M_2)^3,(M_2)^4$, \\[0.3cm] 
$\displaystyle \left. 10\,(M_1)^2\, (M_2)^3,\,5 \,(M_1) \,(M_2)^4,\,(M_2)^5 \right\} $ \vspace*{0.2cm}} \\ 
\hline

\parbox{4.5cm}{\centering \vspace*{0.2cm} \textbf{Algebraic capacity} \vspace*{0.2cm}} & 
\parbox{13cm}{\centering \vspace*{0.2cm} $18$ \vspace*{0.2cm}} \\ 
\hline

\parbox{4.5cm}{\centering \vspace*{0.2cm} \textbf{Number of distinct complex roots} \vspace*{0.2cm}} & 
\parbox{13cm}{\centering \vspace*{0.2cm} $18$ \vspace*{0.2cm}} \\ 
\hline

\parbox{4.5cm}{\centering \vspace*{0.2cm} \textbf{Structural stability} \vspace*{0.2cm}} & 
\parbox{13cm}{\centering \vspace*{0.2cm} \textbf{True} \vspace*{0.2cm}} \\ 
\hline
\end{tabular}
\caption{{\bf Establishing structural stability of two-cisternae dynamical system.} Given a dynamical system, its fixed points are determined by the components of its vector field. If these equations for fixed points can be written as polynomial equations, one can compute a Gröbner basis \cite{Cox2015} for the system to extract the leading monomials and construct the standard monomial basis. The cardinality of this standard basis yields the algebraic capacity of the system. The fixed points of the dynamical system Eqs.\,\eqref{eq:retrograde1_si},\eqref{eq:retrograde2_si} (Eqs.\,\eqref{eq:retrograde1},\eqref{eq:retrograde2} in the main text) can also be computed numerically. Parameter values, $C_{11}=C_{21}=100,C_{12} = C_{22} =20, v=1, d_{12} = 3$, $a_1=a_2=0.05$, $d_1=d_2=1$, $d_{21}=1$. The system's fixed points are considered structurally stable if the number of distinct complex roots it admits exactly matches the algebraic capacity, ensuring that all roots are simple and non-degenerate. This implies that system is locally structurally stable \cite{kuznetsov,golubitzky} which ensures neighborhood stability, i.e., the fixed point point structure of the system remains invariant under small changes of parameters. The number of distinct roots and the algebraic capacity are plotted in Fig.\,\ref{fig:robust_algebraic}.}
\label{tab:algebraic_roots}
\end{table}

\begin{figure}[t!]
\centering
\includegraphics[scale=.5]{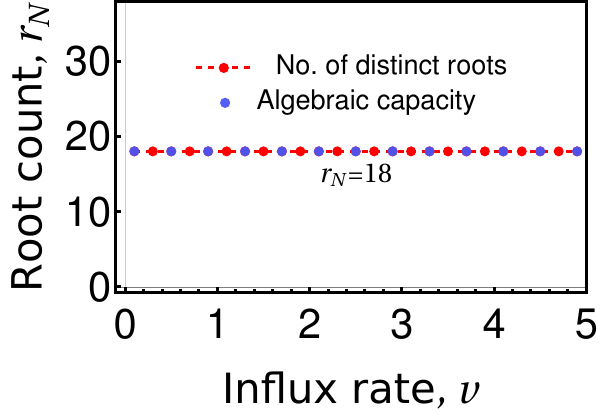}
\caption{\textbf{Structural stability of two-cisternae dynamical system}.
The number of distinct complex roots (red dots) and the algebraic capacity (blue dots) for the dynamical system Eqs.\,\eqref{eq:retrograde1_si},\eqref{eq:retrograde2_si} are plotted and they exactly match over a wide range of influx rate, $v$. This ensures local structural stability (see~\ref{sec:cist_struct} for details), and the change in the system's qualitative behaviour (solution class) is given by bifurcation boundaries (see Fig.\,\ref{fig:overlap}(a,b,c)). Parameter values, $C_{11} = C_{21} = 100$, $C_{12} = C_{22} = 20$, $a_1 = a_2  = 0.05$, $d_{2} = 1$, $d_{12} = 3$, $d_{21} = 1$.}
\label{fig:robust_algebraic}
\end{figure}
\section{Robustness of phases to extrinsic and intrinsic noise  \label{sec:robust_si}}

\subsection{Robustness to extrinsic noise}
As suggested in the above sections, \ref{sec:cist_struct} and Sect.\ref{sec:single_cist}, the dynamical system Eq.\,\eqref{eq:veq_Hd} corresponds to a family of functions with the root structure, robust under small changes in parameters. The structural stability of the system further ensures that the root structure of the dynamical system is robust to intrinsic and extrinsic noise.  
In the following calculations, we subject the dynamical system corresponding to single and multiple cisternae to extrinsic and intrinsic noise and look for bounds on the noise strength for the solutions to remain in the same solution class. To quantify the robustness to extrinsic noise, we assume that each cisterna $i$ with size $M_i$ is driven by an extrinsic multiplicative noise with noise strength $\xi_i(M_i)$, where the multiplicative noise strength at the cisterna $i$ is a function of size of the cisterna $i$. This implicitly assumes that extrinsic fluctuations in the size of a cisterna is only due to local factors, such as local availability of vesicles and fisogens-fusogens or local hydrodynamic interactions \cite{Shapiro,kampen1,kardar}. With this, the dynamical system for a system of cisternae driven by extrinsic multiplicative noise is given by,
\begin{small}
\begin{eqnarray}
\frac{dM_i}{dt} &&=   {\cal K}^{fus} (M_i) - {\cal K}^{fis} (M_i) - {\cal K}^{tr}_{a} (M_i) +{\cal K}^{tr}_{r} (M_i) + \sqrt{\xi_i(M_i)} \eta_i(t) \nonumber \\ &&= f_i(\boldsymbol{M})  + \sqrt{\xi_i(M_i)} \eta_i(t)\,,
\label{eq:robust_langevin_si}
\end{eqnarray}
\end{small}
where $ {\cal K}^{tr}_{a} (M_i) , {\cal K}^{tr}_{r} (M_i)$ are anterograde and retrograde fluxes and together with $ {\cal K}^{fus},{\cal K}^{fis}$ constitute the deterministic term $f_i(\boldsymbol{M})$ for the dynamics of $i^{th}$ cisterna, $\boldsymbol{M} = (M_1,M_2,\cdots)$, $\eta_i(t)$ is Gaussian white noise and, as mentioned above, the noise strength $\xi_i(M_i)$ at the cisternae $i$  is a positive and smooth function of $M_i$  solely dependent on the size of the cisterna $i$. For Eq.\,\eqref{eq:robust_langevin_si}, the Fokker-Planck equations reads (It\^{o} convention \cite{kampen})
\begin{small}
\begin{eqnarray}
\partial_t P(\boldsymbol{M}) = \partial_{M_i} \left( {f_i} \, P( \boldsymbol{M}) \right) + \frac{1}{2} \partial_{M_i} \partial_{M_j} \left({B}_{ij} \, P(\boldsymbol{M})\right)\,,
\label{eq:Fokker_P}
\end{eqnarray}
\end{small}
where ${B}_{ij}$ is diagonal matrix with entries $\xi_i$ and repeated indices are summed over. 

Now, consider Eq.\,\eqref{eq:sveq_H1} (or Eq.\,\eqref{eq:veq_Hd} in the main text) for single cisterna with multiplicative extrinsic noise,
\begin{small}
\begin{eqnarray}
\dot{M} &&=   v \, \left(a_0 +  \frac{M^2}{C_1 + M^2} \right)- \left(\frac{M}{C_2 + M}\right) + \sqrt{\xi(M)} \, \eta(t) \,,
\label{eq:veq_noise_si}
\end{eqnarray}
\end{small}
$\eta$ is Gaussian white noise and $\xi(M) = \epsilon \, M$, $\epsilon$ a constant, i.e., stochastic addition and removal of vesicles is proportional to the cisterna size.

\noindent \textit{Steady state distribution from Fokker-Planck system}: 
Consider the above equation Eq.\,\eqref{eq:veq_noise_si} for single cisterna with added extrinsic multiplicative noise.
With It\^{o} convention, above equation can be converted into a Fokker-Planck system \cite{kampen},
\begin{small}
\begin{equation}
\frac{\partial P(M)}{\partial t} =   -\frac{\partial \left(f(M) \, P(M)\right)}{\partial M} + \frac{1}{2} \frac{\partial^2 \left(\xi(M) P(M)\right)}{\partial M^2} \,,
\label{eq:veq_noise_fp}
\end{equation}
\end{small}
where $f(M)$ is the deterministic part in Eq.\,\eqref{eq:veq_noise_si}. It corresponds to steady state distribution,
\begin{small}
\begin{equation}
P^*(M) = \mathcal{N} \, \exp \left[2 \, \int_0^M \frac{\left(f(s) - \frac{1}{2} \xi'(s)\right)}{\xi(s)}\, ds \right] \implies \frac{d \log(P^*(M))}{dM} = \frac{\left(f(s) - \frac{1}{2} \xi'(s)\right)}{\xi(s)}\,,
\label{eq:catastrophe_P}
\end{equation}
\end{small}
where $\mathcal{N} $ is the normalization factor.
\begin{figure}[t!]
\centering
\includegraphics[width=\textwidth]{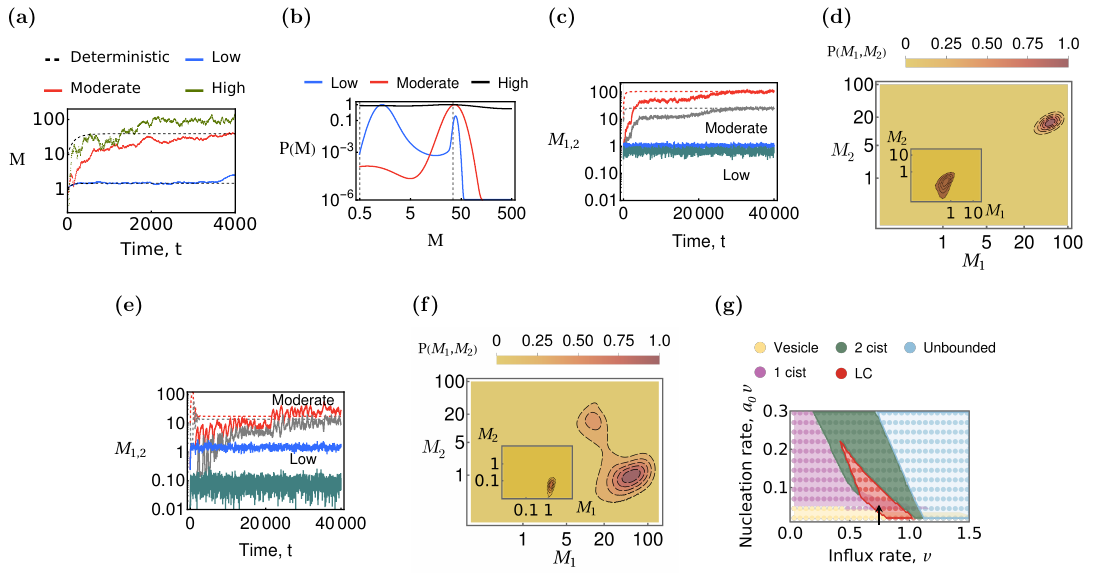}
\caption{{\bf  Robustness of stable cisternal phases under extrinsic noise}. (a) Stochastic trajectories for single cisterna system in the bistable regime with initial condition $M=0$, under extrinsic noise with strength $\epsilon M$, $\epsilon$ a constant (see Eq.\,\eqref{eq:robust_langevin_si}). The two stable fixed points corresponding to the deterministic system are shown as black dotted line. 
Individual trajectories corresponding to low noise ($\epsilon < v \, a_0$, blue), moderate noise ($v \, a_0 < \epsilon < v \, (a_0+1)$, red) and high noise ($\epsilon > v \, (a_0+1)$, green) indicating the robustness of vesicle, cisterna and the unbounded regimes, respectively.  Parameter values, $v = 0.5$, $C_{1}= 100$, $C_{2} = 40$, $a_{0}=0.05$.
(b) Steady state probability distribution of the Fokker-Planck equation Eq.\,\eqref{eq:Fokker_P} for corresponding trajectories in (a) showing bimodal (blue), unimodal (red) and unregulated system (black). Steady state cisterna sizes for the bistable deterministic system are shown as black dotted line.
(c) Stochastic trajectories of ($M_1,M_2$) for a multistable $2$-cisternae reduced (real) system (see proposition \ref{prop:prop}).
($M_1,M_2$) transitions from vesicle to cisterna phase as the noise strength is increased (see Eq.\,\eqref{eq:robust_langevin2}). Parameter values, $v=1.45$, $C_{11} = C_{21} = 100$, $C_{22} = 20$, $C_{12} = 20$, $d_{12} = 1$, $d_{21}=0.3$, $d_{1}=1$, $d_{2}=1.5$, 
$a_{1}=a_{2}=0.05$.
(d) Contour plot showing  steady state probability distribution $P(M_1,M_2)$  of the Fokker-Planck equation corresponding to (c) showing probability distribution for cisterna phase  for moderate noise and vesicle phase in low noise limit (inset). 
(e) Stochastic trajectories of ($M_1,M_2$) for a multistable $2$-cisternae system (with $J_{12}$, $J_{21}$ functions of both $M_1$ and $M_2$, see proposition \ref{prop:prop}) with initial condition ($M_1=0,M_2=0$), showing transition from vesicle phase to the limit cycle phase as the noise strength is increased. Parameter values, $v=1$, $C_{11} = C_{21} = 100$, $C_{22} = 20$, $C_{12} = 20$, $d_{12} = 1.67$, $d_{21}=0.5$, $d_{2}=0.83$, $a_{1}=a_{2}=0.05$.
(f) Contour plot showing steady state probability distribution $P(M_1,M_2)$ for Fokker-Planck system corresponding to (e).  The steady state probability distribution shows limit cycle phase for moderate noise (inferred from two maxima in $P(M_1,M_2)$) and vesicle phase in low noise limit (inset).
(g) Phase diagram in ($a_0 v,v$) plane, where $v$ is the influx rate and $a_0$ is the nucleation constant depicting vesicle phase (yellow), stable single cistena phase (violet), stable $2$-cisternae phase (green), limit cycle phase (red) and unbounded growth (blue). varying nucleation rate can change the solution class from vesicle to limit cycle phase, shown as black arrow. This explains the response to extrinsic noise observed in (e,f), as adding a multiplicative noise with strength, $\sqrt{\epsilon_i M_i}$, $\epsilon_i$ a constant, effectively tweaks the the nucleation rate, see Eq.\,\eqref{eq:shape1_si}.
Parameter values, $C_{11} = C_{21} = 100$, $C_{22} = 20$, $C_{12} = 20$, $d_{12} = 2.5$, $d_{21}=0.85$, $d_{2}=0.4$, $a_{1}=a_{2}=0.05$. Here, the retrograde flux, parametrized by $d_{21}$ is taken to be $M_1$-independent.}
\label{fig:robust_crit}
\end{figure}
Consequently, while the fixed points of the deterministic system are given by the roots of $f(M)$, the extrema of the distribution $P(M)$ are given by the roots of  $\left(f(M) - \frac{1}{2} \xi'(M)\right)$, defined as the shape function in \cite{cobb}. Hence, the extrema of the probability distribution corresponding to the stochastic system Eq.\,\eqref{eq:veq_noise_si} is given by,
\begin{eqnarray}
(-1 + v + a_0 \, v - \epsilon) M^3 +  (v + a_0 \,v - \epsilon) C_2 \, M^2 + C_1 (-1 + a_0 \, v - \epsilon) M + C_1 \, C_2 (a_0 \, v - \epsilon) = 0\,.
\label{eq:shape1_si}
\end{eqnarray}

It is clear that adding multiplicative noise of this functional form effectively leads to a modified version of Eq.\,\eqref{eq:cubicpoly} with nucleation rate $v \, a_0 \rightarrow (v \, a_0-\epsilon)$ and for $\epsilon > a_0 \, v$, sign signature of the sequence of coefficients for the above polynomial Eq.\,\eqref{eq:shape1_si} changes, leading to the change in the number of fixed points. For instance, consider the bistable regime for the deterministic system corresponding to Eq.\,\eqref{eq:veq_noise_si}. In this regime the sign signature for the sequence of coefficients is ($-,+,-,+$), see Table \ref{Tab:Tcr}. As the noise strength is increased,  $\epsilon > a_0 v $, the sign signature changes to ($-,+,-,-$) and the fixed points near zero are destroyed (Fig.\,\ref{fig:robust_crit}(a)) and the system has just one positive fixed point left.  On further increasing the noise strength, $\epsilon > (a_0 +1)v$, sign signature changes to ($-,-,-,-$) , the fixed point away from zero is destroyed as well and the system has no remaining fixed points.
The same can be seen from the corresponding Fokker-Planck equation. As we see in Fig.\,\ref{fig:robust_crit}(b), steady state distribution of the Fokker-Planck equation for one cisterna is  bimodal for low noise ($\epsilon << a_0 \; v$), it becomes unimodal for moderate noise ($v \, a_0 < \epsilon < v \, (a_0+1)$), and the system is without any maxima for $\epsilon > (a_0 +1)v$.  

The extension of Eq.\,\eqref{eq:catastrophe_P} to two cisternae, Eqs.\,\eqref{eq:retrograde1_si},\eqref{eq:retrograde2_si}  (Eqs.\,\eqref{eq:retrograde1},\eqref{eq:retrograde2} in the main text) is not always possible as it requires integrability (potential) conditions~\cite{kampen,gardiner,Kloeden1992}, not satisfied by the dynamical system  Eqs.\,\eqref{eq:retrograde1_si},\eqref{eq:retrograde2_si}. With the potential condition (zero nonequilibrium current \cite{gardiner}), robustness bounds for the root structure of the two cisternae case are similar to  robustness bounds for the one cisterna case. 
As before, if we assume that extrinsic fluctuations in the size of a cisterna is only due to local factors, addition of cisternal size dependent extrinsic multiplicative noise for two cisternae case gives,
\begin{small}
\begin{eqnarray}
\dot{M}_1 = f_1(M_1,M_2) + \sqrt{\xi_1(M_1)} \eta_1   \hspace{1.5cm} \dot{M}_2 = f_2(M_1,M_2) + \sqrt{\xi_2(M_2)} \eta_2\, \,,
\label{eq:robust_langevin2}
\end{eqnarray}
\end{small}

where $f_1,f_2$  are given by deterministic dynamics in Eqs.\,\eqref{eq:retrograde1_si}, \eqref{eq:retrograde2_si} and $\eta_1,\eta_2$ are Gaussian white noises. Furthermore, we analyse the change in root structure for $\xi_1(M_1) = \epsilon_1 \, M_1, \, \xi_2(M_2) = \epsilon_2 \, M_2$, i.e., stochastic addition and removal of vesicles is proportional to the cisterna size, as for single cisterna case. For $2$-cisternae system, if one takes $J_{12} = d_{12} M_1/(C_{12}+M_1)$ and  $J_{21}=d_{21} M_2/(C_{22}+M_2)$ (see Eqs.\,\eqref{eq:retrograde1_si}, \eqref{eq:retrograde2_si}), i.e., if $J_{12}$ is a function of $M_1$ and $J_{21}$ is function of $M_2$ alone, then from proposition \ref{prop:prop}, the system only admits real eigenvalues and has only fixed point solutions. For such a system, ($M_1,M_2$) can go from vesicle phase to cisterna phase as noise strength is increased as can be seen in the stochastic trajectory, Fig.\,\ref{fig:robust_crit}(c) as well as with the steady state distribution of the corresponding Fokker-Planck system, Fig.\ref{fig:robust_crit}(d). Following steps similar to Eq.\,\eqref{eq:catastrophe_P}, it gives that the system remains in vesicle phase for small noise ($\epsilon_1/2 + \epsilon_2/2 < v\, a_1$, see Eqs.\,\eqref{eq:retrograde1_si},\eqref{eq:retrograde2_si},\eqref{eq:robust_langevin2}) but goes to the cisternae phase as the noise strength is increased, as the roots (fixed points) near zero are destroyed. 

Furthermore, for the $2$-cisternae system, if $J_{12}$, $J_{21}$ are functions of both $M_1$ and $M_2$, the system can admit oscillatory solutions, thereby enabling a noise-induced transition  from the vesicle/cisternae phase to the limit cycle (CP) phase (Fig.\ref{fig:robust_crit}(e,f)).
This happens as the multiplicative noise with the strength, $\epsilon_i M_i$ modifies the nucleation rate as mentioned before, and therefore changes the effective bifurcation parameters for the system to enter the oscillatory regime, see the phase diagram, Fig.\,\ref{fig:robust_crit}(g) in the $(v,a_0\,v)$ plane. Crucially, multiplicative noise can cause a noisy system to exhibit oscillations by changing the effective size dependence of the fluxes and consequently, its admissible solution class.

The above analysis can in principle be extended to a general class of noises which are smooth positive functions of cisternal size, but this is beyond the scope of this work. Furthermore, in the bistable parameter regime, system can be in one of the states with escape rate given by Kramers' barrier \cite{gardiner}. his can be used to devise a safety mechanism that yields a robust control mechanism to mitigate the spurt of material flux from ER, which is a subject of future work.
\subsection{Robustness to intrinsic noise}
\label{fig:robust_noise_intrinsic}
We check for robustness of phases under intrinsic noise via  Stochastic simulations (Gillespie algorithm \cite{gillespieE})  using macroscopic kernels -- ${\cal K}^{fus},{\cal K}^{fis}$, and intercisternal rates as propensities. For multiple cisterna, different propensities are given by Eqs.\,\eqref{eq:fusi},\eqref{eq:fisi},\eqref{eq:trai},\eqref{eq:trri}, 
\begin{eqnarray*}
 {\cal K}^{fus} (M_i) &&= \left(a_{i}+\frac{M_i^2}{C_{i1} + M_i^2}\right) \hspace{3.6cm} {\cal K}^{fis} (M_i) = \frac{d_{i} \, M_i}{C_{i2} + M_i} \\ {\cal K}^{tr}_{a} (M_i) &&= \frac{d_{i2} \, M_i}{C_{i2} + M_i}\left(a_{i+1}+\frac{M_{i+1}^2}{C_{(i+1)1} + M_{i+1}^2}\right) \hspace{1.5cm}  {\cal K}^{tr}_{r} (M_i) = \frac{d_{(i+1)2} \, M_{i+1}}{C_{(i+1)2} + M_{i+1}}\left(a_{i}+\frac{M_i^2}{C_{i1} + M_i^2}\right).
\end{eqnarray*}
The Gillespie algorithm is used to generate the stochastic trajectories with the above propensities from which the asymptotic mean size and variation about the mean for one or multiple cisternae system can be computed for a given range of parameters. As in the the deterministic case, the mean value can be used to generate a phase diagram that accounts for intrinsic noise. 

For single cisterna, the modified phase diagrams (Fig.\,\ref{fig:robust_crit_in}(a)) in the ($v,C_2$) plane, where $v$ is influx rate  and $C_2$ is Hill saturation constant for fission, are similar to those of the deterministic case (Fig.\,\ref{fig:p_roots1}(b)), but differs in some regions, especially near the vesicle-cisterna phase boundary. However, when one computes the coefficient of variation (CV, i.e., the ratio of the standard deviation to the mean), one observes these anomalous regions have, CV $>1$ (Fig.\,\ref{fig:robust_crit_in}(b)), suggesting the behaviour is dominated by copy number fluctuations in these regions and is not captured by mean value analysis. A similar trend is observed for the phase diagram with intrinsic noise for the $2$-cisternae case (Fig.\ref{fig:robust_crit_in}(c,d)), compared to the deterministic phase diagram Fig.\ref{fig:overlap}(a).
\begin{figure}[t!]
\centering
\includegraphics[width=\textwidth]{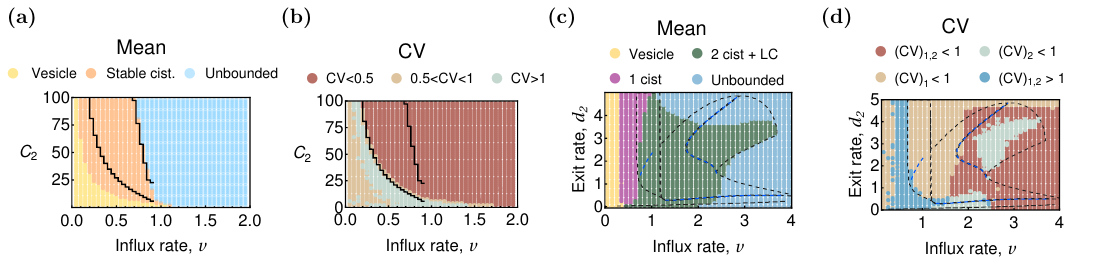}
\caption{{\bf  Robustness of stable cisternal phases under intrinsic noise.}
 (a) Phase diagram for mean cisternal size under intrinsic noise for single cisterna system in ($v,C_2$) plane where $v$ is the influx rate  and $C_2$ is the Hill saturation constant for fission depicting vesicle phase (yellow), stable cisterna phase (orange) and the unbounded growth (blue).
 The vesicle-cisterna and the cisterna-unbounded growth boundaries for the deterministic system are shown as black solid lines, for comparison.
 (b) The coefficient of variation ({\bf CV}, i.e., the ratio of the standard deviation to the mean) in ($v,C_2$) plane is high at the vesicle-cisterna boundary suggesting the system is dominated by copy number fluctuations. (a,b) Parameter values, $C_{1}= 100$, $a_{0}=0.05$ (see Eq.\,\eqref{eq:veq_Hd}).
 (c,d) Phase diagram for mean cisternal size and coefficient of variation for $2$-cisternae system. Phases for two cisterna dynamics under intrinsic noise also suffers from copy number fluctuations near phase/bifurcation boundaries. Saddle-node bifurcation boundaries are shown in grey and Hopf bifurcation boundaries are shown in blue, see Fig.\,\ref{fig:overlap}(a) for annotated bifurcation boundaries. (c,d) Parameter values, $C_{11} = C_{21} = 100$, $C_{22} = 20$, $C_{12} = 20$, $d_{12} = 3$, $d_{21}=0$, $a_{1}=a_{2}=0.05$ (see Eqs.\,\eqref{eq:retrograde1},\eqref{eq:retrograde2}).}
\label{fig:robust_crit_in}
\end{figure}

\section{De novo cisternal assembly time \label{sec:denovo_tau}}
\subsection{De novo assembly of single cisterna \label{sec:1cist_tau}}
In the main text (Sect.\,\ref{sec:single_cist}), we have used the functional form of the de novo formation time for the single cisterna. In this section, we sketch a non-perturbative method for estimating this, which involves computation of time taken by a general nonlinear system to escape a saddle-node ghost and reach the asymptotic solution. Let us consider Eq.\,\eqref{eq:sveq_H1} (Eq.\,\eqref{eq:veq_Hd} in the main text) again,
\begin{small}
\begin{equation}
\frac{dM}{dt} =  v \, \left(a_0 +  \frac{M^2}{C_1 + M^2} \right)- \frac{d \, M}{C_2 + M}\,.
\label{eq:main}
\end{equation}
\end{small}
By integrating the above equation, we can compute the time taken by the above system to reach a desired steady-state cisterna size $M^*$,
\begin{small}
\begin{equation}
\int_0^{\tau_{denovo}} dt=  \int_0^{M^*} dM \, \left[ v \, \left(a_0 +  \frac{M^2}{C_1 + M^2} \right)-\frac{d \, M}{C_2 + M} \right]^{-1}\,,
\label{eq:struct}
\end{equation}
\end{small}
which in principle should give the functional dependence of the cisterna de novo formation time, $\tau_{denovo}$ on the system parameters. However, above integral diverges logarithmically at the roots of Eq.\,\eqref{eq:main}. One way around this is to work in a regime where the system can have at most one fixed point $M^*$ (i.e., cubic discriminant $\Delta<0$), and then analyse $\tau_{denovo}$ as we take different limits, such as  $\Delta \rightarrow 0$. For the above integral Eq.\,\eqref{eq:struct}, this involves tedious expressions for roots of a cubic equation \cite{katz}. To tackle this, we proceed with computation of the cisterna formation time in the following simpler models:\\
\textbf{Saturated fission kernel:} We take a piecewise continuous approximation for fusion kernel and a linear approximation for fission kernel, i.e.,
\begin{small}
\begin{eqnarray}
\label{eq:Rfussimp}
{\cal K}^{fus} = v(a_0 + \Theta(M-C_a)) \hspace{0.8cm}
{\cal K}^{fis} = M/C_b \,,
\label{eq:Rfissimp}
\end{eqnarray}
\end{small}
$\Theta$ being the Heaviside function. It has two stable steady states given by $M^* = \{a_0\, C_b \,v, (1 + a_0) \, C_b \,v$\}. We compute the integration time up to $q\, M^*\,( 0<q<1)$, $M^*$ being the larger steady state.
\begin{small}
\begin{eqnarray}
{\tau_{denovo}} &&=  \int_0^{C_a} dM \, \left[ v(a_0 + \Theta(M-C_a) - M/C_b\right]^{-1} +  \int_{C_a}^{q M^*} dM \, \left[ v(a_0 + \Theta(M-C_a) - M/C_b \right]^{-1} \nonumber  \\ 
&&= \underbrace{- \,C_b \log(a_0 C_b v-C_a) \,+\, C_b \log (a_0 C_b v)}_{\textrm{Saddle-node } 1} \, \underbrace{ +\, C_b \log \left(C_b v (1+a_0) - C_a\right) }_{\textrm{Saddle-node } 2} \, - \, C_b \log \left(C_b v (1+a_0) (1- q)\right) \,.
\label{eq:tdenovo}
\end{eqnarray}
\end{small}
In the above expression, the arguments of the two clubbed logarithmic expressions give the condition for saddle-node bifurcation, i.e, $C_a/(a_0 C_b v) = 1$ and  $C_a/((a_0+1) C_b v) = 1$. From Eq.\,\eqref{eq:tdenovo}, expression for saddle-node $1$ (Fig.\,\ref{fig:SNIC_1cist} (a)), can be rewritten as, 
\begin{small}
\begin{equation}
{\tau_{denovo}} \bigg|_0= - C_b \log \left(1-\frac{C_a}{a_0 C_b v}\right) = - C_b \log (1-x) = - \frac{C_a}{a_0 v} \frac{1}{x} \log(1-x) \sim \frac{C_a}{a_0 v} (1-x)^{-1/2}\,,
\label{eq:tdenovo_0}
\end{equation}
\end{small}
where $x= \frac{C_a}{a_0 C_b v} = \frac{v_{SN}}{v} $, is the bifurcation parameter and we have defined a critical influx rate $v_{SN}$ at which saddle-node bifurcation occurs. The above expression implies that as $v$ nears the critical bifurcation parameter $v_{SN}$, $\tau_{denovo}$ increases. Hence, the effect of the saddle-node bifurcation persists even after the fixed points have disappeared. As mentioned in the Glossary\,\ref{box:definitions}, this phenomenon, known as the saddle-node ghost, causes a critical slowing down where the system `remembers' the vanished fixed points, resulting in a time delay, that as seen in Eq.\,\eqref{eq:tdenovo_0}, scales as the inverse square root of the distance from the bifurcation \cite{strogatz}.
Furthermore, rest of the expression Eq.\,\eqref{eq:tdenovo_0}, including the expression for saddle-node $2$ (Fig.\ref{fig:SNIC_1cist} (a)), can be rewritten as,
\begin{small}
\begin{eqnarray}
\tau_{denovo}\bigg|_{M^{*}} &&= -C_b \log \left(C_b v (1+a_0) (1- 1/e)\right)+ C_b \log \left(C_b v (1+a_0) - C_a\right)
 \nonumber \\ &&= - C_b \log \left(1-q\right) -C_b \log \left(1-\frac{C_a}{(a_0+1) C_b v}\right) \,.
\label{eq:tdenovo_M}
\end{eqnarray}
\end{small}
Summing Eqs.\,\eqref{eq:tdenovo_0} and \eqref{eq:tdenovo_M}  yields the de novo cisterna formation time,
\begin{small}
\begin{equation}
\tau_{denovo} = \underbrace{-C_b \log \left(1-q\right) }_{\tau_f} \underbrace{- C_b \log \left(1-\frac{C_a}{a_0 C_b v}\right)}_{\tau_{SN}^1}\underbrace{-C_b \log \left(1-\frac{C_a}{(a_0+1) C_b v}\right)}_{\tau_{SN}^2}\,,
\label{eq:tdenovo_full}
\end{equation}
\end{small}
where $\tau_f$ is the characteristic time scale at the asymptotic fixed point $M^*$ and $\tau_{SN}^{(1,2)}$ are saddle-node delays. Note that, in lowest order of $C_b$, cisterna formation time scale $\tau_{denovo}\sim C_b \log(1-q)$, with higher order contributions from saddle-node ghosts delays. 
The formula Eq.\,\eqref{eq:tdenovo_full} can be extended to kernels admitting more number of saddle-node bifurcations, 
\begin{small}
\begin{equation}
\tau_{denovo} \approx -C_b \log \left(1-q\right) - C_b \sum_i\log \left(1-x_i\right)\,,
\label{eq:tdenovo_SNs}
\end{equation}
\end{small}
where $x_i$ are all the small parameters corresponding to different \textit{saddle node bifurcations} (e.g. $\frac{C_a}{a_0 C_b v}$ above). 
As we consider more complex flux kernels, the characteristic time scales are more complicated functions, as we will see below. 
\\
\textbf{Unsaturated fission kernel:} For the next level of approximation, we will compute the cisterna formation time for the unsaturated fission kernel, i.e., Hill-type function for fission kernel and Heaviside function for fusion kernel, 
\begin{small}
\begin{eqnarray}
{\tau_{denovo}}  &=& \int_0^{C_1} dM \, \left[ v\left(a_0 + \Theta(M-C_1)\right) -\frac{d_1 \, M}{C_2 + M}\right]^{-1} +  \int_{C_1}^{q\,M^*} dM \, \left[ v\left(a_0 + \Theta(M-C_1)\right) - \frac{d_1 \, M}{C_2 + M}\right]^{-1} 
\end{eqnarray}
\end{small}
\begin{small}
\begin{eqnarray}
&&= \frac{C_2 \, d_1 \log \left(\frac{a_0 \, C_2 \, v}{a_0 \, v (C_1+C_2)-C_1 d_1}\right)+C_1 \, (a_0\, v-d_1)}{(a_0 \,  v-d_1)^2}  + \frac{-(q\,M^*-C_1)\,(d_1-(a_0+1) v) +C_2 \, d_1 \log \left(\frac{-C_1\,  d_1 + (a_0+1) \,( C_1+ C_2)\, v}{-q\,M^*\,d_1 + (a_0+1) \, (q\,M^*+C_2)\,v} \right)}{(a_0\, v+v-d_1)^2}  \,,\nonumber \\ 
\label{eq:tdenovo_sat}
\end{eqnarray}
\end{small}
where $M^* = \, (1+a_0) \, C_2 \, v/(d_1-a_0-a_0\,v),\, 0<q<1$ and $M^*$ is the steady state cisterna size. In the above expression, $d_1>v+a_0\, v $ ensures the existence of a positive, stable fixed point. These conditions, combined with $a_0(C_1+C_2) \, v - C_1\,d_1>0$ and $q\,M^*>C_1$ guarantee the positivity of the above integral. In the lowest order in $q$, Eq.\,\eqref{eq:tdenovo_sat} can be approximated by
\begin{small}
\begin{eqnarray}
{\tau_{denovo}} &\approx& \underbrace{\frac{(C_1 + C_2 \,q)}{(d_1 - (1 + a_0) v)}}_{\tau_f} + \underbrace{\frac{C_2 d_1 }{(a_0 \,  v-d_1)^2} \log \left(\frac{a_0 \, C_2 \, v}{a_0 \, v (C_1+C_2)-C_1 d_1}\right)}_{\tau_{SN}^1} \nonumber + \underbrace{\frac{C_2 \, d_1}{(a_0\, v+v-d_1)^2}  \log \left(\frac{-C_1\,d_1+(a_0+1) v (C_1+C_2)}{(a_0+1) \, C_2 \, v}\right)}_{\tau_{SN}^2} \,.\nonumber \\
\label{eq:tdenovo_sat_full}
\end{eqnarray}
\end{small}
The arguments of the $\log$ functions above give the saddle-node bifurcation conditions, i.e.,  $a_0(C_1+C_2) \, v - C_1 d_1>0$ and  $(a_0+1)(C_1+C_2) \, v - C_1 d_1>0$.
As before, The characteristic time scales in front of the log terms in Eq.\,\eqref{eq:tdenovo_sat_full} are determined by the inverses of the eigenvalues, evaluated at the poles of the integrands in the first and second integrals above, respectively. We can see that the characteristic time scales have been modified compared to Eq.\,\eqref{eq:tdenovo_full}. This is due to the fission rate now depending on the cisternal size, which leads to the characteristic time scales in Eq.\,\eqref{eq:tdenovo_sat_full}, depending on the influx rate, $v$.

\begin{figure*}[t!]
\centering
\includegraphics[width=\textwidth]{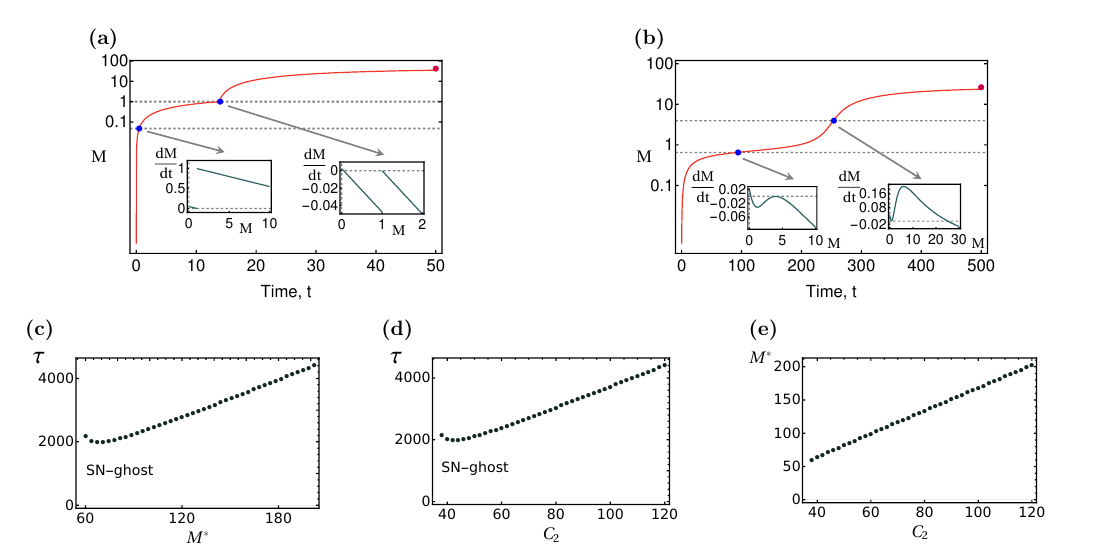}
\caption{{\bf Cisterna formation time for one cisterna.} (a) Cisterna growth over time for Heaviside fusion kernel and saturated fission kernel (see Eq.\,\eqref{eq:Rfissimp}) showing saddle-node bifurcations (blue dots, representing double root points) and the asymptotic cisterna size (red dot). (Inset) The rate of change of cisterna size at the two saddle-node bifurcation points leads to saddle-node delays. Parameter values, $v=1$, $C_a=1$, $C_b=20$, $a_0=0.05$. (b) Cisterna growth over time for Hill-type kernels (see Eq.\,\eqref{eq:veq_Hd}) showing saddle-node bifurcations (blue dots, representing double root points) and the asymptotic cisterna size (red dot). (Inset) rate of change of cisterna size at the two saddle-node bifurcation points. Parameter values, $v=0.7$, $C_1=10$, $C_2=10$, $a_0=0.05$. (c) The de novo cisterna formation time $\tau$ as a function of $M^*$ depicting non-monotonic dependence at lower values of $M^*$ in the proximity of saddle-node ghost (see Eq.\,\eqref{eq:veq_Hd}).
(d) Non-monotonic dependence of cisterna formation time $\tau$ as a function of Hill saturation constant for fission, $C_2$. (e) The steady state cisterna size $M^*$ increases with $C_2$. (c,d,e) Parameter values, $v=0.6$, $C_1=100$, $a_0=0.05$.}
\label{fig:SNIC_1cist}
\end{figure*}
\noindent\textbf{Saddle-node ghost:} Consider now the case where both fusion and fission rates are saturated. This case has the same functional form as the saddle-node bifurcation. The de novo cisterna formation time can be obtained from integration, as before
\begin{small}
\begin{eqnarray}
{\tau_{denovo}}  &=&   \int_0^{A_M^*} dM \, \left[ v\left(a_0 + \frac{M^2}{C_1}\right) -\frac{d_1 \, M}{C_2}\right]^{-1} = -\frac{2 C_1 C_2\left(\tan ^{-1}\left(\frac{2 A_M^* C_2 v-C_1 d_1}{\sqrt{C_1 \left(-C_1 d_1{^2}+4 a_0 C_2^2 v^2\right)}}\right)+\tan^{-1}\left(\frac{C_1 d_1}{\sqrt{C_1 \left(-C_1 d_1{^2}+4 a_0 C_2^2 v^2\right)}}\right)\right)}{\sqrt{C_1 \left(-C_1 d_1^{2}+4 a_0 C_2^{2} v^2\right)}} \,,\nonumber \\
\label{eq:tdenovo_sat_2}
\end{eqnarray}
\end{small}
where $A_M^*$ is a chosen upper limit (terminal state) for integration. Here, as before, the characteristic time in front is given by the inverse of the derivative computed at the pole of the integrand. The condition for saddle-node bifurcation is given by imposing reality condition on the integrand, i.e.,  $ -C_1 d_1{^2}+4 a_0 C_2^2 v^2 >0 $, which is also the discriminant of the polynomial given by the inverse of the integrand. This gives the criteria for small parameter for the saddle-node ghost, $\frac{C_1 d_1{^2}}{4 a_0 C_2^2 v^2} = x < 1$. We can approximate $\tau_{denovo}$ from Eq.\,\eqref{eq:tdenovo_sat_2} by taking $A_M^*$ to be the double root (it is always greater than or equal to the real part of the complex root) and the expression can be written in a simpler form in terms of small quantity, $x$ as
\begin{small}
\begin{eqnarray}
{\tau_{denovo}} = 2 K \frac{x}{\sqrt{1-x^2}} \tan ^{-1} \left( \frac{x}{\sqrt{1-x^2}}\right) \approx K \frac{x}{\sqrt{1-x^2}} \log(1-x) \approx K \frac{x^2}{\sqrt{1-x^2}} \left(\frac{1}{\sqrt{1-x}}\right) \,,
\label{eq:tdenovo_sat_3}
\end{eqnarray}
\end{small}
where $K =C_2/d_1$. This provides several equivalent formulae for computing the saddle-node delay. Moreover, beyond a saddle-node bifurcation point  (discriminant $\Delta < 0 $), indexed by $i$, we can look at the real and imaginary parts of the complex roots (fixed points) $R_i$, i.e., poles of the integrand in Eq.\,\eqref{eq:tdenovo_sat_3}. With straightforward algebraic manipulations, we can see that the small parameter for the $i^{th}$ saddle-node bifurcation,  $x_i= \left(\frac{C_1 d_1{^2}}{4 a_0 C_2^2 v^2}\right)_i$, which is the ratio of the two factors in the discriminant of a quadratic polynomial, can also be computed from the real and imaginary parts of the fixed point, $R_i$
\begin{eqnarray}
x_i = \sqrt{(Re[R_i])^2/\left((Re[R_i])^2 + (Im[R_i])^2\right)}\,.
\label{eq:bif_Ei}
\end{eqnarray}

Taken together, the results from the Eqs.\,\eqref{eq:tdenovo_SNs},\eqref{eq:tdenovo_sat_full},\eqref{eq:bif_Ei}, the de novo cisterna formation time can be computed using the following formula,
\begin{eqnarray}
\tau_{denovo} &\approx &  \tau_{f}(M^*) +  \sum_i  \tau_{SN}^i(E_i,x_i) \nonumber \\ & \hspace{-1.5cm} & i \in \textrm{SN bifurcations encompassed by the trajectory,}
\end{eqnarray}
where $\tau_f$ is the characteristic time scale computed at the steady state $M^*$ (eigenvalue at $M^*$), while the timescale corresponding to $i^{th}$ SN-ghost delay can be computed from the eigenvalue $E_i$ and the small parameter $x_i$. 

In Fig.\,\ref{fig:SNIC_1cist}(b), we have considered the dynamical system for the single cisterna, Eq.\,\eqref{eq:main}, and we have used the above formula to compute the de novo cisterna formation time. It shows the delay caused by the two SN-ghosts $\tau_{SN}$, and the characteristic time $\tau_f$ at the steady state cisternal size.

\subsection{De novo assembly of two cisternae \label{sec:2cist_tau}}

\begin{figure*}[t!]
\centering
\includegraphics[width=\textwidth]{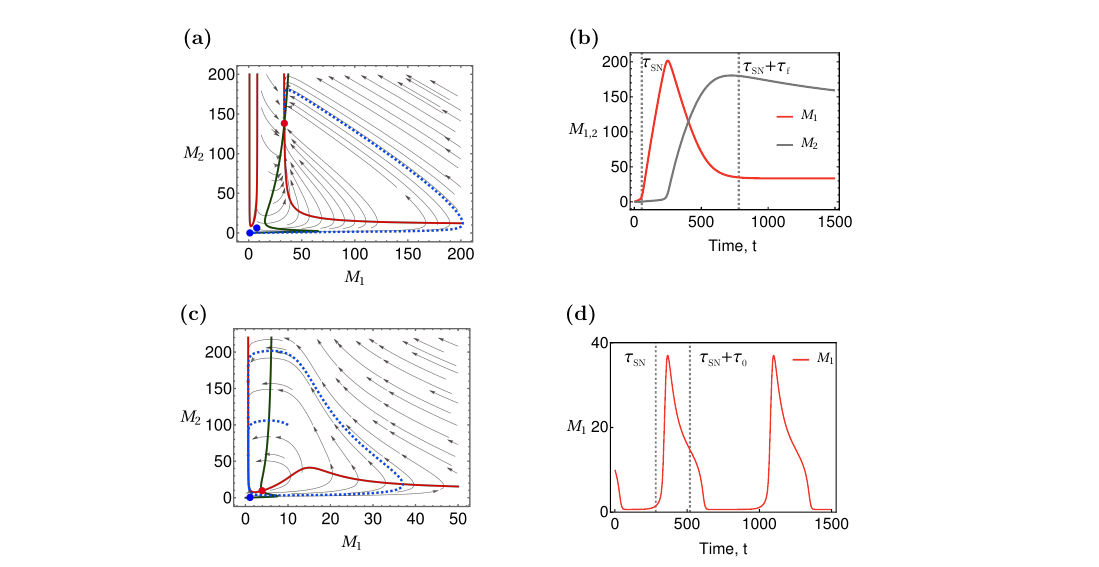}
\caption{{\bf Cisternae formation time for two cisternae.} (a) Trajectory in ($M_1,M_2$) plane for fixed point solutions is shown (blue dotted line). The real parts of the saddle-node ghosts are shown (blue dots), along with the asymptotic fixed point (red dot). Parameter values, $v=2$, $C_{11}= 100$, $C_{22} = 20$, $C_{12} = 20$, $d_{12} = 2$, $d_{21}=0$, $d_{1}=1$, $d_{2}=1.5$, $a_{1}=a_{2}=0.05$. (b) The saddle node delays ($\tau_{SN}$, left black dotted line) and the characteristic time ($\tau_f$) at the asymptotic fixed point in (a) contribute to the cisternae formation time ($\tau_{SN}$ + $\tau_f$, right black dotted line) for (a), see Eq.\,\eqref{eq:FP_taum}.  (c) The trajectory for limit cycle solutions in ($M_1,M_2$) plane is shown (blue dotted line). The real part of the saddle-node ghost is shown (blue dot), along with the central fixed point (red dot). (d) The saddle node delay ($\tau_{SN}$,  left black dotted line) and the linearized period $\tau_0$,
computed at the central fixed point in (c) contribute to the cisternae formation time ($\tau_{SN}$ + $\tau_0$, right black dotted line), see Eq.\,\eqref{eq:LC_taum}. For clarity, only $M_1$ is plotted. Parameter values, $v=2.3$, $C_{11}= 100$, $C_{22} = 20$, $C_{12} = 20$, $d_{12} = 3$, $d_{21}=0$, $d_{1}=1$, $d_{2}=0.8$, $a_{1}=a_{2}=0.05$.}
\label{fig:SNIC_2cist}
\end{figure*}

For the  dynamical system corresponding to two cisternae,
\begin{equation}
\dot{M}_1 = f(M_1,M_2) \hspace{2cm} \dot{M}_2 = g(M_1,M_2)\,,
\label{eq:2cist_ds}
\end{equation}
the time scales can be computed from the eigenvalues of the Jacobian $\nabla(f,g)$ evaluated at the fixed points of the system, given by $f=0,g=0$. As before, there are two time scales (a) $\tau_f$, the characteristic time scale at the final steady state $(M_1^*,M_2^*)$ (b) time delays due to saddle-node ghosts.
Now, up to a linear approximation, the solution of the above dynamical system near a fixed point $(M_1^*,M_2^*)$ is given by
\begin{eqnarray}
\mathbf{M}(t) = A_0 \exp(\lambda_1 t) \; \mathbf{e_1} + B_0 \exp(\lambda_2 t) \; \mathbf{e_2}\,,
\end{eqnarray}
where $\mathbf{M}(t) = (M_1(t),M_2(t))$, $A_0,B_0$ are set by initial conditions, and growth is along the two eigen directions $\mathbf{e}_1,\mathbf{e}_2$, with rates $\lambda_1, \lambda_2$.  Near saddle-node bifurcation, $(A_0,B_0,\lambda_1,\lambda_2)$ are complex-valued, and with basic algebraic manipulations, it can be seen 
that the maximum of imaginary parts of $\lambda_1,\lambda_2$ determines the time scale of escape from the saddle node ghost (it provides the fastest escape route and hence determines the delay time scale). As in the single cisterna case, we can compute the eigenvalues $E_i$'s, at the fixed points (complex roots) $R_i$'s and the bifurcation parameters $x_i^k = \sqrt{(Re[R_i^k])^2/\left((Re[R_i^k])^2 + (Im[R_i^k])^2\right)}; \; k= (1,2)$. These can be used to compute the saddle-node delays,
\begin{eqnarray*}
\tau_{SN}^i(E_i^1,E_i^2,x_i^1,x_i^2) =  \min_k \tau_{SN}^k(E_i^k,x_i^k) \hspace{2cm} \tau_{SN}^k(E_i^k,x_i^k)  = \norm{1/E_i^k} \log(1-x_i^k) \,.
\end{eqnarray*}
Hence, the total time of formation for the two cisternae can be approximated by
\begin{eqnarray}
\tau_{denovo} &\approx &  \tau_{f}(M_1^*,M_2^*)   +   \sum_i  \tau_{SN}^i(E_i^1,E_i^2,x_i^1,x_i^2) \nonumber \\ && i \in \textrm{SN bifurcations encompassed by the trajectory,}
\label{eq:FP_taum}
\end{eqnarray}
where $\tau_f$ is the characteristic time scale computed at the steady state $M^*$ (eigenvalue with the least negative real part at $M^*$), while the timescale corresponding to $i^{th}$ SN-ghost delay is determined by the eigenvalue $E_i = (E_i^1,E_i^2)$ and the bifurcation parameter $x_i=(x_i^1,x_i^2)$. \\

\noindent \textbf{Time scale for the limit cycle:} the above analysis can be extended to compute the time scale for a limit cycle by summing over all the saddle-node ghost delays encompassed by the limit cycle added to the linearized period $\tau_0$, derived from the imaginary part of the eigenvalues at the central fixed point (see Fig.\,\ref{fig:SNIC_2cist}(b)).
\begin{eqnarray}
\tau_{LC} &\approx & \tau_0 +   \sum_i \tau_{SN}^i(E_i^1,E_i^2,x_i^1,x_i^2) \nonumber \\  && i \in \textrm{SN bifurcations encompassed by the trajectory.}
\label{eq:LC_taum}
\end{eqnarray}

\section{Supplementary Movies}

The following video files are associated with this manuscript.

\begin{itemize}
    
    \item \textbf{Video 1 (\texttt{FPsolution.mp4}):} 
    Time evolution of $2$-cisternae system ($M_1(t),M_2(t)$), in the parameter regime where fixed point solutions obtain. It also shows $J_{in}$ -- influx from ER,  $J_{exit}$ -- exit flux from second cisterna,  and $J_{int}$ -- anterograde intercisternal flux. For visualization purposes, all the fluxes in the video have been scaled by a factor of 50.
Parameter values, $v = 2$, $d_1 = 1$, $d_2 = 1.6$, $d_{12} = 2.5$, $d_{21} = 0.2$, $C_{11} = C_{21} = 100$, $C_{22} = 20$, $C_{12} = 20$, $a_1 = a_2 = 0.05$.
See Eqs.\,\eqref{eq:retrograde1},\eqref{eq:retrograde2}  in the main text.
    
        \item \textbf{Video 2 (\texttt{OutphaseLC.mp4}):} 
    Time evolution of $2$-cisternae system ($M_1(t),M_2(t)$), in the parameter regime where out-of-phase limit cycle solutions obtain. It also shows $J_{in}$ -- influx from ER,  $J_{exit}$ -- exit flux from second cisterna,  and $J_{int}$ -- anterograde intercisternal flux. For visualization purposes, all the fluxes in the video have been scaled by a factor of 50.
Parameter values, $v = 1.7$, $d_1 = 1$, $d_2 = 1.5$, $d_{12} = 2.5$, $d_{21} =0.2$, $C_{11} = C_{21} = 100$, $C_{22} = 20$, $C_{12} = 20$, $a_1 = a_2 = 0.05$.
See Eqs.\,\eqref{eq:retrograde1},\eqref{eq:retrograde2} in the main text.

    \item \textbf{Video 3 (\texttt{InphaseLC.mp4}):} 
    Time evolution of $2$-cisternae system ($M_1(t),M_2(t)$), in the parameter regime where in-phase limit cycle solutions obtain. It also shows $J_{in}$ -- influx from ER,  $J_{exit}$ -- exit flux from second cisterna,  and $J_{int}$ -- anterograde intercisternal flux. For visualization purposes, all the fluxes in the video have been scaled by a factor of 50.
Parameter values, $v = 1.2$, $d_1 = 1$, $d_2 = 0.8$, $d_{12} = 3$, $d_{21} =0$, $C_{11} = C_{21} = 100$, $C_{22} = 20$, $C_{12} = 20$, $a_1 = a_2 = 0.05$.
See Eqs.\,\eqref{eq:retrograde1},\eqref{eq:retrograde2} in the main text.
\end{itemize}


\medskip

\end{document}